\documentclass[11pt]{article}
\usepackage[margin=1in]{geometry}
\usepackage{booktabs}
\usepackage{array}
\usepackage{multirow}
\usepackage{multicol}

\usepackage[
    backend=biber,
    style=numeric-comp,
    sortcites,
    sorting=none,
    giveninits=true,
    natbib,
    hyperref=true,
    maxbibnames=99,
    doi=false,isbn=false,url=false,eprint=false
]{biblatex}

\usepackage[hidelinks]{hyperref}

\title{Structured Conformal Inference for Matrix Completion with Applications to Group Recommender Systems}

\usepackage{placeins}
\usepackage{graphicx, color, url}
\usepackage{dsfont}
\usepackage{amssymb}
\usepackage{amsfonts}
\usepackage{amsbsy}
\usepackage{setspace}
\usepackage{bm}
\usepackage{textcomp}
\usepackage{booktabs}
\usepackage{rotating}
\usepackage{enumitem}
\usepackage{amsmath,amsthm,nccmath}
\usepackage{mathtools}
\usepackage{bbm}
\usepackage{soul}
\usepackage{cleveref}

\newcommand{\mat}[1]{\bm{#1}}
\DeclarePairedDelimiter\norm{\lVert}{\rVert}
\DeclarePairedDelimiter\abs{\lvert}{\rvert}
\DeclarePairedDelimiter\paren{(}{)}
\DeclarePairedDelimiter\brac{[}{]}
\DeclarePairedDelimiter\cbrac{\{}{\}}
\newcommand\lcbrac[1]{\left\{ #1 \right\}}

\def\tt{\ensuremath\text}

\def\cD{\ensuremath\mathcal{D}}

\DeclareMathOperator*{\argmax}{arg\,max}

\makeatletter
\newcommand*\bigcdot{\mathpalette\bigcdot@{.5}}
\newcommand*\bigcdot@[2]{\mathbin{\vcenter{\hbox{\scalebox{#2}{$\m@th#1\bullet$}}}}}
\makeatother

\newcommand{\I}[1]{\mathbb{I}\left[#1\right]}
\renewcommand{\P}[1]{\mathbb{P}\left[#1\right]}
\newcommand{\E}[1]{\mathbb{E}\left[#1\right]}

\newcommand{\R}{\mathbb{R}}

\newcommand{\EE}[1]{\mathbb{E}\left[{#1}\right]}

\newcommand{\calD}{{\mathcal{D}}}

\newcommand{\calI}{{\mathcal{I}}}

\newcommand{\calS}{{\mathcal{S}}}

\def\hat{\widehat}
\def\tilde{\widetilde}

\usepackage[dvipsnames]{xcolor}
\usepackage{caption}
\usepackage{subcaption}

\newtheorem{theorem}{Theorem}

\newtheorem{proposition}{Proposition}

\newtheorem{lemma}{Lemma}

\usepackage[colorinlistoftodos]{todonotes}

\usepackage{algorithm, algorithmic}

\author{Ziyi Liang\textsuperscript{1}, Tianmin Xie\textsuperscript{2}\textsuperscript{*}, Xin Tong\textsuperscript{2,3}, and Matteo Sesia\textsuperscript{2,4}}

\begin{document}

\maketitle
\let\thefootnote\relax
\footnotetext{\hspace{-0.62cm}\textsuperscript{1}Department of Statistics, University of California Irvine, Irvine, CA, USA.\\
  \textsuperscript{2}Department of Data Sciences and Operations, University of Southern California, Los Angeles, CA, USA.\\
  \textsuperscript{3}Business School, University of Hong Kong, Hong Kong, China. \\
  \textsuperscript{4}Department of Computer Science, University of Southern California, Los Angeles, CA, USA. \\
  \textsuperscript{*}The first two authors contributed equally to this work.}

\maketitle

\begin{abstract}
We develop a conformal inference method to construct a joint confidence region for a given group of missing entries within a sparsely observed matrix, focusing primarily on entries from the same column. Our method is model-agnostic and can be combined with any ``black-box'' matrix completion algorithm to provide reliable uncertainty estimation for group-level recommendations. For example, in the context of movie recommendations, it is useful to quantify the uncertainty in the ratings assigned by all members of a group to the same movie, enabling more informed decision-making when individual preferences may conflict. Unlike existing conformal techniques, which estimate uncertainty for one individual at a time, our method provides stronger group-level guarantees by assembling a structured calibration dataset that mimics the dependencies expected in the test group. To achieve this, we introduce a generalized weighted conformalization framework that addresses the lack of exchangeability arising from structured calibration, introducing several innovations to overcome associated computational challenges. We demonstrate the practicality and effectiveness of our approach through extensive numerical experiments and an analysis of the MovieLens 100K dataset.  

%%% Local Variables:
%%% mode: latex
%%% TeX-master: "main_jasa"
%%% End:

 \end{abstract}
Laplace's method; simultaneous inference.

\section{Introduction} \label{sec:intro}

\subsection{Uncertainty Estimation for Group Recommender Systems}\label{sec:motivation}

Group recommender systems generate product suggestions collectively for groups of users \citep{dara2020survey}, extending traditional algorithms for individual recommendations.
As digital interactions increasingly emphasize shared experiences, these systems are becoming more significant across various domains.
For example, in the context of movies, they can help friends select a film for a movie night \citep{quijano2011happy} and enhance user experiences on {\em Teleparty}—a platform that enables people to watch films together in sync—by suggesting content that aligns with the interests of all participants.
Group recommender systems also have many other applications, including assisting in planning travel itineraries, making dining reservations, and even recommending job candidates.
%Beyond movie recommendations, group recommender systems also have many other applications, including assisting in planning travel itineraries \citep{najafian2020you}, making dining reservations for team-building events \citep{mccarthy2002pocket}, and even recommending job candidates to hiring committees \citep{baskin2009preference}.
As digital collaboration expands, the relevance of group recommender systems is likely to grow.

Uncertainty estimation is a key challenge for group recommender systems \citep{xiao2017fairness}, particularly due to the need to reconcile potentially conflicting preferences.
For instance, consider three friends planning a movie night: Anna enjoys sci-fi, Ben prefers action, and Dan favors historical films.
These differing tastes must be carefully balanced to find a mutually satisfactory choice.
This challenge is further compounded by the fact that individual preferences are typically inferred from past viewing histories rather than explicitly known, introducing inherent uncertainty \citep{coscrato2023estimating}.

Suppose the system predicts all three users would rate ``Star Wars'' 5/5 but has high confidence only in Anna’s and Ben’s ratings, while Dan’s rating is uncertain and could be as low as 3/5 due to limited data on his sci-fi preferences.
Alternatively, it may confidently predict all three would rate ``Titanic'' 4/5.
Without accounting for uncertainty, the system might recommend ``Star Wars'' solely based on higher predicted ratings.
However, incorporating uncertainty enables more informed decisions.
For instance, an uncertainty-aware system might recommend ``Titanic'' under the {\em least misery} principle, which prioritizes the experience of the least satisfied group member \citep{polylens-grs-2001}, or ``Star Wars'' under the opposite {\em most pleasure} principle.
Thus, uncertainty estimation is essential, and since group recommendations must account for the preferences of all group members, it is crucial that the uncertainty estimates are valid for all group members simultaneously.

% \begin{figure}[!htb]
%   \subfloat[Movie night]{%
%     \includegraphics[width=.22\linewidth]{figures/movie_night_star_war_name2.jpg}%
%     \label{subfig:movie-night}%
%   }\hfill
%     \subfloat[Rating matrix]{%
%     \includegraphics[width=.38\linewidth]{figures/movie_rating_table_bw_3.jpg}%
%     \label{subfig:movie-matrix}%
%   }\hfill
%   \subfloat[Movie bundle]{%
%     \includegraphics[width=.18\linewidth]{figures/movie_bundle.jpg}%
%     \label{subfig:movie-bundle}%
%   }
% \caption{An application of group recommender systems is suggesting a movie for a group of users to watch together (a).
% To make informed recommendations, these systems must estimate the unknown preferences of all group members—a statistical problem that can be framed as constructing joint confidence intervals for $K$ missing entries in the same column of a partially observed rating matrix (b), where $K$ represents the group size.
% The method developed in this paper addresses this challenge using a conformal inference approach, which also extends to other applications, such as bundled product recommendations (c).}
%   \label{fig:motivation}
% \end{figure}

\subsection{Preview of Our Contributions} \label{sec:intro-preview}

\subsubsection{Joint Uncertainty Estimation for Matrix Completion}

We develop a conformal inference \citep{vovk2005algorithmic} method for uncertainty-aware matrix completion, which can be applied to make group recommender systems more reliable.
Recommender systems traditionally estimate user preferences via {\em matrix completion}, predicting missing values in a rating matrix where rows represent users and columns denote products \citep{ramlatchan2018survey}.
This makes it possible to find and leverage meaningful patterns in the data, such as low-rank structures (e.g., similar users like similar movies).
However, the complexity of these models makes uncertainty estimation challenging.
First, traditional theoretical analyses are difficult without strong assumptions on the data distribution, which may not hold in practice.
Second, as machine learning models evolve, statistical methods that rely on their internal mechanisms risk rapid obsolescence.
Conformal inference circumvents these issues by providing rigorous, distribution-free uncertainty estimation while treating the model as a ``black box.''

Group recommender systems require uncertainty estimation {\em jointly} for several users in a group.
While existing methods can provide individual-level uncertainty estimates for matrix completion, their extension to the group setting is nontrivial due to complex dependencies across users.
We address this challenge by developing a method for constructing {\em joint confidence intervals} for $K$ missing entries {\em within the same column} of a partially observed matrix, where $K$ is small relative to the number of rows and columns.
In particular, this method can simultaneously estimates the unobserved ratings given by a given group of $K$ users to the same product, ensuring all true ratings are covered with high probability.

Beyond group recommendations, our method can also be applied to address different matrix completion tasks.
For instance, by transposing the rows and columns of the matrix, the same procedure seamlessly extends to provide uncertainty estimates for {\em bundle recommendations} \citep{zhu2014bundle}, where a system recommends multiple products as a package to a single user.
More broadly, our approach can be adapted to joint inference problems involving different 
subsets of missing entries in a partially observed matrix.
However, for concreteness, we focus on predicting $K$ missing entries within the same column.

\subsubsection{Conformal Inference for Structured Joint Predictions}

Extending the conformal inference framework, we construct confidence (or prediction) intervals using point estimates from a ``black-box'' matrix completion model.
The width of these intervals is tuned using a separate {\em calibration} dataset—a disjoint subset of observed ratings not used for training the matrix completion model.
Since the calibration ratings are observed, we can empirically evaluate the coverage of any confidence interval as a function of the calibration parameter and select the most liberal value (yielding the smallest possible intervals) that still ensures the desired coverage level.
This approach is supported by theoretical guarantees stating that the desired coverage holds at test time, a result typically established in conformal inference via the assumption of {\em exchangeability} between the test and calibration data.
However, in our case these guarantees are less straightforward.

Departing from the standard conformal inference approach, our target of inference is not a scalar quantity (e.g., the rating assigned by a single user) but a vector representing the ratings assigned by all $K$ members of a group.  
These $K$ estimation tasks are inherently dependent, as they share the same data and matrix completion model.  
Moreover, group formation is often not completely random: for example, users with similar tastes may be more likely to belong to the same group, and the matrix completion model may systematically overestimate or underestimate their ratings for a given product.  
As we will see empirically, accounting for these complex dependencies is crucial for obtaining joint confidence intervals that are both valid and informative (i.e., not overly conservative).  
To address this, we calibrate joint confidence intervals using a {\em structured calibration dataset} composed of {\em synthetic groups} of user-product observations.  
These groups are constructed by sampling without replacement from the observed matrix entries to approximate the dependencies expected at test time.  
While this approach is conceptually intuitive, translating it into a concrete method requires overcoming significant challenges.  

\subsubsection{Overcoming Non-Exchangeability via Weighted Conformal Inference }

The structured calibration dataset used by our method does not satisfy the typical exchangeability assumption with the test data, a key requirement in standard conformal inference methods.
This lack of exchangeability seems unavoidable in our context, as calibration can only rely on groups sampled from the observed portion of the ratings matrix, whereas at test time, we seek to estimate missing ratings in the unobserved portion, leading to a distribution shift.
To address this challenge, we extend the {\em weighted exchangeability} framework of \citet{tibshirani-covariate-shift-2019}, originally developed to handle a different {\em covariate shift} problem, as reviewed in Appendix~\ref{app:cov-shift}.
Similar to \citet{tibshirani-covariate-shift-2019}, our solution involves defining and computing suitable weights for each group in our structured calibration set.
Since computing these weights exactly is computationally expensive, we approximate them by combining the {\em Gumbel-max trick} \citep{gumbel1954statistical} with an extension of the classical {\em Laplace method}.
The Gumbel-max trick transforms an intractable summation in the definition of the weights into a more manageable, though still analytically complex, integral, while the Laplace method provides an accurate approximation.
We prove the asymptotic consistency of this approach and demonstrate its effectiveness empirically, showing that it yields valid and informative confidence intervals.

\subsection{Related Works} \label{sec:related}

Despite its deep historical roots in the statistics literature \citep{scheffe1953method}, joint inference for multiple unknowns remains relatively underdeveloped in conformal inference \citep{vovk2005algorithmic,lei2018distribution}.
A few recent works have introduced conformal prediction methods for multi-dimensional responses \citep{diquigiovanni2022conformal, Xu2024ConformalPF, park2024semiparametricconformalprediction}, but these settings differ from ours as they assume observed data with multi-dimensional responses that are exchangeable with the test data.
In the context of matrix completion, \citet{gui2023conformalized} and \citet{shao2023distribution} recently proposed conformal inference methods that compute {\em individual} confidence intervals for single user/movie pairs.
This paper is more closely aligned with the approach of \citet{gui2023conformalized}, reviewed in Appendix~\ref{app:cmc}, but diverges from their work as we seek {\em simultaneous} confidence regions for {\em structured groups} of missing entries---a more challenging problem.

As demonstrated by our numerical experiments on the MovieLens 100K dataset \citep{movielens100k}, previewed in Figure~\ref{fig:movielens-preview} and detailed in Section~\ref{sec:experiments-data}, obtaining informative joint confidence intervals for groups of size $K$ when $K > 1$ is nontrivial.
For instance, applying a {\em Bonferroni} adjustment of order $1/K$ to individual-level conformal inferences \citep{gui2023conformalized} can easily achieve simultaneous validity but typically results in unnecessarily wide intervals, as it fails to account for the complex dependencies among the $K$ tasks based on the same data and model.
In contrast, our method automatically adapts to these dependencies, achieving the desired coverage guarantees with significantly narrower and more informative confidence intervals, particularly for larger values of $K$.

\begin{figure}[!htb]
    \centering
    \includegraphics[width=0.8\linewidth]{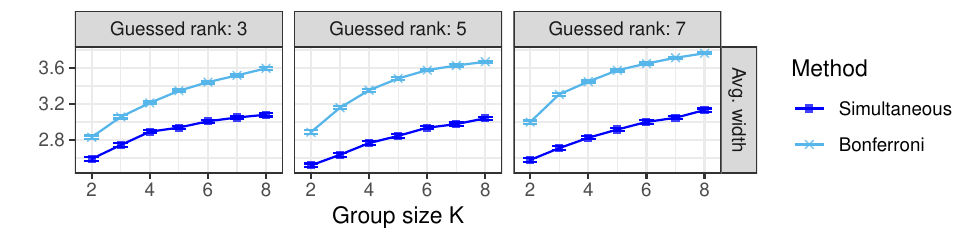}
    \caption{Average width of conformal confidence intervals for groups of $K$ missing entries within the same column of a partially observed ratings matrix from the MovieLens 100K dataset, as a function of $K$.
The results in different columns correspond to a low-rank matrix completion algorithm with varying hypothesized ranks.
Our {\em simultaneous} method ensures 90\% coverage of all true ratings within each group while avoiding unnecessarily wide intervals.
For comparison, we include naive joint intervals obtained via a {\em Bonferroni} correction (of order $1/K$) applied to individual conformal confidence intervals, which is too conservative due to its inability to account for dependencies among the $K$ tasks. }
    \label{fig:movielens-preview}
\end{figure}

More broadly, there exist conformal methods that can handle non-exchangeable data in different contexts \citep{tibshirani-covariate-shift-2019,barber_beyond_2022,dobriban2023symmpi}.
Our solution builds on the framework of \citet{tibshirani-covariate-shift-2019}, originally focused on covariate shift, adapting it to address the unique challenges of our problem.

While we build on conformal inference to avoid parametric assumptions and algorithmic simplifications, there exist alternative approaches for uncertainty estimation in matrix completion.
For instance, \citet{chen_inference_2019}, \citet{xia_statistical_2021}, and \citet{farias_uncertainty_2022} developed asymptotic confidence intervals for missing entries under parametric assumptions, while \citet{yuchi2022bayesian} addressed a similar problem using Bayesian techniques.
Future work could explore integrating these approaches within the ``black-box'' model utilized by our method to obtain even more informative conformal inferences.

\subsection{Outline of the Paper}

This paper is organized as follows: Section~\ref{sec:preliminaries} states our assumptions and goals,
Section~\ref{sec:methodology} presents our method, Section~\ref{sec:implementation} delves into essential implementation aspects,
Section~\ref{sec:empirical} describes empirical demonstrations, and Section~\ref{sec:conclusion} concludes with a discussion.
The appendices are in the Supplementary Material.
Appendix~\ref{app:baselines-background} summarizes some technical background.
Appendix~\ref{app:implementation} explains additional implementation details, Appendix~\ref{app:numerical-results} summarizes further empirical results, and
Appendix~\ref{app:proofs} contains all mathematical proofs.

\section{Setup and Problem Statement}  \label{sec:preliminaries}

Let $\mat{M} \in \mathbb{R}^{n_r \times n_c}$ be a {\em fixed} matrix with $n_r$ rows and $n_c$ columns.
Each entry $M_{r,c}$ is the rating given by user $r \in [n_r] := \{1,\ldots,n_r\}$ to product $c \in [n_c] := \{1,\ldots,n_c\}$.
While it may be common to assume the matrix $\mat{M}$ has a specific structure, such as low-rank, possibly with independent random noise, we treat $\mat{M}$ as deterministic, modeling instead the randomness in the data observation (or missingness) process.

Consider a fixed number of observations, $n_{\mathrm{obs}}$, and let $\mat{M}_{\mat{X}_{\mathrm{obs}}}$ denote the observed portion of $\mat{M}$, indexed by $\mat{X}_{\mathrm{obs}} = (X_1, \ldots, X_{n_{\mathrm{obs}}})$ with each $X_i \in [n_r] \times [n_c]$.
We assume that $\mat{X}_{\mathrm{obs}}$ is randomly sampled {\em without replacement} from $[n_r] \times [n_c]$.
For increased flexibility, we allow this sampling process to be {\em weighted} according to parameters $\mat{w} = (w_{r,c})_{(r,c) \in [n_r] \times [n_c]}$, which we assume to be known for now. See Appendix~\ref{app:missingness-estimation} for details on how these weights may be estimated in practice.
In a compact notation, we write our sampling model as:
\begin{align}\label{model:sampling_wo_replacement}
\mat{X}_{\mathrm{obs}} & \sim \Psi(n_{\mathrm{obs}}, [n_r] \times [n_c], \mat{w}).
\end{align}
In general, $\Psi(m, \mathcal{X}, \mat{w})$ denotes {\em weighted} random sampling without replacement of $m \geq 1$ distinct elements (``balls") from a finite dictionary (``urn") denoted as $\mathcal{X}$, with $|\mathcal{X}| \geq m$.
More precisely, $\mat{X} \sim \Psi(m, \mathcal{X}, \mat{w})$ corresponds to saying that
$\P{\mat{X} = \bm{x}} = (w_{x_1}) / (\sum_{x \in \mathcal{X}} w_x) \cdot (w_{x_2}) / (\sum_{x \in \mathcal{X}} w_x - w_{x_1}) \cdot \ldots \cdot (w_{x_{m}}) / (\sum_{x \in \mathcal{X}} w_x - \sum_{i=1}^{m-1} w_{x_i})$ for any $\mat{x} \in \mathcal{X}^{m}$.
% \begin{align*}
%     \P{\mat{X} = \bm{x}}
%     & = \frac{w_{x_1}}{\sum_{x \in \mathcal{X}} w_x} \cdot \frac{w_{x_2}}{\sum_{x \in \mathcal{X}} w_x - w_{x_1}} \cdot \ldots \cdot \frac{w_{x_{m}}}{\sum_{x \in \mathcal{X}} w_x - \sum_{i=1}^{m-1} w_{x_i}}.
% \end{align*}
Therefore, the weights $\mat{w}$ do not need to be normalized.
This model is a special case of the multivariate Wallenius' noncentral hypergeometric distribution \citep{wallenius1963biased}, and it reduces to random sampling without replacement if all $w_{r,c}$ are the same.

Let $\cD_{\mathrm{obs}}  := \{X_1, \ldots, X_{n_{\mathrm{obs}}}\} \subset [n_r] \times [n_c]$ be the unordered set of observed indices, with complement $\cD_{\mathrm{miss}} \coloneqq [n_r] \times [n_c] \setminus \cD_{\mathrm{obs}}$.
Our goal is to construct a joint confidence region for $K \geq 1$ unobserved entries from the same column, indexed by $\mat{X}^* = (X^*_{1}, X^*_{2}, \ldots, X^*_{K})$, where each $X^*_{k} \in \cD_{\mathrm{miss}}$ and $K$ is a fixed parameter, much smaller than both $n_r$ and $n_c$.
The indices $\mat{X}^*$ are assumed to be sampled without replacement from $\cD_{\mathrm{miss}}$, subject to the constraint that they belong to the same column.
The sampling of $\mat{X}^*$ from $\cD_{\mathrm{miss}}$ is guided by distinct weights $\mat{w}^* = (w^*_{r,c})_{(r,c) \in [n_r] \times [n_c]}$, also assumed to be known for now.
We will discuss later the implications of the choice of $\mat{w}^*$.

% To ensure that the sampling model for $\mat{X}^*$ is well-defined, we exclude columns with fewer than $K$ missing entries.
% Hence, we define a {\em pruned} set of missing indices $\bar{\cD}_{\mathrm{miss}} \subseteq \cD_{\mathrm{miss}}$ as
% $\bar{\cD}_{\mathrm{miss}} = \cbrac{(r,c) \in \cD_{\mathrm{miss}} : n^c_{\mathrm{miss}} \geq K}$,
% where $n^c_{\mathrm{miss}} = \abs{\cbrac{(r',c') \in \cD_{\mathrm{miss}} :  c' = c}}$ is the number of missing entries in column $c$.
Then, we assume $\mat{X}^*$ is sampled as
\begin{align}\label{model:weighted_test_generation}
\begin{split}
\mat{X}^* \mid \cD_{\mathrm{obs}}, \cD_{\mathrm{miss}} & \sim \Psi^{\text{col}}(K, \cD_{\mathrm{miss}}, \mat{w}^*),
\end{split}
\end{align}
where $\Psi^{\text{col}}$ is a constrained version of $\Psi$ that samples indices belonging to the same column.
This distribution is equivalent to the following sequential sampling procedure:
\begin{align}\label{model:weighted_test_generation_sequential}
  \begin{split}
    X^*_{1} \mid \cD_{\mathrm{obs}}, \cD_{\mathrm{miss}} & \sim \Psi(1, \cD_{\mathrm{miss}}, \mat{w}^*) \\
    X^*_{2} \mid \cD_{\mathrm{obs}}, \cD_{\mathrm{miss}}, X^*_{1} & \sim \Psi(1, \cD_{\mathrm{miss}} \setminus \{X^*_{1}\}, \tilde{\mat{w}}^*), \\
    & \; \; \vdots \\
    X^*_{K} \mid \cD_{\mathrm{obs}}, \cD_{\mathrm{miss}}, X^*_{1}, \ldots, X^*_{K-1} & \sim \Psi(1, \cD_{\mathrm{miss}} \setminus \{X^*_{1}, \ldots, X^*_{K-1}\}, \tilde{\mat{w}}^*),
  \end{split}
\end{align}
where the weights $\tilde{\mat{w}}^*$ are given by $\tilde{w}^*_{r,c} = w^*_{r,c} \I{ c = X^*_{1,2}}$ and $X^*_{1,2}$ is the column of $X^*_{1}$.

To ensure \eqref{model:weighted_test_generation} and \eqref{model:weighted_test_generation_sequential} are well-defined, we assume for simplicity that the number of missing entries in all columns is much larger than $K$; i.e., $\abs{\cbrac{(r',c') \in \cD_{\mathrm{miss}}\mid c'=c}} \gg K$, for all $c \in [n_c]$. This is almost always the case when $\mat{M}$ is sparsely observed. For instance, this is clearly the case in the MovieLens dataset, where approximately $94\%$ of ratings are missing. However, in rare cases where some columns have fewer than
$K$ missing entries, these columns should be omitted from the sampling process, as explained in Appendix~\ref{app:corner-cases}.

This premise allows us to state our goal formally.
For a given coverage level $\alpha \in (0,1)$, we seek a {\em joint confidence region}, denoted as $\hat{\mat{C}}_{}(\mat{X}^*; \mat{M}_{\mat{X}_{\mathrm{obs}}}, \alpha) \subseteq \mathbb{R}^K$, for the $K$ missing matrix entries indexed by $\mat{X}^*$.
This confidence region should be informative (i.e., not too wide) and guarantee finite-sample simultaneous coverage, in the sense that
\begin{align} \label{eq:K-coverage}
  \P{ M_{r,c} \in \hat{C}_{r,c}(\mat{X}^*; \mat{M}_{\mat{X}_{\mathrm{obs}}}, \alpha), \; \forall (r,c) \in  \{X^*_{1},\ldots, X^*_{K} \} } \geq 1-\alpha.
\end{align}
Note that the probability in~\eqref{eq:K-coverage} is taken with respect to both $\mat{X}_{\mathrm{obs}}$, sampled according to~\eqref{model:sampling_wo_replacement}, and $\mat{X}^*$, sampled according to~\eqref{model:weighted_test_generation_sequential}, while $\mat{M}$ is fixed.

The sampling weights $\mat{w}$ in~\eqref{model:sampling_wo_replacement} and the test weights $\mat{w}^*$ in~\eqref{model:weighted_test_generation} serve distinct purposes.
The weights $\mat{w}$ model situations where the matrix is not observed uniformly at random. For instance, in collaborative filtering, certain users may be more active, and some movies may receive more ratings. Such patterns can be captured using heterogeneous weights in \eqref{model:sampling_wo_replacement}.
Non-uniform missingness patterns are often evident from the observed data, allowing for empirical estimation of $\mat{w}$, as explained in Appendix~\ref{app:missingness-estimation}.

In contrast, the role of $\mat{w}^*$ is to shape the interpretation of~\eqref{eq:K-coverage}.
If all $w^*_{r,c} = 1$, Equation~\eqref{eq:K-coverage} offers only {\em marginal} coverage, for $\mat{X}^*$ drawn {\em uniformly} from the missing portion of the matrix.
If $w^*_{r,c} = 1$ only when $c \in \mathcal{A}$, for the subset $\mathcal{A} \subset [n_c]$ corresponding to ``action movies'', Equation~\eqref{eq:K-coverage} ensures coverage specifically for action movies, offering a stronger {\em conditional} guarantee.
Thus, $\mat{w}^*$ generally allows interpolation between marginal and conditional perspectives~\citep{guan2023localized}.
However, there is a trade-off: more concentrated weights $\mat{w}^*$ yield stronger guarantees but result in wider (less informative) confidence regions \citep{foygel2021limits}, which is why our flexibility in the choice of $\mat{w}^*$ is desirable.

\section{Methods} \label{sec:methodology}

This section describes the key components of our method, which we call Structured Conformalized Matrix Completion (SCMC).
Section~\ref{sec:method-outline} gives a high-level overview of SCMC.
Section~\ref{sec:calibration-groups} details the construction of the structured calibration set.
Section~\ref{sec:quantile-lemma} recalls a quantile inflation lemma underlying our simultaneous coverage results.
Section~\ref{sec:weights} characterizes the conformalization weights needed to apply our quantile inflation lemma in the context of SCMC.
Section~\ref{sec:bounds} establishes lower and upper simultaneous coverage bounds.
Important computational shortcuts are presented in Section~\ref{sec:implementation}.

\subsection{Method Outline} \label{sec:method-outline}

Having observed the matrix entries indexed by $\cD_{\mathrm{obs}} \subset [n_r] \times [n_c]$, SCMC partitions $\cD_{\mathrm{obs}}$ into two disjoint subsets: a {\em training} set $\cD_{\mathrm{train}}$ and a {\em calibration} set $\cD_{\mathrm{cal}}$, so that $\cD_{\mathrm{obs}} = \cD_{\mathrm{train}} \cup \cD_{\mathrm{cal}}$. Unlike the typical random partitioning in (split) conformal inference, we employ a more sophisticated method to form $\cD_{\mathrm{train}}$ and $\cD_{\mathrm{cal}}$. This approach, detailed in Section~\ref{sec:calibration-groups}, yields a calibration set approximately mirroring the structure of $\mat{X}^*$.

The data in $\cD_{\mathrm{train}}$ are used to train a matrix completion model, producing a point estimate $\hat{\mat{M}}(\cD_{\mathrm{train}})$ of the full matrix $\mat{M}$.
{\em Any} algorithm can be applied for this purpose.
For example, if $\mat{M}$ is suspected to have a low-rank structure, one may employ nuclear norm minimization \citep{candes_exact_2008}, computing
$\hat{\mat{M}}(\cD_{\mathrm{train}}) = \underset{\mat{Z} \in \mathbb{R}^{n_r \times n_c}}{\arg\min} \norm{\mat{Z}}_*$ subject to $\mathcal{P}_{\cD_{\mathrm{train}}}(\mat{Z}) = \mathcal{P}_{\cD_{\mathrm{train}}}(\mat{M})$,
where $\norm{\cdot}_*$ denotes the nuclear norm and $\mathcal{P}_{\cD_{\mathrm{train}}}(\mat{M})$ is the orthogonal projection of $\mat{M}$ onto the subspace of matrices that vanish outside $\cD_{\mathrm{train}}$.
Beyond convex optimization, our method can be combined with any matrix completion algorithm, including deep learning approaches \citep{fan2017deep}, or methods incorporating side information~\citep{natarajan2014inductive}.
While SCMC tends to produce more informative confidence regions if $\hat{\mat{M}}(\cD_{\mathrm{train}})$ estimates $\mat{M}$ more accurately, its coverage guarantee requires no assumptions on how $\hat{\mat{M}}$ is derived from $\cD_{\mathrm{train}}$.

Our method translates any {\em black-box} estimate $\hat{\mat{M}}(\cD_{\mathrm{train}})$ into confidence regions for the missing entries as follows.
Let $\mathcal{C}$ be a pre-specified set-valued function, termed {\em prediction rule}, that takes as input $\hat{\mat{M}}$, a list of $K$ target indices $\mat{x}^*=(x^*_1,\ldots,x^*_K)$, and a parameter $\tau \in [0,1]$, and outputs $\mathcal{C}(\mat{x}^*, \tau, \hat{\mat{M}}) \subseteq \mathbb{R}^K$.
(We will often make the dependence of $\hat{\mat{M}}(\cD_{\mathrm{train}})$ on $\cD_{\mathrm{train}}$ implicit.)
Our method is flexible in the choice of the prediction rule, but we generally require that this function be monotone increasing in $\tau$, in the sense that
 \begin{align} \label{eq:pred-rule-monotone}
  \mathcal{C}(\mat{x}^*, \tau_1, \hat{\mat{M}}) \subseteq \mathcal{C}(\mat{x}^*, \tau_2, \hat{\mat{M}}), \quad \text{ almost surely if } \tau_1 < \tau_2,
\end{align}
and satisfies the following boundary conditions almost surely:
\begin{align} \label{eq:pred-rule-boundary}
  & \mathcal{C}(\mat{x}^*, 0, \hat{\mat{M}}) =   \left\{ \left( \hat{\mat{M}}_{x^*_1} , \ldots, \hat{\mat{M}}_{x^*_K} \right) \right\},
  & \lim_{\tau \to 1^-} \mathcal{C}(\mat{x}^*, \tau, \hat{\mat{M}}) = \mathbb{R}^K.
\end{align}
Intuitively, $\tau=0$ corresponds to placing absolute confidence in the accuracy of $\hat{\mat{M}}$, while approaching $\tau=1$ suggests that the point estimate carries no information about $\mat{M}$.
For example, one may choose
$\mathcal{C}(\mat{x}^*, \tau, \hat{\mat{M}}) =  \left( \hat{\mat{M}}_{x^*_1} \pm \tau/(1-\tau), \ldots, \hat{\mat{M}}_{x^*_K} \pm \tau/(1-\tau) \right)$,
which produces regions in the shape of a hyper-cube.
We will use this approach in our numerical experiments due to its ease of interpretation, though it is not the only option available.
See Appendix~\ref{sec:scores} for further details and examples of alternative prediction rules.

The purpose of the observations indexed by $\cD_{\mathrm{cal}}$, which are not used to train $\hat{\mat{M}}$, is to find the smallest possible $\tau$ needed to achieve~\eqref{eq:K-coverage}.
As detailed in Section~\ref{sec:calibration-groups}, SCMC carefully constructs $n$ {\em calibration groups}
$\{ \mat{X}^{\mathrm{cal}}_{1}, \dots, \mat{X}^{\mathrm{cal}}_{n} \}$ from $\cD_{\mathrm{cal}}$, where each $\mat{X}^{\mathrm{cal}}_i$ consists of $K$ observed matrix entries within the same column; i.e., $\mat{X}^{\mathrm{cal}}_i = (X^{\mathrm{cal}}_{i,1}, \ldots, X^{\mathrm{cal}}_{i,K})$.
% $\mat{X}^{\mathrm{cal}}_i = (X^{\mathrm{cal}}_{i,1}, \ldots, X^{\mathrm{cal}}_{i,K})$, for \textcolor{red}{each} $i \in [n]$, each consisting of $K$ observed matrix entries within the same column.
As explained in the next section, $n$ can be fixed arbitrarily, although it should be small compared to the total number of observed entries and typically at least greater than 100 to avoid excessive variability in the results \citep{vovk2012conditional}.
Intuitively, these calibration groups are constructed in such a way as to (approximately) simulate the structure of $\mat{X}^*$.

For each calibration group, we compute a {\em conformity score} $S_i = S\paren{\mat{X}^{\mathrm{cal}}_i}$, defined as the smallest value of $\tau$ for which the candidate confidence region covers all $K$ entries of $\mat{M}_{\mat{X}^{\mathrm{cal}}_i}$:
\begin{align} \label{eq:scores}
    S_i := \inf \cbrac*{\tau \in \mathbb{R}: \mat{M}_{\mat{X}^{\mathrm{cal}}_i} \in \mathcal{C}(\mat{X}^{\mathrm{cal}}_i, \tau, \hat{\mat{M}}) }.
\end{align}
Then, the calibrated value $\hat{\tau}_{\alpha,K}$ of $\tau$ is obtained by evaluating the following {\em weighted quantile} \citep{tibshirani-covariate-shift-2019} of the empirical distribution of the calibration scores:
\begin{align}\label{eq:weighted-quantile}
  \hat{\tau}_{\alpha,K} = Q\paren[\Big]{1-\alpha; \sum_{i=1}^n p_i\delta_{S_i} + p_{n+1}\delta_{\infty}}.
\end{align}
Above, $Q(1-\alpha;F)$ denotes the $1-\alpha$ quantile of a distribution $F$ on $\mathbb{R} \cup \cbrac{\infty}$; that is, for $S \sim F$,
$Q(\beta; F) = \inf{\cbrac*{s \in \mathbb{R}: \mathbb{P}\paren*{S \leq s} \geq \beta}}$.
The distribution in~\eqref{eq:weighted-quantile} places a point mass $p_i$ on each observed value of $S_i$ and an additional point mass $p_{n+1}$ at $+\infty$.
The expression of the weights $p_i$ and $p_{n+1}$ will be given in Section~\ref{sec:weights}.
These weights generally depend on $\mat{X}^*$ and on all $\mat{X}^{\mathrm{cal}}_i$, although this dependence is kept implicit here for simplicity.

Finally, $\hat{\tau}_{\alpha,K}$ is utilized to construct a joint confidence region
$\hat{\mat{C}}_{}(\mat{X}^*; \mat{M}_{\mat{X}_{\mathrm{obs}}},\alpha) = \mathcal{C}(\mat{X}^*, \hat{\tau}_{\alpha,K}, \hat{\mat{M}})$.
This is proved in Section~\ref{sec:bounds} to achieve~\eqref{eq:K-coverage}, as long as $\mat{X}_{\mathrm{obs}}$ is sampled from~\eqref{model:sampling_wo_replacement} and $\mat{X}^*$ from~\eqref{model:weighted_test_generation_sequential}.
The procedure is outlined by Algorithm~\ref{alg:simultaneous-confidence-region}, while the details are explained in subsequent sections.
Crucially, SCMC can be implemented efficiently and is scalable to large matrices; see Appendix~\ref{app:complexity} for an analysis of its computational costs.

\begin{algorithm}[!htb]
\caption{Simultaneous Conformalized Matrix Completion (SCMC)}
    \label{alg:simultaneous-confidence-region}
    \begin{algorithmic}[1]
        \STATE \textbf{Input}: partially observed matrix $\mat{M}_{\mat{X}_{\mathrm{obs}}}$, with unordered list of observed indices $\cD_{\mathrm{obs}}$;
        \STATE \textcolor{white}{\textbf{Input}:} test group $\mat{X}^*$; nominal coverage level $\alpha \in (0,1)$;
        \STATE \textcolor{white}{\textbf{Input}:} any matrix completion algorithm producing point estimates;
        \STATE \textcolor{white}{\textbf{Input}:} any prediction rule $\mathcal{C}$ satisfying~\eqref{eq:pred-rule-monotone} and~\eqref{eq:pred-rule-boundary};
        \STATE \textcolor{white}{\textbf{Input}:} desired number $n$ of calibration groups.

        \STATE Apply Algorithm~\ref{alg:calibration-group} to obtain $\cD_{\mathrm{train}}$ and the calibration groups $(\mat{X}^{\mathrm{cal}}_{1}, \dots, \mat{X}^{\mathrm{cal}}_{n})$ in $\cD_{\mathrm{cal}}$.
        \STATE
        Compute a point estimate $\hat{\mat{M}}$ by applying the matrix completion algorithm to $\cD_{\mathrm{train}}$.
        \STATE
        Compute the conformity scores $S_i$, for each calibration group $i \in [n]$, by Equation \eqref{eq:scores}.
        %\STATE Compute the weights $p_i$, for all $i \in [n+1]$, with Equation \eqref{eq:gwpm-weights-k}.
        \STATE Compute $\hat{\tau}_{\alpha,K}$ in \eqref{eq:weighted-quantile}, based on the weights $p_i$ given by \eqref{eq:gwpm-weights-k} in Section~\ref{sec:weights}.
        \STATE \textbf{Output}: Joint confidence region $\hat{\mat{C}}_{}(\mat{X}^*; \mat{M}_{\mat{X}_{\mathrm{obs}}},\alpha) = \mathcal{C}(\mat{X}^*, \hat{\tau}_{\alpha,K}, \hat{\mat{M}})$.
    \end{algorithmic}
\end{algorithm}

%We conclude this section by remarking that our method is inspired by the weighted conformal inference framework of \citet{tibshirani-covariate-shift-2019}, but it involves distinct challenges because the groups $\mat{X}^{\mathrm{cal}}_1, \ldots, \mat{X}^{\mathrm{cal}}_n, \mat{X}^*$ are still neither exchangeable \citet{vovk2005algorithmic} nor weighted exchangeable \citet{tibshirani-covariate-shift-2019}.

\subsection{Assembling the Structured Calibration Set}  \label{sec:calibration-groups}

This section explains how to partition $\cD_{\mathrm{obs}}$ into a training set $\cD_{\mathrm{train}}$ and $n$ {\em calibration groups} $\mat{X}^{\mathrm{cal}}_1, \ldots, \mat{X}^{\mathrm{cal}}_n$ that mimic the structure of $\mat{X}^*$.
To begin, note that $n$ cannot exceed $\lfloor n_{\mathrm{obs}}/K \rfloor$ and
$n \leq \sum_{c=1}^{n_c} \left\lfloor n^c_{\mathrm{obs}} / K \right\rfloor \eqqcolon \xi_{\mathrm{obs}}$,
where $n^c_{\mathrm{obs}}$ is the number of observed entries in column $c \in [n_c]$, which is a function of $\mat{X}_{\mathrm{obs}}$ in~\eqref{model:sampling_wo_replacement}.
A practical rule of thumb is to set $n = \min\{1000, \lfloor \xi_{\mathrm{obs}}/2 \rfloor\}$.
In the following, we assume $n$ is parameter (e.g., $n = 1000$) and satisfies $n \leq \xi_{\mathrm{obs}}$; this streamlines the analysis of SCMC without much loss of generality.

% In principle, it would also be possible to set $n$ in a data-independent way so that $n \leq \xi_{\mathrm{obs}}$ holds with high probability, as long as $K$ is not too large compared to $n_{\mathrm{obs}}$ and $n_rn_c$.

For any feasible $n$, Algorithm~\ref{alg:calibration-group} partitions $\cD_{\mathrm{obs}}$ into a training set $\cD_{\mathrm{train}}$ and $n$ calibration groups $\mat{X}^{\mathrm{cal}}_1, \ldots, \mat{X}^{\mathrm{cal}}_n$.
Initially, we set $\cD_{\mathrm{train}} = \emptyset$. We then iterate over each column $c$, adding a random subset of $m^c:=n^c_{\mathrm{obs}} \bmod K$ observations from that column to $\cD_{\mathrm{train}}$, where $n^c_{\mathrm{obs}}$ is the total number of observations in column $c$.
This ensures the remaining number of observations in column $c$ is a multiple of $K$ (possibly zero).
Then, for each $i \in [n]$, $\mat{X}^{\mathrm{cal}}_i$ is obtained by sampling $K$ observations uniformly without replacement from a randomly chosen column.
Finally, all remaining observations are assigned to $\cD_{\mathrm{train}}$.

\begin{algorithm}
\caption{Assembling the structured calibration set for Algorithm~\ref{alg:simultaneous-confidence-region}}
\label{alg:calibration-group}
\begin{algorithmic}[1] % The [1] gives line numbers every line.
\STATE \textbf{Input}: Set $\cD_{\mathrm{obs}}$ of $n_{\mathrm{obs}}$ observed entries; calibration size $n$; group size $K$.

\STATE Initialize an empty set of matrix indices, $\cD_{\mathrm{prune}}=\emptyset$.
\FORALL{columns $c \in [n_c]$}
    \STATE  Define $m^c:=n^c_{\mathrm{obs}} \bmod K$.
    \IF{$m^c \neq 0$}
        \STATE Sample $m^c$ indices  $\paren*{I_1, \dots, I_{m^c}} \sim \Psi(m^c, \cD_{\mathrm{obs}} \cap ([n_r] \times \{c\}), \mat{1})$.
        \STATE Add the entry indices $\cbrac*{I_1, \dots, I_{m^c}}$ to $\cD_{\mathrm{prune}}$.
    \ENDIF
\ENDFOR

\STATE Initialize a set of available observed matrix indices, $\cD_{\mathrm{avail}} = \cD_{\mathrm{obs}}\setminus \cD_{\mathrm{prune}}$.
\STATE Initialize an empty set of matrix index groups, $\mathcal{D}_{\mathrm{cal}}=\emptyset$.
\FOR{$i \in [n]$}
\STATE Sample $\mat{X}^{\mathrm{cal}}_i = (X^{\mathrm{cal}}_{i,1}, \dots,X^{\mathrm{cal}}_{i,K}) \sim \Psi^{\text{col}}(K, \cD_{\mathrm{avail}}, \mat{1})$, with $\Psi^{\text{col}}$ defined as in~\eqref{model:weighted_test_generation}.
\STATE Insert $\mat{X}^{\mathrm{cal}}_i$ in $\mathcal{D}_{\mathrm{cal}}$. Remove $\cbrac{X^{\mathrm{cal}}_{i,1}, \dots, X^{\mathrm{cal}}_{i,K}}$ from $\cD_{\mathrm{avail}}$.
\ENDFOR

\STATE Define: $\cD_{\mathrm{train}} = \cD_{\mathrm{prune}} \cup \mathcal{D}_{\mathrm{avail}}$.
\STATE \textbf{Output}: Set of calibration groups $\cD_{\mathrm{cal}} = \{ \mat{X}^{\mathrm{cal}}_{1}, \dots, \mat{X}^{\mathrm{cal}}_{n} \}$; training set $\cD_{\mathrm{train}} \subset \cD_{\mathrm{obs}}$;
\end{algorithmic}
\end{algorithm}

Algorithm~\ref{alg:calibration-group} is designed to mimic the model for $\mat{X}^*$ in~\eqref{model:weighted_test_generation}, with the key difference that it samples the calibration groups from $\cD_{\mathrm{obs}}$ instead of $\cD_{\mathrm{miss}}$.
This discrepancy is unavoidable, and it results in $\mat{X}^{\mathrm{cal}}_1, \ldots, \mat{X}^{\mathrm{cal}}_n$ being neither exchangeable nor {\em weighted exchangeable} \citep{tibshirani-covariate-shift-2019} with $\mat{X}^*$.
%This is because the calibration groups are all sampled starting from $\cD_{\mathrm{obs}}$, while $\mat{X}^*$ is sampled from the complement set $\cD_{\mathrm{miss}}$.
Consequently, an innovative approach is needed to obtain valid simultaneous confidence regions, as explained in the next section.

\subsection{A General Quantile Inflation Lemma} \label{sec:quantile-lemma}

Consider a conformity score  $S^*$, defined similarly to the scores $S_i$ in~\eqref{eq:scores}:
\begin{align} \label{eq:def-score-test}
S^* := \inf \cbrac*{\tau \in \mathbb{R}: \mat{M}_{\mat{X}^*} \in \mathcal{C}(\mat{X}^*, \tau, \hat{\mat{M}}) }.
\end{align}
In words, $S^*$ is the smallest $\tau$ for which $\mathcal{C}(\mat{X}^*, \tau, \hat{\mat{M}})$ covers all $K$ entries of $\mat{M}_{\mat{X}^*}$. Although $S^*$ is latent, it is a well-defined and useful quantity because it allows us to write the probability that the confidence region output by Algorithm~\ref{alg:simultaneous-confidence-region} simultaneously covers all elements of $\mat{M}_{\mat{X}^*}$ as: $\P{ \mat{M}_{\mat{X}^{*}} \in \mathcal{C}(\mat{X}^*, \hat{\tau}_{\alpha,K}, \hat{\mat{M}})} = \P{ S^* \leq \hat{\tau}_{\alpha,K}}$.
% \begin{align} \label{eq:quantile-inflation-coverage}
%   \P{ \mat{M}_{\mat{X}^{*}} \in \mathcal{C}(\mat{X}^*, \hat{\tau}_{\alpha,K}, \hat{\mat{M}})}
%   & = \P{ S^* \leq \hat{\tau}_{\alpha,K}}.
% \end{align}

%If the scores $S_1, \ldots, S_n, S^*$ were exchangeable, it would be easy to bound \eqref{eq:quantile-inflation-coverage} from below by $1-\alpha$,

To establish that Algorithm~\ref{alg:simultaneous-confidence-region} achieves simultaneous coverage~\eqref{eq:K-coverage}, we need to bound $\P{ S^* \leq \hat{\tau}_{\alpha,K}}$ from below by $1-\alpha$, for a suitable (and practical) choice of the weights $p_i$ and $p_{n+1}$ used to compute $\hat{\tau}_{\alpha,K}$ in~\eqref{eq:gwpm-weights-k}.
This is not straightforward because the scores $S_1, \ldots, S_n, S^*$ are neither exchangeable nor weighted exchangeable, as they respectively depend on $\mat{X}^{\mathrm{cal}}_1, \ldots, \mat{X}^{\mathrm{cal}}_n$ and $\mat{X}^*$.
A solution is provided by the following lemma, which extends a result originally presented by \citet{tibshirani-covariate-shift-2019} to address {\em covariate shift}.
For a discussion on the distinction between covariate shift and our problem, see Appendix~\ref{app:cov-shift}.

\begin{lemma} \label{lem:general_quant}
Let $Z_1, \ldots, Z_{n+1}$ be random variables with joint law $f$.
For any fixed function $s$ and $i \in [n+1]$, define $V_i=s( Z_i, Z_{-i} )$, where $Z_{-i}=\{Z_1, \ldots, Z_{n+1}\} \setminus\{Z_i\}$.
Assume that $V_1,\ldots,V_{n+1}$ are distinct almost surely.
Define also
\begin{equation} \label{eq:general_prob}
p^f_i(Z_1,\ldots,Z_{n+1}) :=
\frac{\sum_{\sigma \in \mathcal{S} : \sigma(n+1)=i} f(Z_{\sigma(1)},\ldots, Z_{\sigma(n+1)})}
{\sum_{\sigma \in \mathcal{S}} f(Z_{\sigma(1)},\ldots, Z_{\sigma(n+1)})}, \quad \forall i \in [n+1],
\end{equation}
where $\mathcal{S}$ is the set of all permutations of $[n+1]$.
Then, for any $\beta\in(0,1)$,
$$
\P{ V_{n+1} \leq Q\bigg(\beta; \, \sum_{i=1}^n p^f_i(Z_1,\ldots,Z_{n+1}) \delta_{V_i} + p^f_{n+1}(Z_1,\ldots,Z_{n+1}) \delta_\infty\bigg) } \geq \beta.
$$
\end{lemma}

% The proof of Lemma~\ref{lem:general_quant} is in Appendix~\ref{app:general-quantile}.
% This result is connected to~\eqref{eq:quantile-inflation-coverage} by letting $Z_1, \ldots, Z_{n+1}$ denote our groups $\mat{X}^{\mathrm{cal}}_1, \ldots, \mat{X}^{\mathrm{cal}}_n, \mat{X}^*$, imagining that $V_1, \ldots, V_n, V_{n+1}$ represent the scores $S_1, \ldots, S_n, S^*$, and recalling the definition of $\hat{\tau}_{\alpha,K}$ in~\eqref{eq:weighted-quantile}.
% Note that, in our case, the function $s$ ignores $Z_{-i}$ but depends implicitly on $\cD_{\mathrm{train}}$ through $\hat{\mat{M}}$; this dependence is safely kept implicit because $\cD_{\mathrm{train}}$ can effectively be treated as fixed for our purposes, as it will become clear later.

Translating Lemma~\ref{lem:general_quant} into a practical method requires evaluating the weights $p^f_i$ in~\eqref{eq:general_prob}, which generally involve an unfeasible sum over an exponential number of permutations.
If the distribution $f$ satisfies a symmetry condition called ``weighted exchangeability'', it was shown by \citet{tibshirani-covariate-shift-2019} that the expression in~\eqref{eq:general_prob} simplifies greatly, but this is not helpful in our case because $(\mat{X}^{\mathrm{cal}}_{1}, \dots, \mat{X}^{\mathrm{cal}}_{n}, \mat{X}^*)$ do not enjoy such a property.
Further, it is unclear how Algorithm~\ref{alg:calibration-group} may be modified to achieve weighted exchangeability.

Fortunately, our groups satisfy a ``leave-one-out exchangeability'' property that still enables an efficient computation of the {\em conformalization} weights in~\eqref{eq:general_prob}. Intuitively, the joint distribution of $\mat{X}^{\mathrm{cal}}_{1}, \dots, \mat{X}^{\mathrm{cal}}_{n}, \mat{X}^*$ is invariant to the reordering of the first $n$ variables.

\begin{proposition} \label{prop:paired-calibration-partial-exch}
Let $\cD_{\mathrm{obs}}$ and $\cD_{\mathrm{miss}}$ be subsets of observed and missing entries, respectively, sampled according to~\eqref{model:sampling_wo_replacement}.
Consider $\mat{X}^*$ sampled according to~\eqref{model:weighted_test_generation} conditional on $\cD_{\mathrm{obs}}$.
Suppose $\mat{X}^{\mathrm{cal}}_{1}, \dots, \mat{X}^{\mathrm{cal}}_{n}$ are the calibration groups output by Algorithm~\ref{alg:calibration-group}, while $\cD_{\mathrm{prune}}$ and $\cD_{\mathrm{train}}$ are the corresponding training and pruned subsets.
Conditional on $\cD_{\mathrm{prune}}$ and $\cD_{\mathrm{train}}$, $(\mat{X}^{\mathrm{cal}}_{1}, \mat{X}^{\mathrm{cal}}_{2}, \ldots, \mat{X}^{\mathrm{cal}}_{n}, \mat{X}^*) \overset{d}{=} (\mat{X}^{\mathrm{cal}}_{\sigma(1)}, \mat{X}^{\mathrm{cal}}_{\sigma(2)}, \ldots, \mat{X}^{\mathrm{cal}}_{\sigma(n)}, \mat{X}^*)$ for any permutation $\sigma$ of $[n]$.
\end{proposition}
%The proof of Proposition~\ref{prop:paired-calibration-partial-exch} is in Appendix~\ref{app:paired-conformalization}.
This is established by deriving the distribution of $(\mat{X}^{\mathrm{cal}}_{1}, \dots, \mat{X}^{\mathrm{cal}}_{n}, \mat{X}^*)$ conditional on $\cD_{\mathrm{prune}}$ and $\cD_{\mathrm{train}}$.
The significance of this result becomes evident in the following specialized version of Lemma~\ref{lem:general_quant}, which suggests a useful expression for the conformalization weights.

\begin{lemma} \label{lem:partial_quant}
Let $Z_1, \ldots, Z_{n+1}$ be leave-one-out exchangeable random variables, so that there exists a permutation-invariant function $g$ such that their joint law $f$ can be factorized as $f(Z_1,\ldots,Z_{n+1}) = g(\{Z_1,\ldots,Z_{n+1}\}) \cdot \bar{h}(\{Z_1,\ldots,Z_n\}, Z_{n+1})$,
for some function $\bar{h}$ taking as first input an unordered set of $n$ elements.
For any fixed function $s: \mathbb{R}^{n+1} \mapsto \mathbb{R}$ and any $i \in [n+1]$, define $V_i=s( Z_i, Z_{-i} )$, where $Z_{-i}=\{Z_1, \ldots, Z_{n+1}\} \setminus\{Z_i\}$.
Assume that $V_1,\ldots,V_{n+1}$ are almost surely distinct.
%Similarly, let $z_1, \ldots, z_{n+1}$ and $v_1, \ldots, v_{n+1}$ indicate possible realizations of $Z_1, \ldots, Z_{n+1}$ and $V_1, \ldots, V_{n+1}$, respectively.
Then, $\forall \beta \in (0,1)$,
\begin{align} \label{eq:partial_quant}
  \mathbb{P}\bigg\{ V_{n+1} \leq Q\bigg(\beta; \, \sum_{i=1}^n
  p_i(Z_1,\ldots,Z_{n+1}) \delta_{V_i} + p_{n+1}(Z_1,\ldots,Z_{n+1})
  \delta_\infty\bigg) \bigg\} \geq \beta,
\end{align}
where $p_i(Z_1,\ldots,Z_{n+1}) = \bar h(Z_{-i},Z_i) / \sum_{j=1}^{n+1} \bar h(Z_{-j},Z_j)$.
\end{lemma}
%The proof of Lemma~\ref{lem:partial_quant} is in  Appendix~\ref{app:general-quantile}.

%Next, we will present a precise characterization of the appropriate function $\bar{h}$  involved in~\eqref{eq:partial_exch}, which in our case can be extracted from the proof of  Proposition~\ref{prop:paired-calibration-partial-exch}. This will pave the way to the computation of the conformalization weights utilized by Algorithm~\ref{alg:simultaneous-confidence-region}. Additional algorithmic details concerning the computation of these weights are deferred to Section~\ref{sec:implementation}.

\subsection{Characterization of the Conformalization Weights} \label{sec:weights}

Before we characterize explicitly the conformalization weights needed to apply Lemma~\ref{lem:partial_quant}, we need to introduce some intuitions and notations.
Define \( \bar{\cD}_{\mathrm{obs}} \) as an augmented version of \( \cD_{\mathrm{obs}} \) that includes the indices of the test group \( \mat{X}^* = (X^*_1, X^*_2, \ldots, X^*_K) \); i.e.,
$\bar{\cD}_{\mathrm{obs}} \coloneqq \cD_{\mathrm{obs}} \cup \left( \cup_{k=1}^K \{ X^*_k \} \right).$
Let \( D_{\mathrm{obs}} \) and \( \bar{D}_{\mathrm{obs}} \) denote possible realizations of \( \cD_{\mathrm{obs}} \) and \( \bar{\cD}_{\mathrm{obs}} \), respectively. Let \( \mat{x}_i = (x_{i,1}, \ldots, x_{i,K}) \) be a realization of \( \mat{X}^{\mathrm{cal}}_i \) and $\mat{x}_{n+1}$ a realization of $\mat{X}^*$.

Essentially, Lemma~\ref{lem:partial_quant} indicates that for $i \in [n]$, the conformalization weight $p_i$ is proportional to the joint distribution of the calibration groups and the test group when $\mat{x}_{n+1}$ is swapped with $\mat{x}_{i}$.
This swapping can be conceptualized as follows. In an alternate scenario, $\mat{x}_i$ is the realization of the test group $\mat{X}_{n+1}$, while $\{ \mat{x}_1,\ldots,\mat{x}_{n+1}\} \setminus \{\mat{x}_{i}\}$ becomes the realization of the calibration groups $\{ \mat{X}^{\text{cal}}_1,\ldots,\mat{X}^{\text{cal}}_{n}\}$. Intuitively, if the probability of observing this swapped realization is high, it suggests $\mat{x}_i$ resembles a test group, leading to a large $p_i$. The notation below is introduced to define relevant quantities after this hypothetical swapping.
For any $i \in [n]$, define \( D_{\mathrm{obs};i} \coloneqq \bar{D}_{\mathrm{obs}} \setminus \left( \bigcup_{k=1}^K \{ x_{i,k} \} \right) \) as the swapped observation set.
Further, let \( D_{\mathrm{obs};n+1} \coloneqq D_{\mathrm{obs}} \) represent the original observed data.

With the hypothetical swapping, the pruning process might also be affected. We define the following notations to quantify changes in the pruning process. Let \( n_{\mathrm{obs}}^c \) denote the number of observed entries in column \( c \) from \( D_{\mathrm{obs}} \). Define $\bar{n}_{\mathrm{obs}}^c \coloneqq n_{\mathrm{obs}}^c - \left( n_{\mathrm{obs}}^c \bmod K \right)$ as the number of observed entries remaining in column \( c \) after the random pruning step of Algorithm~\ref{alg:calibration-group}. For any \( i \in [n+1] \), let \( c_i \) denote the column to which the group \( \mat{x}_i \) belongs; that is,
$c_i \coloneqq x_{i,k,2},  \forall k \in [K],$
where \( x_{i,k,2} \) is the column index of the \( k \)-th entry in \( \mat{x}_i \).

Finally, let $D_{\mathrm{miss}}$ denote the realization of $\cD_{\mathrm{miss}}$, and $n^c_{\mathrm{miss}}$ be the number of entries in column $c$ from $D_{\mathrm{miss}}$. With slight abuse of notation, we denote the set of missing indices in column \( c_{n+1} \) excluding those in the test group \( \mat{x}_{n+1} \) as
$D_{\mathrm{miss}}^{c_{n+1}} \setminus \mat{x}_{n+1} \coloneqq D_{\mathrm{miss}}^{c_{n+1}} \setminus \left\{ x_{n+1,k} \right\}_{k=1}^K.$ With these notations established, we are now ready to state how Lemma~\ref{lem:partial_quant} applies in our setting, providing an explicit expression for the conformalization weights in~\eqref{eq:partial_quant}.

\begin{lemma} \label{lemma:conformalization-weights-explicit}
Under the setting of Proposition~\ref{prop:paired-calibration-partial-exch}, let $Z_1, \ldots, Z_n, Z_{n+1}$ denote $\mat{X}^{\mathrm{cal}}_1, \ldots, \mat{X}^{\mathrm{cal}}_n, \mat{X}^*$, and $V_1, \ldots, V_n, V_{n+1}$ represent the corresponding scores $S_1, \ldots, S_n, S^*$ given by~\eqref{eq:scores} and~\eqref{eq:def-score-test}, respectively, based on a matrix estimate $\hat{\mat{M}}$ computed based on the observations in $\cD_{\mathrm{train}}$. Then, Equation~\eqref{eq:partial_quant} from Lemma~\ref{lem:partial_quant} applies conditional on $\cD_{\mathrm{train}}$ and $\cD_{\mathrm{prune}}$,
with weights
\begin{align}\label{eq:gwpm-weights-k}
\begin{split}
    p_i(\mat{x}_{1}, \dots, \mat{x}_{n}, \mat{x}_{n+1})
    & \propto \mathbb{P}_{\mat{w}} \paren*{\cD_{\mathrm{obs}}= D_{\mathrm{obs};i}} \cdot \paren*{ \tilde{w}^*_{x_{i,1}} \prod\limits_{k=2}^{K}\tilde{w}^*_{x_{i,k}} } \\
    & \quad \cdot \left[ \cfrac{\mbinom{n^{c_i}_{\mathrm{obs}} }{\bar{n}^{c_i}_{\mathrm{obs}}}}{\mbinom{n^{c_i}_{\mathrm{obs}} -K}{\bar{n}^{c_i}_{\mathrm{obs}} -K}} \cdot  \cfrac{\mbinom{n^{c_{n+1}}_{\mathrm{obs}} }{\bar{n}^{c_{n+1}}_{\mathrm{obs}}}}{\mbinom{n^{c_{n+1}}_{\mathrm{obs}} +K}{\bar{n}^{c_{n+1}}_{\mathrm{obs}}+K}} \cdot \prod\limits_{k=1}^{K-1} \frac{\bar{n}^{c_i}_{\mathrm{obs}}-k}{\bar{n}^{c_{n+1}}_{\mathrm{obs}}+K-k} \right]^{\I{c_i \neq c_{n+1}}}.
  \end{split}
\end{align}
Above, $\tilde{w}^*_{x_{i,1}}$ and $\tilde{w}^*_{x_{i,k}}$ have explicit expressions that depend on the weights $\mat{w}^*$ in~\eqref{model:weighted_test_generation}; i.e.,
\begin{align}\label{eq:w-weights-i1}
    \tilde{w}^*_{x_{i,1}}
  = \frac{w^*_{x_{i,1}}} {\sum_{(r,c)\in D_{\mathrm{miss}}}w^{*}_{r,c}-\sum\limits_{k=1}^{K} \paren*{w^{*}_{x_{n+1,k}}-w^{*}_{x_{i,k}}}},
%  \quad \text{with}
\end{align}

% \begin{align*}
%     u^*_{x_{i,1}}  =  \I{c_{i} \neq c_{n+1}} \I{n^{c_{i}}_{\mathrm{miss}} <K} \paren*{\sum\limits_{(r,c)\in D^{c_{i}}_{\mathrm{miss}}} w^*_{r,c} -\I{n^{c_{n+1}}_{\mathrm{miss}} <2K} \sum\limits_{(r,c)\in D^{c_{n+1}}_{\mathrm{miss}} \setminus \mat{x}_{n+1}} w^*_{r,c}},
% \end{align*}
and, for all $k \in \{2,\ldots, K\}$,
\begin{align}\label{eq:w-weights-ik}
  \tilde{w}^*_{x_{i,k}}
  = \frac{w^*_{x_{i,k}}}{\sum\limits_{(r,c)\in D^{c_{i}}_{\mathrm{miss}}}  w^*_{r,c} + \sum\limits_{k'=k}^K w^{*}_{x_{i,k'}} -\I{c_i=c_{n+1}} \sum\limits_{k'=1}^K w^{*}_{x_{n+1,k'}}}.
\end{align}
% \textcolor{red}{
% \begin{align}\label{eq:w-weights-ik}
%   \tilde{w}^*_{x_{i,k}}
%   = \frac{w^*_{x_{i,k}}}{\sum\limits_{(r,c)\in \bar{D}^{c_{i}}_{\mathrm{miss}}}  w^*_{r,c} + \sum\limits_{k'=k}^K w^{*}_{x_{i,k'}} -\I{c_i=c_{n+1}} \sum\limits_{k'=1}^K w^{*}_{x_{n+1,k'}} + \I{n^{c_{i}}_{\mathrm{miss}} <K} \sum\limits_{(r,c)\in D^{c_{i}}_{\mathrm{miss}}} w^*_{r,c}}.
% \end{align}
% This is equivalent to the following simplification, since if $n^{c_{i}}_{\mathrm{miss}} <K, \sum\limits_{(r,c)\in \bar{D}^{c_{i}}_{\mathrm{miss}}}  w^*_{r,c} =0,$ otherwise, $\bar{D}^{c_{i}}_{\mathrm{miss}} = D^{c_{i}}_{\mathrm{miss}}$
% \begin{align}\label{eq:w-weights-ik}
%   \tilde{w}^*_{x_{i,k}}
%   = \frac{w^*_{x_{i,k}}}{\sum\limits_{(r,c)\in D^{c_{i}}_{\mathrm{miss}}}  w^*_{r,c} + \sum\limits_{k'=k}^K w^{*}_{x_{i,k'}} -\I{c_i=c_{n+1}} \sum\limits_{k'=1}^K w^{*}_{x_{n+1,k'}}}.
% \end{align}
% }
\end{lemma}

The main challenge in the computation of~\eqref{eq:gwpm-weights-k} arises from the term $\mathbb{P}_{\mat{w}} \paren*{\cD_{\mathrm{obs}}= D_{\mathrm{obs};i}}$, which is the probability of observing the matrix entries in $D_{\mathrm{obs};i}$ and depends on the sampling weights $\mat{w}$ in~\eqref{model:sampling_wo_replacement}.
Although this probability cannot be evaluated analytically, it can be approximated with an efficient algorithm, which makes it possible to compute the conformalization weights in~\eqref{eq:gwpm-weights-k} at cost of $\mathcal{O}(n_{r}n_c +nK)$, as explained in Section~\ref{sec:implementation}.

\subsection{Finite-Sample Coverage Bounds} \label{sec:bounds}

Algorithm~\ref{alg:simultaneous-confidence-region} produces confidence regions with simultaneous coverage for $\mat{X}^*$ sampled according to~\eqref{model:weighted_test_generation}.
This follows by integrating Proposition~\ref{prop:paired-calibration-partial-exch}, Lemma~\ref{lem:partial_quant}, and Equation~\eqref{eq:gwpm-weights-k}.

\begin{theorem} \label{thm:coverage-lower}
Suppose $\cD_{\mathrm{obs}}$ and $\cD_{\mathrm{miss}}$ are sampled according to~\eqref{model:sampling_wo_replacement}.
Let $\mat{X}^*$ be a test group sampled according to~\eqref{model:weighted_test_generation} conditional on $\cD_{\mathrm{obs}}$.
Then, for any fixed level $\alpha \in (0,1)$, the joint confidence region output by Algorithm~\ref{alg:simultaneous-confidence-region} satisfies~\eqref{eq:K-coverage} conditional on $\cD_{\mathrm{train}}, \cD_{\mathrm{prune}}$; i.e.,
$$\P{ \mat{M}_{\mat{X}^{*}} \in \hat{\mat{C}}(\mat{X}^*; \mat{M}_{\mat{X}_{\mathrm{obs}}}, \alpha) \mid \cD_{\mathrm{train}}, \cD_{\mathrm{prune}}} \geq 1-\alpha.$$
\end{theorem}

Note that the probability in Theorem~\ref{thm:coverage-lower} is taken over the randomness in $\cbrac{\mat{X}^{\mathrm{cal}}_{1}, \dots, \mat{X}^{\mathrm{cal}}_{n}}$ and $\mat{X}^*$, while $\cD_{\mathrm{train}}$ and $\cD_{\mathrm{prune}}$ are fixed.
Therefore, this result implies~\eqref{eq:K-coverage}.
Further, it is also possible to bound our simultaneous coverage from above.
\begin{theorem} \label{thm:upper_bound}
Under the same setting of Theorem~\ref{thm:coverage-lower},
$\P{ \mat{M}_{\mat{X}^{*}} \in \hat{\mat{C}}(\mat{X}^*; \mat{M}_{\mat{X}_{\mathrm{obs}}},\alpha)} \leq 1-\alpha+\EE{\max_{i \in [n+1]} p_i(\mat{X}^{\mathrm{cal}}_{1}, \dots, \mat{X}^{\mathrm{cal}}_{n}, \mat{X}^*)}$,
where the weights $p_i$ are given by~\eqref{eq:gwpm-weights-k} and the expectations can also be taken conditional on $\cD_{\mathrm{train}}$ and $\cD_{\mathrm{prune}}$, as in Theorem~\ref{thm:coverage-lower}.
\end{theorem}

%Theorem~\ref{thm:upper_bound} is proved in Appendix~\ref{app:proofs-bounds}.
A numerical investigation presented in Appendix~\ref{sec:upper-bound-numerical} demonstrates that in practice this upper bound converges to $1-\alpha$ as $n$ increases.
This is consistent with our empirical observations that Algorithm~\ref{alg:simultaneous-confidence-region} is not too conservative, as previewed in Figure~\ref{fig:movielens-preview}.

% \begin{align}
% \Delta=\frac{1}{2} \sum_{(i, j) \in \cbrac{\mat{X}^{\mathrm{cal}}_{1}, \dots, \mat{X}^{\mathrm{cal}}_{n}, \mat{X}^*}}\left|\frac{\widehat{w}_{i j}}{\sum_{\left(i^{\prime}, j^{\prime}\right) \in \cbrac{\mat{X}^{\mathrm{cal}}_{1}, \dots, \mat{X}^{\mathrm{cal}}_{n}, \mat{X}^*} } \widehat{w}_{i^{\prime} j^{\prime}}}-\frac{w_{i j}}{\sum_{\left(i^{\prime}, j^{\prime}\right) \in\cbrac{\mat{X}^{\mathrm{cal}}_{1}, \dots, \mat{X}^{\mathrm{cal}}_{n}, \mat{X}^*}} w_{i^{\prime} j^{\prime}}}\right|
% \end{align}

% {\color{blue}

Algorithm~\ref{alg:simultaneous-confidence-region} and Theorem~\ref{thm:coverage-lower} assume the sampling weights $\mat{w}$ are known. In practice, however, $\mat{w}$ are often unknown and need to be estimated, as detailed in Appendix~\ref{app:missingness-estimation}.
Therefore, we also present a coverage result for confidence regions based on estimated sampling weights $\hat{\mat{w}}$.
For $i\in [n+1]$, let $\hat{p}_i$ denote the conformalization weights computed using $\hat{\mat{w}}$, and let \(\hat{\mat{C}}^{\mathrm{est}}\) be the resulting confidence region.
To quantify the error between the true and estimated conformalization weights, we define the estimation gap
$\Delta=\frac{1}{2} \sum_{i=1}^{n+1}\left|\widehat{p}_i - p_i\right|
$.

\begin{theorem} \label{thm:coverage-lower-empirical}
For any $\alpha \in (0,1)$, the joint confidence region output by Algorithm~\ref{alg:simultaneous-confidence-region} using estimated sampling weights $\widehat{\mat{w}}$ satisfies
$
\P{ \mat{M}_{\mat{X}^{*}} \in \hat{\mat{C}}^{\mathrm{est}}(\mat{X}^*; \mat{M}_{\mat{X}_{\mathrm{obs}}}, \alpha) \mid \cD_{\mathrm{train}}, \cD_{\mathrm{prune}}} \geq 1-\alpha- \EE{\Delta}
$, under the same assumptions as Theorem~\ref{thm:coverage-lower}.
\end{theorem}

% }%end color blue

\section{Computational Shortcuts} \label{sec:implementation}

\subsection{Efficient Evaluation of the Conformalization Weights} \label{sec:implementation-weights}

We now explain how to efficiently approximate the weights $p_i$ in \eqref{eq:gwpm-weights-k}, for all $i \in [n+1]$.
The main challenge is to evaluate $\mathbb{P}_{\mat{w}} \paren*{\cD_{\mathrm{obs}}= D_{\mathrm{obs};i}}$ according to the \textit{weighted sampling without replacement} model defined in~\eqref{model:sampling_wo_replacement}, which is of order $\mathcal{O}(n_{\text{obs}}!)$ if approached directly.
In truth, it suffices to relate this probability, which depends on the index $i \in [n+1]$, to $\mathbb{P}_{\mat{w}} \paren*{\cD_{\mathrm{obs}}= D_{\mathrm{obs}}}$, which is constant and can thus be ignored when computing \eqref{eq:gwpm-weights-k}. In this section, we demonstrate that their ratio can be computed efficiently in $\mathcal{O}(n_{r}n_c +nK)$.

We begin by expressing $\mathbb{P}_{\mat{w}} \paren*{\cD_{\mathrm{obs}}= D_{\mathrm{obs}}}$ and $\mathbb{P}_{\mat{w}}(\calD_{\mathrm{obs}}=D_{\mathrm{obs};i})$, for any $i \in [n+1]$, as closed-form integrals.
Let $\delta=\sum_{(r,c)\in D_{\mathrm{miss}}}w_{r,c}$ denote the cumulative weight of all missing indices and, for any positive scaling parameter $h > 0$, define $\Phi(\tau; h)$ of $\tau \in (0,1]$ as
\begin{align} \label{eq:integral-def-phi}
    \Phi(\tau; h) := h\delta \tau^{h\delta-1} \prod_{(r,c) \in D_{\mathrm{obs}}} \paren*{1-\tau^{h w_{r,c}}}, \quad \quad \phi(\tau; h) := \log \Phi(\tau; h).
\end{align}
Further, define also $d_i := \sum_{k=1}^K(w_{x_{i,k}}-w_{x_{n+1,k}})$ for all $i \in [n+1]$.

\begin{proposition}\label{prop:integral}
For any fixed $n_{\mathrm{obs}} < n_r n_c$, scaling parameter $h > 0$, and $i \in [n+1]$,
%\begin{align}\label{eq:prob-obs-approx}
%\begin{split}
%    \mathbb{P}_{\mat{w}}(\calD_{\mathrm{obs}}=D_{\mathrm{obs}}) &= \int_{0}^{1} \Phi(\tau; h) \,d\tau.
%\end{split}
%\end{align}
%Further, for any $i \in [n+1]$,
\begin{align} \label{eq:prob-obs-approx-i}
  & \mathbb{P}_{\mat{w}}(\calD_{\mathrm{obs}}=D_{\mathrm{obs};i})
    = \int_{0}^{1} \Phi(\tau; h) \cdot \eta_i(\tau) \,d\tau,
\end{align}
where, for any $\tau \in (0,1)$,
\begin{align} \label{eq:def-integral-tau}
  & \eta_i(\tau; h) := \frac{\tau^{hd_i}(\delta + d_i)}{\delta} \cdot \paren*{\prod\limits_{k=1}^K \frac{1-\tau^{hw_{x_{n+1,k}}}}{1-\tau^{hw_{x_{i,k}}}}}.
\end{align}

\end{proposition}

Note that $\eta_{n+1}(\tau)=1$ for all $\tau$ if $i = n+1$, and in that case Proposition~\ref{prop:integral} recovers a classical result by \citet{wallenius1963biased}.
% See the proof of Proposition~\ref{prop:integral} in Appendix~\ref{app:proofs-weights-fast} for further details.
Further, the function $\eta_i$ in~\eqref{eq:def-integral-tau} is a product of $K$ simple functions of $\tau$, and therefore it is straightforward to evaluate even for large matrices.

Proposition~\ref{prop:integral} provides the foundation for evaluating the weights $p_i$ in~\eqref{eq:gwpm-weights-k}.
The remaining difficulty is that~\eqref{eq:prob-obs-approx-i} has no analytical solution.
Fortunately, the function $\Phi(\tau;h)$ satisfies some properties that make it feasible to approximate this integral accurately.
\begin{lemma} \label{lemma:phi-max}
    If $h > 1-\delta$, the function $\Phi(\tau;h)$ defined in~\eqref{eq:integral-def-phi} has a unique stationary point with respect to $\tau$ at some value $\tau_h \in (0,1)$. Further, $\tau_h$ is a global maximum.
\end{lemma}
See Figure~\ref{fig:phi-eta-plots} for a visualization of \(\Phi(\tau; h)\) and \(\eta_i(\tau; h)\), in two examples with independent and uniformly distributed sampling weights $\mat{w}$. These results show \(\Phi(\tau; h)\) becomes increasingly concentrated around its unique maximum for larger sample sizes, while \(\eta_i(\tau; h)\) remains relatively smooth at that point.
Thus, it makes sense to approximate this integral through a careful extension of Laplace's method. This is explained below.

\begin{figure*}[!htb]
  \subfloat[$n_{\text{obs}}=10$, $n_r n_c =100$, $K=10$.]{%
    \includegraphics[width=.48\linewidth]{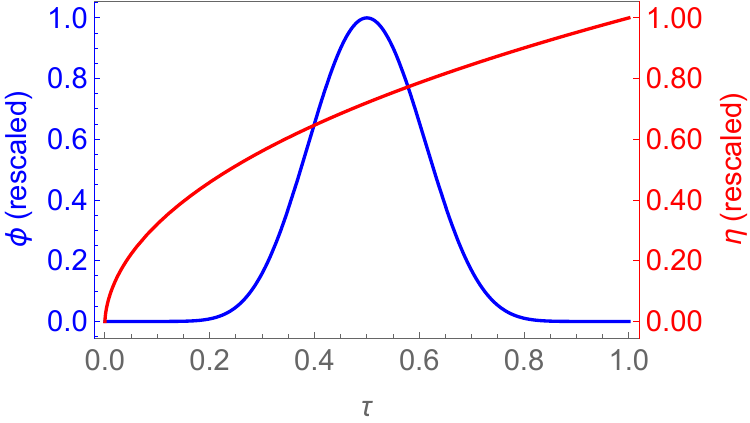}%
    % the sampling weights are uniformly distributed in (0,1)
    % the observed group to be swapped with the test samples are uniformly selected from observed entries
    % the test group are selected uniformly at random from the missing entries
    \label{subfig:a}%
  }\hfill
  \subfloat[$n_{\text{obs}}=200$, $n_r n_c =1000$, $K=10$.]{%
    % \subfloat[$n_{\text{obs}}=20$, $n_r n_c =100$, $K=10$]{%
    \includegraphics[width=.48\linewidth]{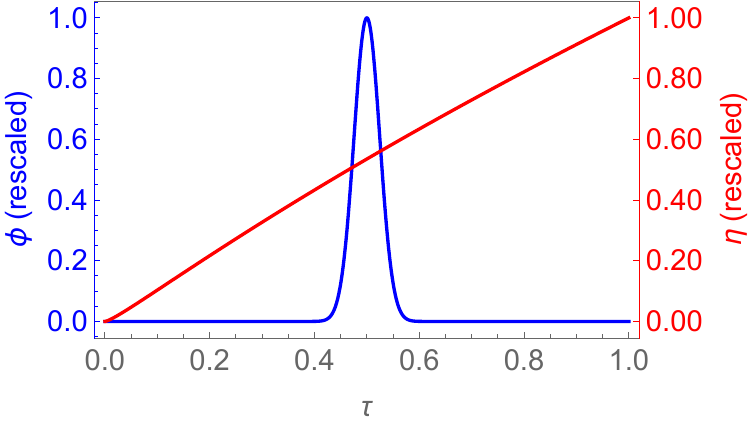}%
    \label{subfig:b}%
  }\\
  % \subfloat[$n_{\text{obs}}=100$, $n_r n_c =1000$, $K=10$]{%
  % \includegraphics[width=.48\linewidth]{figures/n1000_r0.1_uniform_weights_rescaled.pdf}%
  % \label{subfig:c}%
  % }\hfill
  %   \subfloat[$n_{\text{obs}}=200$, $n_r n_c =1000$, $K=10$]{%
  %   \includegraphics[width=.48\linewidth]{figures/n1000_r0.2_uniform_weights_rescaled.pdf}%
  %   \label{subfig:d}%
  % }
  \caption{Centered and re-scaled functions \(\Phi(\tau; h)\) and $\eta_i(\tau; h)$, with $h$ chosen so that \(\Phi(\tau; h)\) is maximized at $\tau = 1/2$, for different matrix sizes and numbers of observations. }
  \label{fig:phi-eta-plots}
\end{figure*}

The first step to approximate~\eqref{eq:prob-obs-approx-i} with a {\em generalized} Laplace method (justified in Section~\ref{sec:laplace}), is to move the peak of the integrand away from the boundary.
To this end, define $\tau_h := \argmax_{\tau \in (0,1)} \Phi(\tau; h)$,
where $h > 0$ is a free parameter. Our goal is to tune $h$ such that the peak is centered within the integration domain, specifically setting $h$ such that $\tau_{h}=1/2$.
According to Lemma~\ref{lemma:phi-max}, the function $\Phi(\tau; h)$ has a unique global maximum at $\tau_h \in (0,1)$ when $h > 1/\delta$. Therefore, a suitable value of $h > 1/\delta$ that satisfies $\tau_h = 1/2$ can be found with the Newton-Raphson algorithm; see Appendix~\ref{app:scaling-parameter} for further details.

As the number of observations grows, the peak of $\Phi(\tau;h)$ increasingly dominates the integral. A Taylor expansion reveals that the integral is primarily determined by the value of $\eta_i(\tau; h) \cdot \Phi(\tau; h)$ at $\tau = \tau_h$, and by the curvature of $\log \Phi(\tau; h)$ at the peak, namely $\phi^{''}(\tau_h; h)$.
This leads to the following Laplace approximation,
\begin{align}\label{eq:relation-obs}
\begin{split}
    \mathbb{P}_{\mat{w}}(\calD_{\mathrm{obs}}=D_{\mathrm{obs};i}) %&= \int_{0}^{1} \Phi(\tau; h) \cdot  \eta_i(\tau; h)\,d\tau \\
    \approx \eta_i(\tau_h; h) \cdot \Phi(\tau_h; h) \sqrt{\frac{-2\pi}{\phi^{''}(\tau_h; h)}} %\\
    \approx \eta_i(\tau_h; h) \cdot \mathbb{P}_{\mat{w}}(\calD_{\mathrm{obs}}=D_{\mathrm{obs}}).
\end{split}
\end{align}
This approximation is highly accurate in the large-sample limit and is particularly useful because it allows us to approximate the ratio $\mathbb{P}_{\mat{w}}(\calD_{\mathrm{obs}}=D_{\mathrm{obs};i})/\mathbb{P}_{\mat{w}}(\calD_{\mathrm{obs}}=D_{\mathrm{obs}})$ with $\eta_i(\tau_h; h)$ in \eqref{eq:def-integral-tau}, which is straightforward to calculate.
For example, if the sampling weights $\mat{w}$ in~\eqref{model:sampling_wo_replacement} are the same, $\eta_i(\tau; h) \equiv 1$ for all  $i \in [n+1]$ and any $\tau \in (0,1)$.

By combining~\eqref{eq:relation-obs} with~\eqref{eq:gwpm-weights-k}, it follows that, for each $i \in [n+1]$, the conformalization weight $p_i$ can be approximately rewritten in the large-sample limit as
$ p_i(\mat{x}_{1}, \dots, \mat{x}_{n}, \mat{x}_{n+1})
    \approx \bar{p}_i(\mat{x}_{1}, \dots, \mat{x}_{n}, \mat{x}_{n+1})/ \sum_{j=1}^{n+1} \bar{p}_j(\mat{x}_{1}, \dots, \mat{x}_{n}, \mat{x}_{n+1})$,
with the un-normalized weight $\bar{p}_i$ given by:
\begin{align} \label{eq:gwpm-weights-k-fast}
\begin{split}
  \bar{p}_i(\mat{x}_{1}, \dots, \mat{x}_{n}, \mat{x}_{n+1})
  & = \eta_i(\tau_h) \cdot \paren*{ \tilde{w}^*_{x_{i,1}} \prod\limits_{k=2}^{K}\tilde{w}^*_{x_{i,k}} } \\
    & \quad \cdot \left[ \cfrac{\mbinom{n^{c_i}_{\mathrm{obs}} }{\bar{n}^{c_i}_{\mathrm{obs}}}}{\mbinom{n^{c_i}_{\mathrm{obs}} -K}{\bar{n}^{c_i}_{\mathrm{obs}} -K}} \cdot  \cfrac{\mbinom{n^{c_{n+1}}_{\mathrm{obs}} }{\bar{n}^{c_{n+1}}_{\mathrm{obs}}}}{\mbinom{n^{c_{n+1}}_{\mathrm{obs}} +K}{\bar{n}^{c_{n+1}}_{\mathrm{obs}}+K}} \cdot \prod\limits_{k=1}^{K-1} \frac{\bar{n}^{c_i}_{\mathrm{obs}}-k}{\bar{n}^{c_{n+1}}_{\mathrm{obs}}+K-k} \right]^{\I{c_i \neq c_{n+1}}}.
  \end{split}
\end{align}
This makes Algorithm~\ref{alg:simultaneous-confidence-region} practical because $\bar{p}_i$ in~\eqref{eq:gwpm-weights-k-fast} only involves simple arithmetic and can be computed efficiently in $\mathcal{O}(n_r n_c + nK)$ for all $i\in [n+1]$, as explained in Appendix~\ref{app:complexity}.

The approximation in \eqref{eq:relation-obs} is not derived from a straightforward application of the traditional Laplace method, which deals with integrals of simpler functions (see Appendix~\ref{app:laplace-review}). However, the Taylor approximation concepts foundational to the Laplace method are adaptable enough to be extended to our context. Consequently, we have developed a specialized version of the Laplace approximation tailored for \eqref{eq:relation-obs}, and demonstrated that the approximation is consistent under additional technical assumptions, as detailed in Appendix~\ref{app:laplace-theory}.

\section{Empirical Demonstrations} \label{sec:empirical}

\subsection{Benchmarks and Setup}
We apply SCMC to simulated and real data, comparing its performance to those of two intuitive baselines: a naive {\em unadjusted} approach and a conservative {\em Bonferroni} approach.
The unadjusted heuristic is very similar to the approach of \citet{gui2023conformalized}: it applies Algorithm~\ref{alg:simultaneous-confidence-region} with $K=1$ repeatedly for every individual entry in $\mat{X}^*$; this does not ensure group-level coverage. We refer to Appendix~\ref{app:cmc} for details about the relation between this unadjusted benchmark and the original approach of \citet{gui2023conformalized}, since the difference between the two is subtle and not central to this comparison.
The Bonferroni baseline applies Algorithm~\ref{alg:simultaneous-confidence-region} with $K=1$ at level $\alpha/K$ instead of $\alpha$.
Both baselines are applied using the same matrix completion model leveraged by our method, and their predictions are calibrated using $Kn$ observed matrix entries.

These experiments are organized as follows.
Section~\ref{sec:exp-synthetic} utilizes simulated data, with Section~\ref{sec:exp-synthetic-uniform} focusing on (known) uniform sampling weights, and Section~\ref{sec:exp-synthetic-heterogeneous} allowing the sampling weights for the observed data to be heterogeneous (although still known exactly).
Section~\ref{sec:experiments-data} describes more realistic experiments using the MovieLens data, considering estimated sampling weights.
Section~\ref{sec:additional-numerical-results} highlights some additional results from the Appendix.
Appendices~\ref{app:additional-experiments} and~\ref{app:movielens-conditional} describe experiments with heterogeneous test weights using synthetic data and MovieLens data, respectively. Appendix~\ref{app:exp-mc-algorithms} presents results from different matrix completion algorithms. Appendix~\ref{app:exp-estimated-missingness} compares the performance of SCMC under true versus estimated sampling weights.
Appendix~\ref{sec:upper-bound-numerical} investigates the tightness of the theoretical coverage upper bounds derived in Section~\ref{sec:bounds}.

\subsection{Numerical Experiments with Synthetic Data} \label{sec:exp-synthetic}

\subsubsection{Uniform Sampling Weights} \label{sec:exp-synthetic-uniform}

A matrix $\mat{M}$ with $n_r = 200$ rows and $n_c = 200$ columns is generated using a ``signal plus noise'' model that exhibits both a low-rank structure and column-wise dependencies, reflecting, for example, that users may tend to agree on the quality of certain movies.
The strength of these dependencies is controlled by a parameter $\mu \geq 0$, with larger values of $\mu$ indicating stronger dependencies; see Appendix~\ref{app:exp-synthetic-uniform} for further details.
As we shall see, the advantage of our approach over the Bonferroni benchmark becomes more pronounced as $\mu$ increases.
Given $\mat{M}$, we begin by considering a completely random observation pattern.

We observe $n_{\mathrm{obs}}=8000$ entries of $\mat{M}$, sampled according to~\eqref{model:sampling_wo_replacement} with $w_{r,c} = 1$ for all $(r,c) \in [n_r] \times [n_c]$.
Let $\cD_{\mathrm{obs}}$ denote the (unordered) observed indices.
Then, 100 test groups $\mat{X}^*$ of size $K$, where $K \geq 2$ is a control parameter, are sampled without replacement from $\cD_{\mathrm{miss}} = [n_r] \times [n_c] \setminus \cD_{\mathrm{obs}}$, according to~\eqref{model:weighted_test_generation} with $w^*_{r,c} = 1$ for all $(r,c) \in \cD_{\mathrm{miss}}$.
The confidence region for the group $\mat{X}^*$ is constructed by Algorithm~\ref{alg:simultaneous-confidence-region} with $n = \min \cbrac{1000, \lfloor \xi_{\text{obs}}/2 \rfloor}$, where $\xi_{\text{obs}}$ denotes the maximum possible number of calibration groups.
The matrix completion model is thus trained using $n_{\mathrm{train}}=n_{\mathrm{obs}}-Kn$ observed entries of $\mat{M}$, indexed by $\cD_{\mathrm{train}}$. While SCMC can leverage any completion model producing point predictions, here we employ an alternating least squares approach \citep{hu_cf_2008}, designed to recover low-rank signals.
% We apply this algorithm with a hypothesized rank of 30.
% However, the validity of SCMC is independent of both $\mat{M}$ and the matrix completion model.
See Appendix~\ref{app:exp-mc-algorithms} for experiments with different models.

\begin{figure}[!htb]
    \centering
    \includegraphics[width=0.8\linewidth]{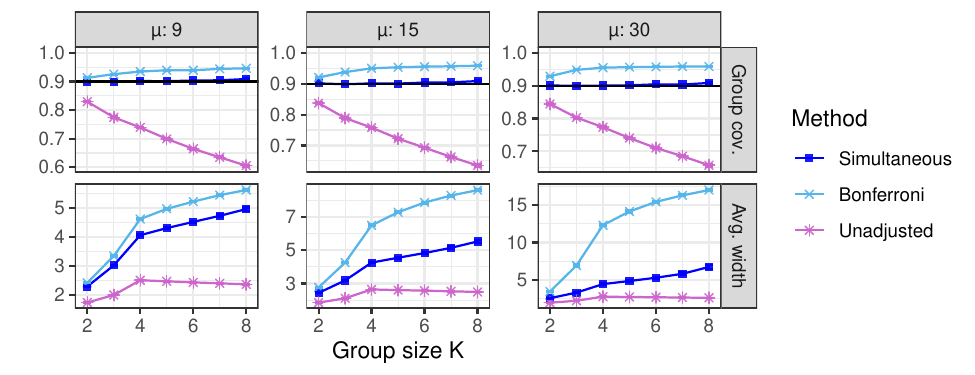}
    \caption{Performance of SCMC and two baselines on simulated data, as a function of the group size $K$ and of the column-wise noise parameter $\mu$. The nominal coverage is $90\%$.}
    \label{fig:exp_uniform_vary_k}
\end{figure}

Figure~\ref{fig:exp_uniform_vary_k} summarizes the results as a function of $K$, for different values of $\mu$ and $\alpha = 10\%$.
Each method is assessed in terms of the average width of the output confidence regions and of the empirical simultaneous coverage for the 100 test groups. All results are averaged over 300 independent experiments.
Our method attains 90\% simultaneous coverage, while the unadjusted baseline becomes increasingly invalid as $K$ grows.
Further, our method leads to more informative confidence regions compared to the Bonferroni baseline, as the latter is increasingly conservative for larger $K$ and $\mu$.
See Figure~\ref{fig:exp_uniform_vary_mu} in Appendix~\ref{app:numerical-results} for a different view of these results, shown as a function of $\mu$ for different values of $K$.

\subsubsection{Heterogeneous Sampling Weights} \label{sec:exp-synthetic-heterogeneous}

We now consider experiments where the matrix $\mat{M}$ with $n_c = n_r = 300$ is modeled as a low-rank component plus additive noise that has no special patterns, while he sampling weights $\mat{w}$ of the observation model~\eqref{model:weighted_test_generation} are heterogeneous.
Specifically, $\mat{w}$ is chosen to introduce a column-wise dependent missingness pattern, resulting in some columns being more densely observed than others. For example, popular movies tend to receive more ratings, leading to denser observations in certain columns. The degree of heterogeneity is controlled by a parameter $s$, where a decrease in $s$ indicates stronger heterogeneity of sparsity levels. For precise model descriptions, see Appendix~\ref{app:exp-synthetic-heterogeneous}.
Using this model, we randomly sample without replacement $n_{\mathrm{obs}}=27,000$ matrix entries (from a total of 90,000) and then apply Algorithm~\ref{alg:simultaneous-confidence-region} similarly to the previous section, using $n = \min \cbrac{2000, \lfloor \xi_{\text{obs}}/2 \rfloor}$ calibration groups.
%and allocating the remaining $n_{\mathrm{train}}=n_{\mathrm{obs}}-Kn$ observations to train the alternating least squares algorithm as in the previous section.
%with hypothesized rank 8.
%Results using different matrix completion models are presented in \Cref{app:exp-mc-algorithms}.
The two baseline approaches are implemented as in Section~\ref{sec:exp-synthetic-uniform}.
All methods are evaluated on a test set of 100 test groups sampled without replacement according to the model defined in \eqref{model:weighted_test_generation}, with uniform weights $w^*_{r,c} = 1$ for all $(r,c) \in [300] \times [300]$.
The $\alpha$ level is $10\%$. All results are averaged over 300 independent experiments.

\begin{figure}[!htb]
    \centering
    \includegraphics[width=0.8\linewidth]{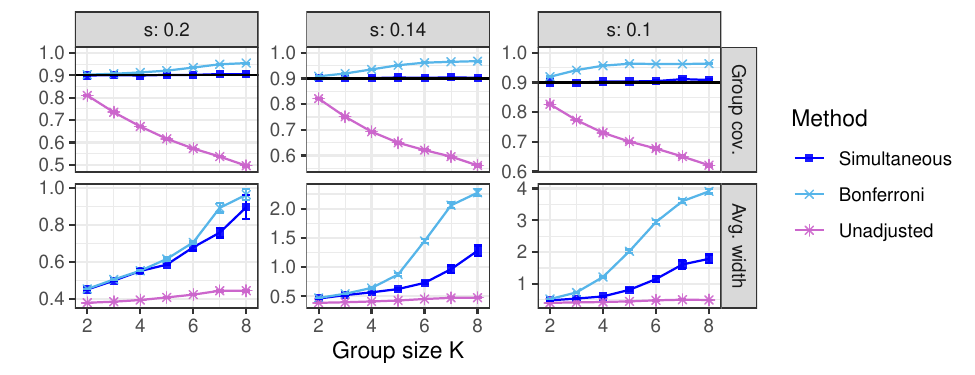}
    \caption{Performance of SCMC and two baselines on simulated data with a heterogeneous observation pattern whose strength is controlled by a scale parameter $s$. Smaller $s$ indicates stronger heterogeneity. Other details are as in Figure~\ref{fig:exp_uniform_vary_k}.}
    \label{fig:exp_biased_obs_vary_k}
\end{figure}

Figure~\ref{fig:exp_biased_obs_vary_k} summarizes the results as a function of group size $K$, for different heterogeneity levels controlled by $s$.
SCMC achieves valid simultaneous coverage and provides more informative confidence regions compared to the Bonferroni approach, which becomes overly conservative for smaller values of $s$.
This can be understood as follows.
The matrix completion model naturally performs better at estimating missing entries in more densely observed columns. Consequently, the heterogeneous sampling model introduces column-wise dependencies in the residual matrix $\hat{\mat{M}} - \mat{M}$.
These dependencies, which intensify as $s$ decreases, make the simultaneous inference task intrinsically challenging, resulting in wider confidence regions for all methods, but have a disproportionate adverse effect on the Bonferroni approach.
See Figure~\ref{fig:exp_biased_obs_vary_scale} in Appendix~\ref{app:numerical-results} for a different view of these results, highlighting the behavior of all methods as a function of $s$, for different values of $K$.

\subsection{Numerical Experiments with MovieLens Data} \label{sec:experiments-data}

We apply our method to the MovieLens 100K data \citep{movielens100k}, which comprises 100,000 numerical ratings given by 943 users for 1682 movies.
Approximately 94\% of the ratings are missing.
To reduce the memory requirements of the completion algorithm estimating $\hat{\mat{M}}$, we reduce the matrix size by half, focusing on a smaller rating matrix $\mat{M} \in \mathbb{R}^{800\times1000}$, corresponding to a random subset of 800 users and 1000 movies.

% As usual, we denote the set of indices for the observed matrix entries as $\cD_{\tt{obs}}$ and its complement as $\cD_{\tt{miss}} = [800]\times[1000] \setminus \cD_{\tt{obs}}$.
Since the true sampling weights $\mat{w}$ are unknown in this application, we compute estimated weights $\hat{\mat{w}}$ with a data-driven approach as described in Appendix~\ref{app:missingness-estimation}.
Algorithm~\ref{alg:simultaneous-confidence-region} is then applied with $\hat{\mat{w}}$ instead of $\mat{w}$, to construct confidence regions for 100 random test groups $\mat{X}^*$. We utilize  $n = \min \cbrac{1000, \lfloor \xi_{\text{obs}}/2 \rfloor}$ calibration groups and vary the group size $K$ as a control parameter.
The test groups are randomly sampled without replacement from $\cD_{\tt{miss}}$ according to the model defined in~\eqref{model:weighted_test_generation}, with uniform weights $\mat{w}^*$.
The alternating least square algorithm is trained as described in the previous sections.
%, applying the alternating least squares approach of \citet{hu_cf_2008} based on $n_{\mathrm{train}}=n_{\mathrm{obs}}-Kn$ observations.
The hypothesized rank of $\mat{M}$ utilized by this model to obtain $\hat{\mat{M}}$ is varied as an additional control parameter.
% As before, the baseline approaches are also applied based on the same matrix completion model, to facilitate the comparison with our method.

Figure~\ref{fig:movielens-preview}, previewed in Section~\ref{sec:intro-preview}, reports on the results of these experiments as a function of the group size $K$ and of the hypothesized rank utilized by the matrix completion model.
The confidence regions are assessed based on their average width alone, since it is impossible to measure the empirical coverage given that the ground truth is unknown.
The results show SCMC produces more informative (narrower) confidence regions compared to the Bonferroni approach, consistent with Section~\ref{sec:exp-synthetic}. Figure~\ref{fig:movielens-preview} compares only with the Bonferroni baseline because the unadjusted baseline is not intended to provide valid simultaneous coverage, making it less suitable for comparisons lacking a verifiable ground truth. Nevertheless, Figure~\ref{fig:movielens-fullmiss} in Appendix~\ref{app:additional-exp-movielens} includes a comparison with both baselines, demonstrating that our simultaneous confidence regions are not much wider than those produced by the unadjusted baseline. Further, our method's higher reliability compared to the unadjusted baseline is supported by the following additional experiments, conducted using the same data but under a more artificial setting in which the ground truth is known.

To evaluate the coverage on the MovieLens data, we carry out similar but more closely controlled experiments in which the test groups are drawn not from $\cD_{\tt{miss}}$ (for which the ground truth is unknown) but from a hold-out subset $\cD_{\tt{hout}}$ containing 20\% of the observed matrix indices in $\cD_{\tt{obs}}$.
Algorithm~\ref{alg:simultaneous-confidence-region} is then applied to construct confidence regions for the unobserved ratings of 100 random test groups $\mat{X}^*$ sampled from $\cD_{\tt{hout}}$, proceeding as described before but utilizing only the observed data in $\cD_{\tt{obs}} \setminus \cD_{\tt{hout}}$ instead of $\cD_{\tt{obs}}$.

% In this setting, both the sampling weights of the observed data and those of the test groups are assumed to be uniform; i.e., $w_{r,c} = w^*_{r,c}=1$ for all $(r,c) \in [n_r] \times [n_c]$. Note that this is a well-justified modeling choice here, since our controlled experimental setup is effectively treating $\cD_{\tt{obs}}$ as the full set of all matrix indices, while $\cD_{\tt{miss}}$ is completely disregarded from the sampling processes generating both the observed and test data.
% Figure~\ref{fig:movielens-preview}, previewed earlier in Section~\ref{sec:intro-preview}, compares the performances of all methods under this experimental setup, averaging over 100 independent repetitions.
% Note that the matrix completion model here is applied with a hypothesized rank equal to 7.
% The results show that our method consistently achieves the desired simultaneous coverage, unlike the unadjusted baseline, and produces more informative confidence regions compared to the Bonferroni approach.

% Finally, to investigate the robustness of SCMC to possible misspecification of the sampling weights, we conduct additional experiments under the following more challenging  setup.
% Algorithm~\ref{alg:simultaneous-confidence-region} is still applied leveraging only the data in $\cD_{\tt{obs}} \setminus \cD_{\tt{hout}}$ instead of $\cD_{\tt{obs}}$ but now it utilizes estimated sampling weights $\hat{\mat{w}}$ computed with an approach inspired by \citet{gui2023conformalized}, as described in Appendix~\ref{app:missingness-estimation}.

Since the estimation of $\hat{\mat{w}}$ acknowledges the existence of unobserved entries in $\cD_{\tt{miss}}$, our method provides coverage for $\mat{X}^*$ sampled from $\cD_{\tt{miss}} \cup \cD_{\tt{hout}}$ instead of $\cD_{\tt{hout}}$.
Of course, we can only evaluate the empirical coverage for test groups sampled from $\cD_{\tt{hout}}$. This makes these experiments valuable for understanding the robustness of our inferences against potential distribution shifts between $\cD_{\tt{miss}} \cup \cD_{\tt{hout}}$ and $\cD_{\tt{hout}}$.

Figure~\ref{fig:movielens-masked} focuses on test groups sampled from $\cD_{\tt{hout}}$.
%The results are reported as a function of $K$, for different hypothesized ranks in the matrix completion model.
Our method leads to more informative inferences compared to the Bonferroni approach, and it nearly achieves 90\% simultaneous coverage for the test groups sampled from $\cD_{\tt{hout}}$, even though in theory it is valid on average over all test groups sampled from $\cD_{\tt{miss}} \cup \cD_{\tt{hout}}$.
This demonstrates the robustness of our method towards possible misspecification of the sampling weights.

\begin{figure}[!htb]
  \centering
  \includegraphics[width=0.8\linewidth]{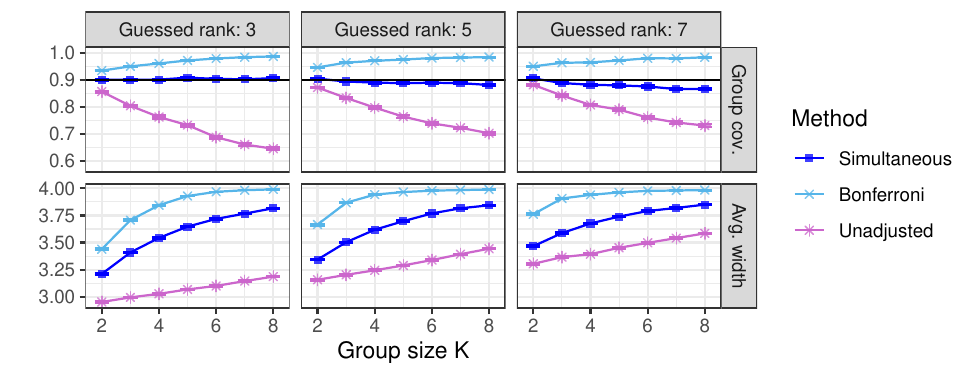}
  \caption{Performance of SCMC and two baselines on the MovieLens data set.
    These confidence regions are evaluated on a subset of hold-out observations, for which the ground truth is known, but they are calibrated to achieve valid simultaneous coverage on average over all latent matrix entries. Other details are as in Figure~\ref{fig:movielens-preview}.}
  \label{fig:movielens-masked}
\end{figure}

\subsection{Additional Numerical Experiments}\label{sec:additional-numerical-results}

Figures~\ref{fig:exp-solver-uniform}--\ref{fig:exp-solver-biased} in Appendix~\ref{app:exp-mc-algorithms} confirm that the relative performance of SCMC and benchmark methods remains consistent across different matrix completion models.
Figures~\ref{fig:exp-est-uniform}--\ref{fig:exp-est-biased} in Appendix~\ref{app:exp-estimated-missingness} show that SCMC maintains coverage close to the nominal level even when the sampling weights $\mat{w}$ are replaced by an empirical estimate $\hat{\mat{w}}$, demonstrating robustness to reasonable estimation errors.
Appendix~\ref{app:additional-experiments} further illustrates that SCMC can approximate conditional coverage when applied with appropriate test sampling weights $\mat{w}^*$ in synthetic settings, while Appendix~\ref{app:movielens-conditional} extends this conditional validation to the MovieLens dataset, using movie genre information to construct test weights.
Figure~\ref{fig:movielens-conditional} analyzes four genres (Children's, Crime, Drama, Romance), revealing that SCMC adaptively produces narrower confidence regions for more predictable genres (e.g., Children's) and wider regions for noisier, more subjective genres (e.g., Drama).

\section{Discussion} \label{sec:conclusion}

This paper introduces a method for simultaneous conformal inference in matrix completion.  
While motivated by group recommender systems, our approach is modular and extends beyond this initial application.  
In particular, our structured calibration strategy could be readily adapted for other tasks, such as bundle recommendations.  
More broadly, the leave-one-out exchangeability framework and related conformalization techniques may help extend weighted conformal inference beyond covariate shift in other directions \citep{tibshirani-covariate-shift-2019}.
Future work could refine our method by accounting for discrepancies between the sampling ($\mat{w}$) and test ($\mat{w}^*$) weights in Algorithm~\ref{alg:calibration-group}, moving beyond the current sampling without replacement scheme in Algorithm~\ref{alg:calibration-group} and potentially leading to even more informative confidence regions, albeit at the cost of breaking the ``leave-one-out exchangeability'' property that enables our computational shortcuts.
Determining whether such an extension can be implemented efficiently remains an open challenge.  
Additionally, it may be possible to enhance our approach by incorporating the jackknife+ procedure \citep{barber_cv+_2021} instead of the split conformal approach, which could be particularly beneficial in settings with very limited observations.  
We leave these extensions for future research.

% \section*{Software Availability}

% A software package implementing the methods and numerical experiments described in this paper is available at \url{https://github.com/ZiyiLiang/simultaneous-matrix-completion}.

% \section*{Acknowledgements}

% M.~S.~was partly supported by NSF grant DMS 2210637.

%%% Local Variables:
%%% mode: latex
%%% TeX-master: "main_jasa"
%%% End:

\section*{Acknowledgements}

M.~S.~and T.~X.~were partly supported by NSF grant DMS 2210637.

\printbibliography
\newpage

\appendix
% Special numbering for appendix (Use "S" instead of "A" in Supplement)
\renewcommand{\thesection}{A\arabic{section}}
\renewcommand{\theequation}{A\arabic{equation}}
\renewcommand{\thetheorem}{A\arabic{theorem}}
\renewcommand{\thecorollary}{A\arabic{corollary}}
\renewcommand{\theproposition}{A\arabic{proposition}}
\renewcommand{\thelemma}{A\arabic{lemma}}
\renewcommand{\thetable}{A\arabic{table}}
\renewcommand{\thefigure}{A\arabic{figure}}
\renewcommand{\thealgorithm}{A\arabic{algorithm}}

\section{Additional Technical Background} \label{app:baselines-background}

\subsection{Relation to Conformal Inference under Covariate Shift} \label{app:cov-shift}

This paper builds upon the {\em weighted conformal inference} framework of \citet{tibshirani-covariate-shift-2019}, originally designed to address non-exchangeability arising from {\em covariate shift}.  
In a typical regression or classification setting, where the data consist of pairs $(X,Y)$ representing a covariate vector $X$ and a label $Y$, covariate shift can occur when the calibration and test data distributions satisfy:
\begin{align}\label{eq:covariate-shift-setup}
\begin{split}
    (X_i, Y_i) & \overset{\text{i.i.d.}}{\sim} P_X \times P_{Y \mid X}, \qquad i = 1,\dots, n_{\text{obs}},\\
    (X^{*}, Y^*) & \overset{\text{ind.}}{\sim} \Tilde{P}_X \times P_{Y \mid X},
  \end{split}
\end{align}
where $(X_i, Y_i), i = 1,\dots, n_{\text{obs}}$, represent the observed covariates and labels in the calibration set, and $(X^*, Y^*)$ denotes the test point.  
Here, covariate shift arises when the distribution of $X$ at test time, $\Tilde{P}_X$, differs from the corresponding distribution $P_X$ in the calibration set, while the conditional distribution $P_{Y \mid X}$ remains unchanged.  
\citet{tibshirani-covariate-shift-2019} demonstrated that this is a special case of non-exchangeability that is relatively easy to address within a conformal inference framework.

To address covariate shift, \citet{tibshirani-covariate-shift-2019} introduced the notion of {\em weighted exchangeability}, a relaxation of full exchangeability that still holds under~\eqref{eq:covariate-shift-setup} even if $\Tilde{P}_X \neq P_X$.  
Their key theoretical result states that if the calibration and test data satisfy weighted exchangeability, valid conformal inferences can be obtained by appropriately reweighting the calibration points, adjusting for their distributional shift relative to $\Tilde{P}_X$.

However, the distribution shift encountered in our structured matrix completion setting is different from covariate shift.  
In our case, non-exchangeability arises due to sampling without replacement.  
Specifically, we assume that observed matrix indices are drawn without replacement from $[n_r]\times [n_c]$, followed by sequential sampling of structured calibration and test groups from the observed and missing entries, respectively.  
This results in two key differences from covariate shift:
(i) The structured calibration groups are mutually dependent, as they must all belong to the observed portion of the matrix, unlike the i.i.d. case in~\eqref{eq:covariate-shift-setup}.
(ii) The test group is also dependent on the calibration groups, as it is sampled exclusively from previously unobserved entries.

Consequently, our data do not satisfy the weighted exchangeability property considered by \citet{tibshirani-covariate-shift-2019}.  
Nonetheless, we show that key ideas from their work, originally embedded in the proof of their quantile inflation result, can be extended to our setting.  
Specifically, Lemma~\ref{lem:general_quant} generalizes their quantile inflation result, making it sufficiently flexible to address the challenges posed by our problem.  
Similar to \citet{tibshirani-covariate-shift-2019}, we show how to obtain valid conformal inferences for non-exchangeable data by assigning appropriate weights to the calibration points (in our case, the structured calibration groups), adjusting for their departure from the test-time distribution.  
However, an additional complication arises: while weighted exchangeability leads to simple weight formulations, the weights derived in Lemma~\ref{lem:general_quant} are generally computationally intractable.

Fortunately, we can show that the weights in Lemma~\ref{lem:general_quant} remain computationally feasible in our matrix completion setting due to a distinct partial symmetry property, which we term {\em leave-one-out exchangeability}.  
Although different from weighted exchangeability, this property similarly ensures that the weights remain tractable.  
This allows us to develop specialized techniques for efficient computation of these generalized weights.  

\subsection{Individual-Level Conformalized Matrix Completion }\label{app:cmc}

This section reviews the {\em conformalized matrix completion} method proposed by \citet{gui2023conformalized}, which is designed to produce confidence intervals for one missing entry at a time.

The setup of \citet{gui2023conformalized} is similar to ours as they also treat $\mat{M}$ as fixed and assume the randomness in the matrix completion problem comes from the observation process or, equivalently, the missingness mechanism.
However, their modeling choices do not match exactly with ours.
Specifically, they assume that each matrix entry in row $r$ and column $c$ is independently observed with some (known) probability $p_{r,c}$, which roughly corresponds to our sampling weights $w_{r,c}$ in~\eqref{model:sampling_wo_replacement}; i.e., $\mathcal{D}_{\mathrm{obs}} := \{(r,c) \in [n_r]\times[n_c]: Z_{r,c}=1\}$, where
\begin{align}\label{model:cmc-observed}
    Z_{r,c} = \I{(r,c) \text{ is observed}} \overset{\mathrm{ind.}}{\sim} \mathrm{Bernoulli}(p_{r,c}), \qquad \forall (r,c) \in [n_r]\times[n_c].
\end{align}
Therefore, the total number of observed entries is a random variable in \citet{gui2023conformalized}, whereas we can allow $n_{\mathrm{obs}}$ to be fixed within the sampling without replacement model defined in~\eqref{model:sampling_wo_replacement}.
As shown in this paper, our modeling choice is natural when aiming to construct group-level simultaneous inferences.
The model assumed by \citet{gui2023conformalized} also differs from ours in its requirement that all sampling weights must be strictly positive; $p_{r,c}>0$ for all $(r,c) \in [n_r] \times [n_c]$ in~\eqref{model:cmc-observed}.
Further, the approach of \citet{gui2023conformalized} differs from ours in that they assume the missing matrix index of interest, namely $\mathcal{I}^* \in [n_r]\times[n_c]$, to be sampled {\em uniformly} at random from $\mathcal{D}_{\mathrm{miss}}$, that is, $\mathcal{I}^* \sim \mathrm{Uniform}(\mathcal{D}_{\mathrm{miss}})$, where $\mathcal{D}_{\mathrm{miss}} := [n_r]\times[n_c] \setminus \mathcal{D}_{\mathrm{obs}}$. By contrast, our sampling model for the test groups, defined in~\eqref{model:weighted_test_generation}, can accommodate heterogeneous weights $\mat{w}^*$.

The method proposed by \citet{gui2023conformalized} constructs conformal confidence intervals for individual missing entries as follows.
First,  $\mathcal{D}_{\mathrm{obs}}= \cbrac{(r,c)\in [n_r]\times[n_c]: Z_{r,c}=1}$ is partitioned into a training set $\mathcal{D}_{\mathrm{train}}$ and a disjoint calibration set $\mathcal{D}_{\mathrm{cal}}$ by randomly sampling $\tilde{Z}_{r,c} \sim \mathrm{Bernoulli}(q)$ independently for all $(r,c) \in [n_r] \times [n_c]$, for some fixed parameter $q \in (0,1)$, and then defining
\begin{align}\label{model:cmc-split}
    \mathcal{D}_{\mathrm{train}} \coloneqq \cbrac{(r,c) \in \mathcal{D}_{\mathrm{obs}}: \tilde{Z}_{r,c}=1}, \quad \quad \mathcal{D}_{\mathrm{cal}} \coloneqq \cbrac{(r,c) \in \mathcal{D}_{\mathrm{obs}}: \tilde{Z}_{r,c}=0}.
\end{align}
Similar to us, \citet{gui2023conformalized} utilize  $\mathcal{D}_{\mathrm{train}}$ to compute $\hat{\mat{M}}$, leveraging any black-box algorithm, and then evaluate conformity scores on the calibration data as explained below.

Let $\mathcal{C}(\cdot, \tau, \hat{\mat{M}})$ denote a pre-specified prediction rule for a single matrix entry, which should be monotonically increasing in the tuning parameter $\tau \in [0,1]$ as explained in Section~\ref{sec:method-outline}; for example, this could correspond to the hyper-cubic prediction rule described in Section~\ref{sec:method-outline} in the special case of $K=1$.
For any $I \in [n_r] \times [n_c]$, let $S_i = S(I_i)$ denote the conformity score corresponding to $\mathcal{C}(\cdot, \tau, \hat{\mat{M}})$, as in~\eqref{eq:scores}.
Imagining that the calibration set contains the indices of $n$ matrix entries---$\mathcal{D}_{\mathrm{cal}}=\cbrac{\mathcal{I}_1, \dots, \mathcal{I}_{n}}$---the method of \citet{gui2023conformalized} evaluates $S_i = S(I_i)$ for all $i \in [n]$ and then calibrates the tuning parameter $\tau$ by computing
\begin{align}\label{eq:cmc-weighted-quantile}
\hat{\tau}^{\mathrm{indiv}}_{\alpha,1} = Q\paren*{1-\alpha; \sum_{i=1}^{n} p_i^{\mathrm{indiv}} \delta_{S_i} + p_{n+1}^{\mathrm{indiv}}\delta_{\infty}},
\end{align}
where the conformalization weights $p_i^{\mathrm{indiv}}$ are given by
\begin{align}\label{eq:cmc-weights}
    & p_i^{\mathrm{indiv}}(I_1, \dots, I_{n+1}) = \frac{R_{I_i}}{\sum_{j}^{n+1} R_{I_j}},
      \qquad  R_{I_i} =  \frac{1-p_{I_i}}{p_{I_i}},
\end{align}
for all $i \in [n+1]$, with the convention that $I_{n+1} = I^*$.
Finally, the confidence interval for the latent value of $\mat{M}$ at index $I^*$ is given by:
\begin{align}\label{eq:cmc-prediction-region}
    \hat{\mat{C}}^{\mathrm{indiv}}(\calI^*; \mat{M}_{\mathcal{D}_{\mathrm{obs}}}, \alpha) = \mathcal{C}(\calI^*, \hat{\tau}^{\mathrm{indiv}}_{\alpha,1}, \hat{\mat{M}}).
\end{align}
The following result establishes that the confidence intervals defined in~\eqref{eq:cmc-prediction-region} have guaranteed marginal coverage at level $1-\alpha$.

\begin{proposition}[from \citet{gui2023conformalized}] \label{prop:cmc-coverage}
    Suppose $\mathcal{D}_{\mathrm{obs}}$ is sampled according to~\eqref{model:cmc-observed} and $\calI^* \sim \mathrm{Uniform}(\mathcal{D}_{\mathrm{miss}})$.
    Then, for any $\alpha \in (0,1)$, $\hat{\mat{C}}^{\mathrm{indiv}}(\calI^*; \mat{M}_{\mathcal{D}_{\mathrm{obs}}},\alpha)$ in~\eqref{eq:cmc-prediction-region} satisfies
    \begin{align*}
        \P{ \mat{M}_{\mathcal{I}^{*}} \in \hat{\mat{C}}^{\mathrm{indiv}}(\calI^*; \mat{M}_{\mathcal{D}_{\mathrm{obs}}},\alpha) \mid \mathcal{D}_{\mathrm{train}}} \geq 1-\alpha.
    \end{align*}
  \end{proposition}

  \begin{proof}
    This result follows directly from the proof of Theorem~3.2 in \citet{gui2023conformalized}. Alternatively, the following proof can be obtained by applying our Lemma~\ref{lem:partial_quant}.
    Conditioning on $n$, such that $\mathcal{D}_{\mathrm{cal}}=\cbrac{\mathcal{I}_1, \dots, \mathcal{I}_{n}}$, note that the joint distribution of $\mathcal{I}_1, \dots, \mathcal{I}_{n}, \calI^*$ trivially satisfies the leave-one-out exchangeability condition defined in~Lemma~\ref{lem:partial_quant}.
    Specifically, let $I_1, \dots,I_{n}, I_{n+1}$ be a realization of $\mathcal{I}_1, \dots, \mathcal{I}_{n}, \calI^*$, so that $D_{\mathrm{cal}} \coloneqq \cbrac{I_1, \dots, I_{n}}$ is a realization of $\mathcal{D}_{\mathrm{cal}}$ sampled according to~\eqref{model:cmc-split}. Then,
    \begin{align*}
      &\mathbb{P}\paren*{ \calI_1 = I_1, \dots, \calI_{n}=I_{n}, \calI^*=I_{n+1} \mid \mathcal{D}_{\mathrm{train}}, \abs*{\mathcal{D}_{\mathrm{obs}}}=n_{\mathrm{obs}}, \abs*{\mathcal{D}_{\mathrm{cal}}}=n}\\
      & \quad = \mathbb{P}\paren*{ \mathcal{D}_{\mathrm{cal}} =D_{\mathrm{cal}}, \calI^*=I_{n+1} \mid \mathcal{D}_{\mathrm{train}}, \abs*{\mathcal{D}_{\mathrm{obs}}}=n_{\mathrm{obs}}, \abs*{\mathcal{D}_{\mathrm{cal}}}=n}\\
      & \quad = g(D_{\mathrm{cal}}\cup \cbrac*{I_{n+1}} ) \cdot \prod_{(r,c) \in D_{\mathrm{cal}}} \frac{p_{r,c}}{1-p_{r,c}},
    \end{align*}
    where the second equality follows from Lemma~3.1 in \citet{gui2023conformalized}, for a suitable function $g$ that is invariant to any permutation of its input.
    Further, it follows that
    \begin{align*}
      &\mathbb{P}\paren*{ \calI_1 = I_1, \dots, \calI_{n}=I_{n}, \calI^*=I_{n+1} \mid \mathcal{D}_{\mathrm{train}}, \abs*{\mathcal{D}_{\mathrm{obs}}}=n_{\mathrm{obs}}, \abs*{\mathcal{D}_{\mathrm{cal}}}=n}\\
      &\quad = \brac*{g(D_{\mathrm{cal}}\cup \cbrac*{I_{n+1}}) \cdot \prod_{(r,c) \in D_{\mathrm{cal}} \cup \cbrac*{I_{n+1}}} \frac{p_{r,c}}{1-p_{r,c}}} \cdot \frac{1-p_{I_{n+1}}}{p_{I_{n+1}}}\\
      &\quad = \brac*{g(D_{\mathrm{cal}}\cup \cbrac*{I_{n+1}}) \cdot \prod_{(r,c) \in D_{\mathrm{cal}} \cup \cbrac*{I_{n+1}}} \frac{p_{r,c}}{1-p_{r,c}}} \cdot R_{I_{n+1}},
    \end{align*}
    with $R_{I_{n+1}}$ defined as in~\eqref{eq:cmc-weights}.
    This proves that $\mathcal{I}_1, \dots, \mathcal{I}_{n}, \calI^*$ are leave-one-out exchangeable random variables by the definition in Lemma~\ref{lem:partial_quant}, with $\Bar{h}(D_{\mathrm{cal}},I_{n+1} )= R_{I_{n+1}}$.
    Therefore, the coverage guarantee of~Proposition~\ref{prop:cmc-coverage} follows directly from~Lemma~\ref{lem:partial_quant}.
  \end{proof}

\subsection{Limitations of the Unadjusted and Bonferroni Baselines} \label{app:baselines-overview}

It is not easy to construct informative simultaneous confidence regions satisfying \eqref{eq:K-coverage} and, to the best of our knowledge, there are no satisfactory alternatives to the method proposed in this paper.
In fact, standard conformal methods are designed to deal with one test point at a time, and directly aggregating separate prediction intervals into a joint confidence region is neither precise nor efficient in our context, as explained in more detail below.

Recall that the conformalized matrix completion method of \citet{gui2023conformalized}, reviewed in Appendix~\ref{app:cmc}, is designed to construct a confidence interval $\hat{C}^{\mathrm{indiv}}(X^*_{1}; \mat{M}_{\mat{X}_{\mathrm{obs}}}, \alpha)$ for one missing entry at a time, denoted as $X^*_{1}$, such that
\begin{align} \label{eq:naive-coverage-events}
  \mathbb{P}[M_{X^*_{1}} \in \hat{C}^{\mathrm{indiv}}(X^*_{1}; \mat{M}_{\mat{X}_{\mathrm{obs}}}, \alpha) ] \geq 1-\alpha
\end{align}
under a suitable sampling model for $\mat{X}_{\mathrm{obs}}$ and $X^*_{1}$.
The model for $\mat{X}_{\mathrm{obs}}$ and $X^*_{1}$ considered by \citet{gui2023conformalized} is different from ours, as they treat $n_{\mathrm{obs}}$ as random, rely on independent Bernoulli observations instead of sampling without replacement, and do not consider the possibility that the sampling weights $\mat{w}^*$ in~\eqref{model:weighted_test_generation} may be non-uniform.
However, a similar idea can be adapted to construct confidence intervals $\hat{C}^{\mathrm{indiv}}(X^*_{1}; \mat{M}_{\mat{X}_{\mathrm{obs}}},\alpha)$ for a single matrix entry $X^*_{1}$ under our sampling model~\eqref{model:sampling_wo_replacement}--\eqref{model:weighted_test_generation}, as explained in  Appendix~\ref{app:bonferroni}.
In any case, regardless of these modeling details, the limitations of the baseline approaches within our simultaneous inference context can already be understood as follows.

If the goal is to make joint predictions for a group of $K$ matrix entries, concatenating individual-level predictions clearly does not guarantee {\em simultaneous} coverage in the sense of~\eqref{eq:K-coverage}, as the errors across different coordinates tend to accumulate.
This may be seen as an instance of the prototypical multiple testing problem.
The unadjusted baseline approach essentially computes:
\begin{align} \label{eq:invalid-approach}
  \hat{\mat{C}}^{\mathrm{Unadj}}(\mat{X}^*; \mat{M}_{\mat{X}_{\mathrm{obs}}},\alpha)
  := \left( \hat{C}^{\mathrm{indiv}}(X^*_{1}; \mat{M}_{\mat{X}_{\mathrm{obs}}}, \alpha), \ldots, \hat{C}^{\mathrm{indiv}}(X^*_{K}; \mat{M}_{\mat{X}_{\mathrm{obs}}},\alpha)\right).
\end{align}
As demonstrated by Figure~\ref{fig:movielens-masked} and other synthetic experiments in Section~\ref{sec:exp-synthetic}, this approach often leads to low simultaneous coverage.

Figure~\ref{fig:movielens-preview} previewed the performance of a second baseline approach that relies on a simple but inefficient {\em Bonferroni} correction to approximately ensure simultaneous coverage.
Intuitively, this tries to (conservatively) account for the multiplicity of the problem by applying~\eqref{eq:invalid-approach} at level $\alpha/K$ instead of $\alpha$, computing
\begin{align} \label{eq:Bonferroni-approach}
  \hat{\mat{C}}^{\mathrm{Bonf}}(\mat{X}^*; \mat{M}_{\mat{X}_{\mathrm{obs}}}, \alpha)
  := \left( \hat{C}^{\mathrm{indiv}} \left(X^*_{1}; \mat{M}_{\mat{X}_{\mathrm{obs}}},\frac{\alpha}{K} \right), \ldots, \hat{C}^{\mathrm{indiv}}\left(X^*_{K}; \mat{M}_{\mat{X}_{\mathrm{obs}}},\frac{\alpha}{K} \right) \right).
\end{align}
Although a Bonferroni correction may seem reasonable at first sight, it is still unsatisfactory for at least two reasons.
Firstly, it is not rigorous because we know the $K$ missing entries indexed by $\mat{X}^*$ must belong to the same column, but this constraint cannot be easily taken into account by individual-level predictions.
Secondly, and even more crucially, the Bonferroni correction tends to be overly conservative in practice because the coverage events $M_{X^*_{k}} \in \hat{C}^{\mathrm{indiv}}(X^*_{k}; \mat{M}_{\mat{X}_{\mathrm{obs}}},\alpha)$ for different values of $k \in [K]$ are mutually dependent, since they are all affected by the same observations $\mat{X}_{\mathrm{obs}}$.
These dependencies, however, are potentially very complex.

\subsection{Implementation Details for the Baselines}\label{app:bonferroni}

To facilitate the empirical comparison with our method, which relies on the sampling model for $\mat{X}_{\mathrm{obs}}$ and $\mat{X}^*$ defined~\eqref{model:sampling_wo_replacement}--\eqref{model:weighted_test_generation}, in this paper we apply the unadjusted and Bonferroni baseline approaches described in Appendix~\ref{app:baselines-overview} based on individual-level conformal prediction intervals $\hat{C}^{\mathrm{indiv}}$ obtained as follows.
Instead of directly applying the conformalized matrix completion method of \citet{gui2023conformalized}, we repeatedly apply our method separately for each element $X^*_k$ in $\mat{X}^* = (X^*_{1}, X^*_{2}, \ldots, X^*_{K})$, imagining each time that we are dealing with a trivial group of size 1.
This provides us with individual-level prediction intervals $\hat{C}^{\mathrm{indiv}}$ that are similar in spirit to those of \citet{gui2023conformalized} but whose construction more faithfully mirrors the sampling model assumed in this paper (although they still ignore the constraint that all elements of $\mat{X}^*$ must belong to the same column).
In summary, the implementation of the unadjusted and Bonferroni baseline approaches applied in this paper is outlined by Algorithms~\ref{alg:naive-confidence-region} and~\ref{alg:bonferroni-confidence-region}, respectively.

\begin{algorithm}[!htb]
\caption{Unadjusted confidence region for multiple missing matrix entries}
    \label{alg:naive-confidence-region}
    \begin{algorithmic}[1]
        \STATE \textbf{Input}: partially observed matrix $\mat{M}_{\mat{X}_{\mathrm{obs}}}$, with unordered list of observed indices $\cD_{\mathrm{obs}}$;
        \STATE \textcolor{white}{\textbf{Input}:} test group $\mat{X}^*$; nominal coverage level $\alpha \in (0,1)$;
        \STATE \textcolor{white}{\textbf{Input}:} any matrix completion algorithm producing point estimates;
        \STATE \textcolor{white}{\textbf{Input}:} any prediction rule $\mathcal{C}$ satisfying~\eqref{eq:pred-rule-monotone} and~\eqref{eq:pred-rule-boundary};
        \STATE \textcolor{white}{\textbf{Input}:} desired number $n$ of calibration entries.
        \STATE Apply Algorithm~\ref{alg:calibration-group} with group size $K=1$ to obtain $\cD_{\mathrm{train}}$, $\cD_{\mathrm{cal}}$;
        \STATE Compute a point estimate $\hat{\mat{M}}$, looking only the observations in $\cD_{\mathrm{train}}$.
        \STATE Compute the conformity scores $S_i$, for all $i \in [n]$, with Equation \eqref{eq:scores}.
        %\STATE Compute the weights $p_i$, for all $i \in [n+1]$, with Equation \eqref{eq:gwpm-weights-k}.
        \FORALL{$k \in [K]$}
            \STATE Compute $\hat{\tau}_{\alpha,1}$ in \eqref{eq:weighted-quantile}, based on the weights $p_i$ given by \eqref{eq:gwpm-weights-k} with $\mat{x}_{n+1}=X^*_{k}$ in Section~\ref{sec:weights}.
            \STATE Compute $\hat{C}^{\mathrm{indiv}}(X^*_{k}; \mat{M}_{\mat{X}_{\mathrm{obs}}}, \alpha)$ as reviewed in Appendix~\ref{app:cmc}.
        \ENDFOR
        \STATE \textbf{Output}: Joint confidence region $\hat{\mat{C}}^{\mathrm{Unadj}}(\mat{X}^*; \mat{M}_{\mat{X}_{\mathrm{obs}}},\alpha)$ given by Equation \eqref{eq:invalid-approach}.
    \end{algorithmic}
\end{algorithm}

\begin{algorithm}[!htb]
\caption{Bonferroni-style confidence region for multiple missing matrix entries}
    \label{alg:bonferroni-confidence-region}
    \begin{algorithmic}[1]
        \STATE \textbf{Input}: partially observed matrix $\mat{M}_{\mat{X}_{\mathrm{obs}}}$, with unordered list of observed indices $\cD_{\mathrm{obs}}$;
        \STATE \textcolor{white}{\textbf{Input}:} test group $\mat{X}^*$; nominal coverage level $\alpha \in (0,1)$;
        \STATE \textcolor{white}{\textbf{Input}:} any matrix completion algorithm producing point estimates;
        \STATE \textcolor{white}{\textbf{Input}:} any prediction rule $\mathcal{C}$ satisfying~\eqref{eq:pred-rule-monotone} and~\eqref{eq:pred-rule-boundary};
        \STATE \textcolor{white}{\textbf{Input}:} desired number $n$ of calibration entries.
        \STATE Apply Algorithm~\ref{alg:calibration-group} with group size $K=1$ to obtain $\cD_{\mathrm{train}}$, $\cD_{\mathrm{cal}}$;
        \STATE Compute a point estimate $\hat{\mat{M}}$, looking only the observations in $\cD_{\mathrm{train}}$.
        \STATE Compute the conformity scores $S_i$, for all $i \in [n]$, with Equation \eqref{eq:scores}.
        %\STATE Compute the weights $p_i$, for all $i \in [n+1]$, with Equation \eqref{eq:gwpm-weights-k}.
        \FORALL{$k \in [K]$}
            \STATE Compute $\hat{\tau}_{\frac{\alpha}{K},1}$ in \eqref{eq:weighted-quantile}, based on the weights $p_i$ given by \eqref{eq:gwpm-weights-k} with $\mat{x}_{n+1}=X^*_{k}$ in Section~\ref{sec:weights}.
            \STATE Compute $\hat{C}^{\mathrm{indiv}}(X^*_{k}; \mat{M}_{\mat{X}_{\mathrm{obs}}}, \frac{\alpha}{K})$ as reviewed in Appendix~\ref{app:cmc}.
        \ENDFOR
        \STATE \textbf{Output}: Joint confidence region $\hat{\mat{C}}^{\mathrm{Bonf}}(\mat{X}^*; \mat{M}_{\mat{X}_{\mathrm{obs}}},\alpha)$ given by Equation \eqref{eq:Bonferroni-approach}.
    \end{algorithmic}
\end{algorithm}

% \subsection{Variation of the Classical Laplace Method} \label{app:laplace-variation}

\subsection{Review of the Classical Laplace Method} \label{app:laplace-review}

This section provides a concise review of the classical version of Laplace's method, as detailed in \citet{butler2007saddlepoint}.
This method is a powerful tool for approximating analytically intractable integrals of the form $\int_{a}^{b} e^{nf(x)} h(x) dx$, where the function $f$ is sufficiently well-behaved and smooth, with a unique global maximum at an interior point $x_0 \in (a,b)$, the function $h$ is positive and does not vary significantly near $x_0$, and  $n$ is a relatively large constant.
The method hinges on the principle that this integral's value is predominantly determined by a small region around the point where $f$ achieves its maximum.
This idea is explained in more detail and motivated precisely below.

Let \(f(x)\) be a twice continuously differentiable function on an interval \((a, b)\), and assume there exists a unique global maximum at an interior point \(x_0 \in (a, b)\), such that \(f(x_0) = \max_{x \in (a, b)} f(x)\) and \(f''(x_0) < 0\).
Suppose $h$ is a function that varies slowly around $x_0$ and is such that \(h(x)>0\) for all $x \in (a,b)$.
Then, Laplace's approximation involves replacing the integral $\int h(x) e^{nf(x)} dx$ with
\begin{equation}
    \int_{a}^{b} e^{n f(x)} h(x) \,dx
    \approx e^{n f(x_0)} h(x_0) \sqrt{\frac{2\pi}{-n f''(x_0)}}.
\end{equation}
A standard mathematical justification for this approximation starts by proving that, under suitable technical assumptions on $f$ and $h$ in the spirit of the intuitive conditions outlined above,
\begin{equation} \label{eq:laplace-limit-classical}
\lim_{n \to \infty} \frac{\int_{a}^{b} e^{n f(x)} h(x) \,dx}{e^{n f(x_0)} h(x_0) \sqrt{\frac{2\pi}{n \left(-f''(x_0)\right)}}} = 1.
\end{equation}

The classical proof of~\eqref{eq:laplace-limit-classical} consists of three high-level steps:\vspace{-0.5em}
\begin{enumerate}\setlength\itemsep{0em}
    \item \textbf{Local second-order approximation:} Approximate \(f(x)\) near \(x_0\) using a second-order Taylor expansion: 
    $f(x) \approx f(x_0) + \frac{1}{2} f''(x_0) (x - x_0)^2$.

    \item \textbf{Integral transformation:} Standardize the quadratic term in the integral to apply results from Gaussian integral analysis.

    \item \textbf{Asymptotic evaluation:} Assess the integral in the standardized coordinates to achieve the asymptotic equivalence in~\eqref{eq:laplace-limit-classical}.
\end{enumerate}

The proof of Theorem \ref{thm:laplace-independent}, presented in Appendix~\ref{app:laplace-theory}, follows a similar high-level strategy, although its details are significantly more involved because our integral of interest in~\eqref{eq:prob-obs-approx-i} cannot be directly written as $\int_{a}^{b} e^{n f(x)} h(x) \,dx$ for some functions $f,g$.

\subsection{Consistency of the Generalized Laplace Approximation} \label{sec:laplace}

In this section, we provide a justification of~\eqref{eq:relation-obs}. For simplicity, but without much loss of generality, this theorem relies on some additional technical assumptions, which will be justified in our context towards the end of this section.
% %
% This result is presented informally here for simplicity, but a formal statement can be found in Appendix~\ref{app:laplace-theory}.

% \begin{theorem}[Informal statement of Theorem \ref{thm:laplace-independent}]
% \label{thm:laplace-informal}
% Let $\{v_i\}_{i=1}^\infty$ denote a sequence of i.i.d.~random variables from some distribution $F$ supported on $(0,1)$, and  $\{x_i\}_{i=1}^\infty$ an independent Bernoulli sequence, with $x_i \sim \text{Bernoulli}(v_i)$. Define $\delta_m = \sum_{i=1}^m (1-x_i)v_i$ and
%     \begin{align}\label{def:phi-independent}
%     \Phi_m(\tau) := h_m\delta_m \tau^{h_m\delta_m-1} \prod_{i=1}^m  \left(1-\tau^{h_m v_i}\right)^{x_i},
%     \end{align}
%     where $h_m$ is such that $\tau_h :=\argmax_{\tau \in (0,1)} \Phi_m(\tau)=1/2$. Define also $\phi_m(\tau) \coloneqq \log \Phi_m(\tau)$.
%     Then, for a sequence $f_m$ bounded away from 0 and satisfying certain smoothness conditions,
%         $\lim_{m \to \infty} [\int_0^1 f_m(\tau) \Phi_m(\tau) \, d\tau] / [ f_m\left(\tau_h \right) \cdot \Phi_m\left(\tau_h \right) \sqrt{-2\pi/\phi_m''(1/2)} ] = 1$ almost surely.
% \end{theorem}

% We begin by stating a formal version of Theorem~\ref{thm:laplace-informal}, the result providing the motivation to apply the Laplace method to  Equation~\eqref{eq:relation-obs}.

\begin{theorem} \label{thm:laplace-independent}
    Let $\{w_i\}_{i=1}^\infty$ be a sequence of i.i.d.~random variables drawn from a distribution $F$ with support on the open interval $(0,1)$.
    Consider a sequence of mutually independent Bernoulli random variables $\{x_i\}_{i=1}^\infty$, where each $x_i \overset{\text{ind.}}{\sim} \mathrm{Bernoulli}(w_i)$.
      Define $\delta_m = \sum_{i=1}^m (1-x_i)w_i$ and
    \begin{align}\label{def:phi-independent}
    \Phi_m(\tau) := h_m\delta_m \tau^{h_m\delta_m-1} \prod_{i=1}^m \left(1-\tau^{h_m w_i}\right)^{x_i}.
    \end{align}
    Above, $h_n$ is the unique root of the function
    \begin{align}  \label{eq:laplace-hn-root}
        z(h) \coloneqq \frac{\phi_m'\left(\frac{1}{2}\right)}{2h} = \delta_m - \frac{1}{h} - \sum_{i=1}^m \frac{x_i w_{i}}{2^{h w_{i}} - 1}
    \end{align}
    in the interval $[\delta_m, \infty)$, where $\phi_m(\tau)$ is the logarithm of $\Phi_m(\tau)$, namely
    \begin{align}
    \phi_m(\tau) \coloneqq \log \Phi_m(\tau) = \log(h_m) + \log(\delta_m) + (h_m\delta_m-1)\log(\tau) + \sum_{i=1}^m x_i \log\left(1-\tau^{h_m w_i}\right).
    \end{align}
    Then, $\argmax_{\tau \in [0,1]} \Phi_m(\tau) = \frac{1}{2}$.

    Further, consider a sequence of functions $\{f_m\}$, where each $f_m \in C^1(0, 1)$ and $f_m\left(\frac{1}{2}\right) > \epsilon_0$, for some constant $\epsilon_0 >0$ and all $m$. Suppose there exists some $M > 0$ such that $|f_m'(x)| \leq M$ for all $x \in (0,1)$ and for all $m$. Then, it holds that
    \begin{align}
        \int_0^1 f_m(\tau) \Phi_m(\tau) \, d\tau \sim f_m\left(\frac{1}{2}\right) \cdot \Phi_m\left(\frac{1}{2}\right) \sqrt{\frac{-2\pi}{\phi_m''\left(\frac{1}{2}\right)}} \text{ almost surely as } m \to \infty,
    \end{align}
    or equivalently,
    \begin{align}
        \lim_{m \to \infty} \frac{\int_0^1 f_m(\tau) \Phi_m(\tau) \, d\tau}{f_m\left(\frac{1}{2}\right) \cdot \Phi_m\left(\frac{1}{2}\right) \sqrt{\frac{-2\pi}{\phi_m''\left(\frac{1}{2}\right)}}} = 1 \text{ almost surely}.
    \end{align}
\end{theorem}

The proof of Theorem~\ref{thm:laplace-independent} appears in Appendix~\ref{app:laplace-theory}.
To relate this result to the Laplace approximation described in Section~\ref{sec:implementation-weights}, compare $\Phi(\tau; h)$ in~\eqref{eq:integral-def-phi} with  $\Phi_m(\tau)$ in~\eqref{def:phi-independent}.
Given a mapping $\sigma:[n_r n_c] \mapsto [n_r] \times [n_c]$, we can express $\Phi(\tau; h)$ as
$$
\Phi(\tau; h) = h\delta \tau^{h\delta-1} \prod_{(r,c) \in D_{\mathrm{obs}}} \paren*{1-\tau^{h w_{r,c}}} = h\delta \tau^{h\delta-1} \prod_{i=1}^{n_r n_c}  \paren*{1-\tau^{h w_{\sigma(i)}}}^{\mathbbm{1}\{\sigma(i) \in D_{\mathrm{obs}}\}}.
$$
Thus, the discrepancy between $\Phi(\tau; h)$ and $\Phi_{n_r n_c}(\tau)$ can be traced to the different sampling models for the matrix observations and the $x_i$ variables in Theorem~\ref{thm:laplace-independent}.
In Theorem~\ref{thm:laplace-independent}, the observations follow independent Bernoulli distributions, whereas in~\eqref{model:weighted_test_generation} they are sampled without replacement.
These views can be reconciled as follows.
Sampling without replacement is natural for the problem studied in this paper, but it would make the proof of Theorem \ref{thm:laplace-independent} too complicated.
Nevertheless, these two models are qualitatively consistent.
For example, if the weights $\mat{w}$ in~\eqref{model:sampling_wo_replacement} are constant, our model corresponds to that of Theorem~\ref{thm:laplace-independent} with $v_i \equiv v$, for some $v\in (0,1)$, after conditioning on $n_{\text{obs}}$.

\section{Additional Methodological Details} \label{app:implementation}

% \subsection{Weighted Sampling Without Replacement} \label{app:weighted-sampling}

% The distribution $\Psi(m, \mathcal{X}, \mat{w})$ defined in Section~\ref{sec:preliminaries} generally corresponds to weighted random sampling without replacement of $m \geq 1$ distinct elements (``balls") from a finite dictionary (``urn"), denoted as $\mathcal{X}$, with $|\mathcal{X}| \geq m$, according to an unnormalized set of weights $\mat{w} = (w_x)_{x \in \mathcal{X}}$.
% More precisely, if $\mat{X} = (X_1, \ldots, X_m)$ is a random sample from $\Psi(m, \mathcal{X}, \mat{w})$ and $\mat{x} = (x_1, \ldots, x_m) \in \mathcal{X}^m$ is a collection of $m$ {\em distinct} values from $\mathcal{X}$, this model corresponds to 
% \begin{align*}
%     \P{\mat{X} = \bm{x}}
%     & = \frac{w_{x_1}}{\sum_{x \in \mathcal{X}} w_x} \cdot \frac{w_{x_2}}{\sum_{x \in \mathcal{X}} w_x - w_{x_1}} \cdot \ldots \cdot \frac{w_{x_m}}{\sum_{x \in \mathcal{X}} w_x - \sum_{i=1}^{m-1} w_{x_i}}.
% \end{align*}
% In the special case where all weights are the same, this model simply reduces to random sampling without replacement,
% \begin{align*}
%     \P{\mat{X} = \bm{x}}
%     & = \frac{1}{|\mathcal{X}|} \cdot \frac{1}{|\mathcal{X}|-1} \cdot \ldots \cdot \frac{1}{|\mathcal{X}|-(m-1)}
%     = \frac{1}{m! { |\mathcal{X}| \choose m}}.
% \end{align*}

\subsection{Corner Cases in Sampling the Test Group}
\label{app:corner-cases}
When defining model~\eqref{model:weighted_test_generation}, we assume all columns have much more than $K$ missing entries. This assumption typically holds for sparsely observed matrices. However, for completeness, we address the corner cases where some columns contain fewer than $K$ missing entries.

To ensure the sampling model for $\mat{X}^*$ is well-defined, we need to exclude columns with fewer than $K$ missing entries. We define a {\em pruned} set of missing indices $\bar{\cD}_{\mathrm{miss}} \subseteq \cD_{\mathrm{miss}}$ as
$\bar{\cD}_{\mathrm{miss}} = \cbrac{(r,c) \in \cD_{\mathrm{miss}} : n^c_{\mathrm{miss}} \geq K}$,
where $n^c_{\mathrm{miss}} = \abs{\cbrac{(r',c') \in \cD_{\mathrm{miss}} :  c' = c}}$ is the number of missing entries in column $c$.
Then, we assume $\mat{X}^*$ is sampled as
\begin{align}\label{model:corner-weighted_test_generation}
\begin{split}
\mat{X}^* \mid \cD_{\mathrm{obs}}, \cD_{\mathrm{miss}} & \sim \Psi^{\text{col}}(K, \bar{\cD}_{\mathrm{miss}}, \mat{w}^*),
\end{split}
\end{align}
where $\Psi^{\text{col}}$ is a constrained version of $\Psi$ that samples indices belonging to the same column.
This distribution is equivalent to the following sequential sampling procedure:
\begin{align}\label{model:corner-weighted_test_generation_sequential}
  \begin{split}
    X^*_{1} \mid \cD_{\mathrm{obs}}, \cD_{\mathrm{miss}} & \sim \Psi(1, \bar{\cD}_{\mathrm{miss}}, \mat{w}^*) \\
    X^*_{2} \mid \cD_{\mathrm{obs}}, \cD_{\mathrm{miss}}, X^*_{1} & \sim \Psi(1, \bar{\cD}_{\mathrm{miss}} \setminus \{X^*_{1}\}, \tilde{\mat{w}}^*), \\
    & \; \; \vdots \\
    X^*_{K} \mid \cD_{\mathrm{obs}}, \cD_{\mathrm{miss}}, X^*_{1}, \ldots, X^*_{K-1} & \sim \Psi(1, \bar{\cD}_{\mathrm{miss}} \setminus \{X^*_{1}, \ldots, X^*_{K-1}\}, \tilde{\mat{w}}^*),
  \end{split}
\end{align}
where the weights $\tilde{\mat{w}}^*$ are given by $\tilde{w}^*_{r,c} = w^*_{r,c} \I{ c = X^*_{1,2}}$ and $X^*_{1,2}$ is the column of $X^*_{1}$.

The sequential sampling procedure in~\eqref{model:weighted_test_generation_sequential} requires at least one column with more than $K$ unobserved entries. This is always the case as long as $n_{\mathrm{obs}} < n_c (n_r - K + 1)$, which is an extremely mild assumption in real applications.

When $\abs{\cbrac{(r',c') \in \cD_{\mathrm{miss}}\mid c'=c}} \gg K$ for all $c \in [n_c]$, $\bar{\cD}_{\mathrm{miss}}=\cD_{\mathrm{miss}}$, hence model~\eqref{model:weighted_test_generation} becomes a simplified special case of the more comprehensive model~\eqref{model:corner-weighted_test_generation}, which is designed to handle corner cases. This simplification allows us to avoid additional technical notation and facilitates a clearer presentation in the main section. Nevertheless, we will apply the general model~\eqref{model:corner-weighted_test_generation} in all relevant theoretical results to ensure robustness and completeness.

\subsection{Practical Computation of the Conformity Scores} \label{sec:scores}

As detailed in~Section~\ref{sec:method-outline}, our method allows flexibility in the choice of the prediction rule $\mathcal{C}$, which uniquely determines the conformity scores. In this section, we explore three practical options for the prediction rules and their respective conformity scores.

\subsubsection{Hyper-Cubic Confidence Regions}

An intuitive prediction rule, introduced in Section~\ref{sec:method-outline}, is:
$$\mathcal{C}(\mat{x}^*, \tau, \hat{\mat{M}}) =  \left( \hat{\mat{M}}_{x^*_1} \pm \frac{\tau}{1-\tau}, \ldots, \hat{\mat{M}}_{x^*_K} \pm \frac{\tau}{1-\tau} \right),$$
with the parameter $\tau$ taking value in $[0,1]$. 
Note that this rule leads to hyper-cubic confidence regions, with constant widths for all users in a group. 

The conformity scores corresponding to this rule can be written explicitly, for any $i \in [n]$, as:
\begin{align*}
  S_i 
%  &\coloneqq \inf \cbrac*{\tau \in [0,1]: \mat{M}_{\mat{X}^{\mathrm{cal}}_i} \in \paren*{\hat{\mat{M}}_{X^{\mathrm{cal}}_{i,1}}\pm \frac{\tau}{1-\tau}, \ldots, \hat{\mat{M}}_{X^{\mathrm{cal}}_{i,K}}\pm \frac{\tau}{1-\tau}} } \\
  & = \max\left\{
    \frac{\abs{\hat{\mat{M}}_{X^{\mathrm{cal}}_{i,1}}-\mat{M}_{X^{\mathrm{cal}}_{i,1}}}}{1+\abs{\hat{\mat{M}}_{X^{\mathrm{cal}}_{i,1}}-\mat{M}_{X^{\mathrm{cal}}_{i,1}}}} ,\ldots, \frac{\abs{\hat{\mat{M}}_{X^{\mathrm{cal}}_{i,K}}-\mat{M}_{X^{\mathrm{cal}}_{i,K}}}}{1+\abs{\hat{\mat{M}}_{X^{\mathrm{cal}}_{i,K}}-\mat{M}_{X^{\mathrm{cal}}_{i,K}}}} \right\}.
\end{align*}

\textbf{Remark.} The function $x\mapsto x/(1+x)$ is an strictly increasing function on $x \geq 0$. Therefore, we can equivalently define the prediction set as the following. Let
\begin{align} \label{eq:score_usual_def}
    \Tilde{S}_i=\max \cbrac*{
\abs{\hat{\mat{M}}_{X^{\mathrm{cal}}_{i,1}}-\mat{M}_{X^{\mathrm{cal}}_{i,1}}}, \ldots, \abs{\hat{\mat{M}}_{X^{\mathrm{cal}}_{i,K}}-\mat{M}_{X^{\mathrm{cal}}_{i,K}}} }
\end{align}
and define the alternate confidence set as
$$\label{eq:prediction-rule-usual-def}
\mathcal{C}'(\mat{x}^*, \tau, \hat{\mat{M}}) =  \paren*{\hat{\mat{M}}_{x^*_1} \pm \tau, \ldots, \hat{\mat{M}}_{x^*_K} \pm \tau},
$$
with $\tau$ taking value in $[0,+\infty)$.
The expression in~\eqref{eq:score_usual_def} is more closely related to the typical notation in the conformal inference literature; e.g., see \citet{lei2018distribution}.

\subsubsection{Hyper-Rectangular Confidence Regions}

An alternative type of prediction rule, yielding intervals of varying lengths for different users, involves scaling the hyper-cube defined in~\eqref{eq:prediction-rule-usual-def}. 
This modification may be particularly useful in applications involving count data with wide ranges, where the variance may be expected to increase in proportion to the observed values. 
We define this linearly-scaled prediction rule as
\begin{align}\label{eq:prediction-rule-rectangle}
\mathcal{C}(\mat{x}^*, \tau, \hat{\mat{M}}) =  \paren*{\hat{\mat{M}}_{x^*_1}\pm \abs{\hat{\mat{M}}_{x^{*}_{1}}} \tau, \ldots, \hat{\mat{M}}_{x^*_K}\pm \abs{\hat{\mat{M}}_{x^{*}_{K}}} \tau},
\end{align}
which leads to confidence regions in the shape of a hyper-rectangle. The corresponding scores are:
\begin{align*}
  S_i
%  &\coloneqq \inf \cbrac*{\tau \in \mathbb{R}: \mat{M}_{\mat{X}^{\mathrm{cal}}_i} \in \paren*{\hat{\mat{M}}_{X^{\mathrm{cal}}_{i,1}}\pm  \abs{\hat{\mat{M}}_{X^{\mathrm{cal}}_{i,1}}}\tau, \ldots, \hat{\mat{M}}_{X^{\mathrm{cal}}_{i,K}}\pm \abs{\hat{\mat{M}}_{X^{\mathrm{cal}}_{i,K}}} \tau} } \\
  &=\max\left\{  \frac{\abs{\hat{\mat{M}}_{X^{\mathrm{cal}}_{i,1}}-\mat{M}_{X^{\mathrm{cal}}_{i,1}}}}{\abs{\hat{\mat{M}}_{X^{\mathrm{cal}}_{i,1}}}}, \ldots,  \frac{\abs{\hat{\mat{M}}_{X^{\mathrm{cal}}_{i,K}}-\mat{M}_{X^{\mathrm{cal}}_{i,K}}}}{\abs{\hat{\mat{M}}_{X^{\mathrm{cal}}_{i,K}}}} \right\}.
\end{align*}

\subsubsection{Hyper-Spherical Confidence Regions}

The prediction rules described above all result in confidence regions with a hyper-rectangular shape. Alternatively, one can construct a confidence region with a hyper-spherical shape using the following prediction rule, where $\|\cdot \|_2$ represents the Euclidean norm:
\begin{align}\label{eq:prediction-rule-sphere}
    \mathcal{C}(\mat{x}^*, \tau, \hat{\mat{M}}) =  \lcbrac{\mat{x} \in \R^K : \| \mat{x}-\hat{\mat{M}}_{\mat{X}^{\mathrm{cal}}_i} \|_2 \leq \tau }.
\end{align}
The corresponding conformity scores are
\begin{align*}
    S_i\coloneqq \inf \cbrac*{\tau \in \mathbb{R}: \| \mat{M}_{\mat{X}^{\mathrm{cal}}_i}-\hat{\mat{M}}_{\mat{X}^{\mathrm{cal}}_i} \|_2 \leq \tau } = \| \mat{M}_{\mat{X}^{\mathrm{cal}}_i}-\hat{\mat{M}}_{\mat{X}^{\mathrm{cal}}_i} \|_2.
\end{align*}
Note that replacing the Euclidean norm with the max norm in \eqref{eq:prediction-rule-sphere} recovers the hyper-cubic prediction rule.

The concept of a hyper-spherical confidence region is quite rare in the conformal inference literature, where the majority of existing methods focus on constructing confidence sets for a single test point individually. However, when aiming to provide simultaneous coverage for multiple entries, it becomes possible to develop confidence regions of varying geometric shapes.

\subsection{Efficient Evaluation of the Conformalization Weights} \label{app:scaling-parameter}

We discuss in more detail here the choice of scaling parameter $h$ for the function $\Phi(\tau; h)$ defined in Equation \eqref{eq:integral-def-phi} (Section~\ref{sec:implementation-weights}).
This free parameter controls the location of $\tau_h = \argmax_{\tau \in (0,1)} \Phi(\tau;h)$.
Since our Laplace approximation hinges on $\tau_h$ being not too close to the integration boundary, an intuitive and effective choice is to set $h$ so that $\tau_h = 1/2$; e.g., see \citet{fog2007wnchypg}.
We explain below how to achieve this using the Newton-Raphson algorithm.

Recall Lemma~\ref{lemma:phi-max}, which tells us that the function \(\Phi(\tau; h)\) has a unique stationary point with respect to $\tau$ at some value $\tau_h \in (0,1)$, and that this stationary point is a global maximum.
Therefore, since the function \(\Phi(\tau; h)\) is smooth, the Newton-Raphson algorithm can be applied as follows to find a value of $h>0$ such that $\tau_h = 1/2$.
Define 
\begin{align*}
    z(h) := \frac{\phi'(\frac{1}{2}; h)}{2h}=\delta-\frac{1}{h}-\sum_{(r,c) \in D_{\mathrm{obs}}}\frac{ w_{r,c}}{2^{h w_{r,c}}-1},
\end{align*}
and note that this function is monotonically increasing in $h$.
Then it suffices to find $h$ such that $z(h) =0$. 
Note that $z(h)$ is smooth, and $z'(h)>0$, $z''(h)<0$  for $h \in [1/\delta, \infty)$. 
It is also clear that the solution of $z(h) = 0$ must be greater than $1/\delta$.
Further, $z(h)$ has a unique root in the interval $[1/\delta, \infty)$ because $z(1/\delta) <0$ and $\lim_{h\rightarrow \infty}z(h)=\delta >0$.

Thus, it follows from Theorem 2.2 in \citet{atkinson1989numerical} that the Newton-Raphson algorithm will converge to the root $\tau_h$ quadratically, for any starting point within the interval $[1/\delta, \infty)$. In practice, one can choose $1/\delta$ as the starting point.

The time complexity of the Newton-Raphson iteration depends on the desired precision level. If the tolerable error is a predetermined small constant, the iteration terminates after a constant number of updates due to quadratic convergence. Evaluating $z(h)$ and $z'(h)$ at any given $h$ requires $\mathcal{O}(n_{\mathrm{obs}})$. Hence, solving $z(h)=0$ takes $\mathcal{O}(n_{\mathrm{obs}})$.

\subsection{Computational Shortcuts and Complexity Analysis}\label{app:complexity}

The SCMC method described in this paper can be implemented efficiently and is able to handle completion tasks involving large matrices.
Its nice scaling properties are demonstrated in this section, which summarizes the results of an analysis of the computational complexity of different components of Algorithm~\ref{alg:simultaneous-confidence-region}.
We refer to the subsections below for the details behind this analysis and an explanation of the underlying computational shortcuts with which all redundant operations are streamlined.

In summary, the cost of producing a joint confidence region for a test group $\mat{X}^*$ of size $K$ using Algorithm~\ref{alg:simultaneous-confidence-region} is $\mathcal{O}(T+n_r n_c +n(K +\log n))$, where $T$ denotes the fixed cost of training the black-box matrix completion model based on $\mat{M}_{\mat{X}_{\mathrm{obs}}}$ and $n$ is the number of calibration groups. 
Further, it is possible to recycle redundant calculations when constructing simultaneous confidence regions for $m$ distinct test groups $\mat{X}^*$, as explained in Appendix~\ref{app:complexity}. Therefore, the overall cost of obtaining $m$ distinct confidence regions for $m$ different groups of size $K$ is only $\mathcal{O}(T+n_r n_c + n(mK + \log n))$.
See Table~\ref{tab:comp-complexity} for a summary of these results.

\begin{table}[!htb]
\caption{Computational efficiency of different components of SCMC. } \label{tab:comp-complexity}
\centering
\begin{tabular}{lll}
\toprule
Module & Cost for one test group & Cost for $m$ test groups \\
\midrule
Overall (Algorithm~\ref{alg:simultaneous-confidence-region}) & $\mathcal{O}(T+n_r n_c +n(K +\log n))$ & $\mathcal{O}(T+n_r n_c + n(mK + \log n))$ \\
Algorithm~\ref{alg:calibration-group} & $\mathcal{O}(n_c n_r+nK)$ & $\mathcal{O}(n_c n_r+nK)$ \\
\bottomrule
\end{tabular}
\end{table}

Notably, the cost $\mathcal{O}(T)$ of fitting the matrix completion model depends on the specific algorithm used, as our methodology is model-agnostic.  
However, for reference, we provide some concrete examples of $\mathcal{O}(T)$ for commonly used approaches below.  
For example, defining $n_{\max}=\max\{n_r, n_c\}$ and $r$ as the hypothesized rank of the rating matrix, the computational complexity of common low-rank matrix completion methods is as follows:
\begin{itemize}
    \item Convex nuclear norm minimization \citep{candes_exact_2008}: $\mathcal{O}(n_{\max}^3)$.
    \item Singular Value Thresholding (SVT) \citep{svt_cai_2010}: $\mathcal{O}(rn_{\max}^2)$.
    \item Alternating Least Squares (ALS) \citep{hu_cf_2008}: $\mathcal{O}(rn_{\max}^2)$.
\end{itemize}
Thus, in practice, the computational bottleneck is typically the matrix completion algorithm itself rather than our inference framework.

\subsubsection{Evaluation of the Conformalization Weights}\label{app:approx-computation}

Evaluating the simplified weights $\bar{p}_i$ in~\eqref{eq:gwpm-weights-k-fast}  only involves arithmetic operations and can be carried out for all $i \in [n+1]$ at a total computational cost roughly of order $\mathcal{O}(n_r n_c+nK)$.
To understand this, first note that computing $\eta_i(\tau_h)$ defined in \eqref{eq:def-integral-tau}, for all $i \in [n+1]$ with any given $\tau_h$ and $h$, has cost $\mathcal{O}(nK)$; and finding the optimal values of $\tau_h$ and $h$ has cost $\mathcal{O}(n_{\mathrm{obs}})$, or equivalently no worse than $\mathcal{O}(n_r n_c)$, as explained in Section~\ref{app:scaling-parameter}.

Next, evaluating $\tilde{w}^*_{x_{i,1}} \cdot \ldots \cdot \tilde{w}^*_{x_{i,K}}$ for all $i \in [n+1]$ has cost $\mathcal{O}(n_r n_c+nK)$, because the constant $\sum_{(r,c)\in D_{\mathrm{miss}}}w^{*}_{r,c}$ in the denominators of~\eqref{eq:w-weights-i1} and~\eqref{eq:w-weights-ik} can be pre-computed at cost $\mathcal{O}(n_{r}n_c)$, while the remaining terms in \eqref{eq:w-weights-i1} and~\eqref{eq:w-weights-ik} can be evaluated at cost $\mathcal{O}(K)$ separately for each $i \in [n+1]$.

The cost of evaluating the term within the square brackets in~\eqref{eq:gwpm-weights-k-fast} for all $i \in [n+1]$ is $\mathcal{O}(n_{r} + nK)$. This is achieved by pre-computing factorials up to $n_r$ since $n^{c}_{\mathrm{obs}}$ is upper-bounded by $n_r$ for any $c \in [n_c]$. Then for each $i \in [n+1]$, computing binomial coefficients, given the pre-computed factorials, takes constant time, and the remaining term in the brackets requires $\mathcal{O}(K)$. 
Putting everything together, the conformalization weights in~\eqref{eq:gwpm-weights-k-fast} has cost $\mathcal{O}(n_{r} + nK)$ for all $i \in [n+1]$.

\subsubsection{Computational Cost Analysis of Algorithm~\ref{alg:simultaneous-confidence-region}}
\textbf{Analysis for a single test group.}
The cost of computing a confidence region for a single test group is $\mathcal{O}(T+n_r n_c +n(K +\log n))$, as shown below.
\begin{itemize} \setlength\itemsep{0em}
    \item Training the black-box matrix completion model has cost $\mathcal{O}(T)$.
    \item Post training, the cost of computing scores $S_i$ for all $i \in [n]$ is $\mathcal{O}(nK)$.
    \item The cost of computing $p_i$ for all $i \in [n+1]$ is $\mathcal{O}(n_r n_c+nK)$; see Section~\ref{app:approx-computation}.
    \item After the conformalization weights are computed, the cost of computing $\hat{\tau}_{\alpha,K}$ is $\mathcal{O}(n\log n)$. This is because sorting the scores $S_i$ for all $i \in [n]$ has a worst-time cost of $\mathcal{O}(n\log n)$, while it takes $\mathcal{O}(n)$ to find the weighted quantile based on  $p_i$ and the sorted scores.
\end{itemize}
Therefore, the overall cost is $\mathcal{O}(T+n_r n_c +n(K +\log n))$.

\textbf{Analysis for $m$ distinct test groups.}
The cost of computing confidence regions for $m$ distinct test groups is $\mathcal{O}(T+n_r n_c + n(\log n+mK))$, as shown below.
\begin{itemize}  \setlength\itemsep{0em}
    \item Training the black-box matrix completion model has cost $\mathcal{O}(T)$, since the model only needs to be trained once.
    \item The cost of computing conformity scores $S_i$ for all $i \in [n]$ is $\mathcal{O}(nK)$, since the calibration groups are the same for any new test group.
    \item 
    The cost of computing $p_i$ for all $i \in [n+m]$ is $\mathcal{O}(n_r n_c+mnK)$; see Section~\ref{app:approx-computation}.
    \item After the conformalization weights are computed, the cost of computing the confidence sets for all $m$ test groups is $\mathcal{O}(n\log n+nm)$. Sorting the scores $S_i$ for all $i \in [n]$ has a worst-time cost of $\mathcal{O}(n\log n)$, which only needs to be performed once. For any $j\in [m]$, it takes $\mathcal{O}(n)$ to find the weighted quantile given weights $\cbrac*{p_1, \dots, p_n, p_{n+j}}$ and the sorted scores.
\end{itemize}
Therefore, the overall cost is $\mathcal{O}(T+n_r n_c + n(\log n+mK))$.

\subsubsection{Computational Cost Analysis of Algorithm~\ref{alg:calibration-group}}
\begin{itemize} \setlength\itemsep{0em}
    \item For each column, the cost of computing $m^c:=n^c_{\mathrm{obs}}-\lfloor n^c_{\mathrm{obs}}/K \rfloor < K$ is $\mathcal{O}(n_r)$, and the cost of sampling $m^c$ indices uniformly at random is $\mathcal{O}(K)$. Hence the cost of sampling the pruned indices for all columns is $\mathcal{O}(n_c(n_r+K))$, which simplifies to $\mathcal{O}(n_c n_r)$ by the fact that $K<n_r$.
    \item Initializing $\cD_{\mathrm{avail}}$ given the pruned indices $\cD_{\mathrm{prune}}$ has cost of $\mathcal{O}(n_c n_r)$.
    \item After $\cD_{\mathrm{avail}}$ is initialized, the cost of sampling the $i$th calibration group (and updating $\mathcal{D}_{\mathrm{cal}}$ and $\cD_{\mathrm{avail}}$) is $\mathcal{O}(K)$, for each $i \in [n]$. Hence sampling all $n$ calibration groups takes $\mathcal{O}(nK)$.
\end{itemize}
Therefore, Algorithm~\ref{alg:calibration-group} has time complexity of $\mathcal{O}(n_c n_r+nK)$, and it does not need to be repeatedly applied when dealing with distinct groups involving the same matrix.

\subsection{Estimation of the Sampling Weights}\label{app:missingness-estimation}

We describe here a method, inspired by \citet{gui2023conformalized},  to estimate empirically the sampling weights $\mat{w}$ for our sampling model in~\eqref{model:sampling_wo_replacement}, leveraging the available matrix observations indexed by $\cD_{\mathrm{obs}}$. 
In general, this estimation problem is made feasible by introducing the assumption that $\mat{w}$ has some lower-dimensional structure that can be summarized for example by a parametric model.
The approach suggested by \citet{gui2023conformalized} assumes that the weight matrix $\mat{w} \in \mathbb{R}^{n_r\times n_c}$ is low-rank. 
For simplicity, we follow the same approach here, although our framework could also accommodate alternative estimation techniques in situations where different modeling assumptions about $\mat{w}$ may be justified.

Suppose the sampling weights follow the parametric model
\begin{align*}
    \log\paren*{\frac{w_{r,c}}{1-w_{r,c}}}=A_{r,c},
\end{align*}
where $\mat{A} \in \mathbb{R}^{n_r\times n_c}$ is a matrix with rank $\rho$ and bounded infinity norm; i.e., $||\mat{A}||_\infty \leq \nu$, for some pre-defined constant $\nu \in \mathbb{R}$.
Then, if each matrix entry $(r,c)$ is independently observed (i.e., included in $\cD_{\mathrm{obs}}$) with probability $w_{r,c}$, i.e.,
\begin{align} \label{eq:estimation-bernoulli}
  \mathbb{I}\left[ (r,c) \in \cD_{\mathrm{obs}} \right] \overset{\text{ind.}}{\sim} \text{Bernoulli}(w_{r,c}),
\end{align}
then the log-likelihood of $\mat{A}$ can be written as
\begin{align} \label{eq:likelihood}
    \mathcal{L}_{\cD_{\mathrm{obs}}}(\mat{A})&=\sum\limits_{(r,c) \in \cD_{\mathrm{obs}}} \log(l(A_{r,c})) + \sum\limits_{(r,c) \in [n_r]\times[n_c] \setminus
     \cD_{\mathrm{obs}}} \log(1-l(A_{r,c})),
\end{align}
where $l(t) = \paren*{1+\exp(-t)}^{-1}$.
This suggests estimating  $\mat{A}$ by solving
\begin{align*}
  \hat{\mat{A}} = \argmax\limits_{\mat{A} \in \mathbb{R}^{n_r\times n_c}} & \mathcal{L}_{\cD_{\mathrm{obs}}}(\mat{A}) \\
  \text{subject to: } & \norm{\mat{A}}_* \leq \nu \sqrt{\rho n_r n_c}, \\
                                                                          & \norm{\mat{A}}_{\infty} \leq \nu,
\end{align*}
where $\norm{\cdot}_*$ is the nuclear norm. 
Finally, having obtained $\hat{\mat{A}}$, the estimated sampling weights $\hat{w}_{r,c}$ for each $(r,c) \in [n_r]\times[n_c]$ are given by
\begin{align} \label{eq:w-hat}
  \hat{w}_{r,c} = 1/(1+\exp(-\hat{A}_{r,c})).
\end{align} 
In practice, the numerical experiments described in this paper apply this estimation procedure using the default choices of the parameters $\rho$ and $\nu$ suggested by \citet{gui2023conformalized}.

It is worth remarking that the independent Bernoulli observation model~\eqref{eq:estimation-bernoulli} underlying this maximum-likelihood estimation approach differs from the weighted sampling without replacement model~\eqref{model:sampling_wo_replacement} that we utilize to calibrate our simultaneous conformal inferences.
This discrepancy, however, is both useful and unlikely to cause issues, as explained next.
On the one hand, sampling without replacement model is essential to capture the structured nature of our group-level test case $\mat{X}^*$ and of the calibration groups $\mat{X}^{\mathrm{cal}}_{1}, \dots, \mat{X}^{\mathrm{cal}}_{n}$.
On the other hand, sampling without replacement would make the likelihood function in~\eqref{eq:likelihood} intractable, unnecessarily hindering the estimation process.
Fortunately, the interpretation of the sampling weights $w_{r,c}$ remains largely consistent across the models~\eqref{model:sampling_wo_replacement} and~\eqref{eq:estimation-bernoulli}, which justifies the use of the estimated weights $\hat{w}_{r,c}$ in~\eqref{eq:w-hat} to calibrate conformal inferences under the model defined in~\eqref{model:sampling_wo_replacement}.

\FloatBarrier

\section{Additional Empirical Results}\label{app:numerical-results}

\subsection{Additional Experiments with Synthetic Data} \label{app:additional-experiments}

\subsubsection{Heterogeneous Test Sampling Weights} \label{sec:worst-slab-experiments}

This section describes experiments in which the test group $\mat{X}^*$ is sampled according to a model~\eqref{model:weighted_test_generation} with heterogeneous weights $\mat{w}^*$.
As explained in Section~\ref{sec:preliminaries}, the heterogeneous nature of these weights makes it feasible to ensure valid coverage conditional on interesting features of $\mat{X}^*$. Therefore, the following experiments demonstrate the ability of our method to smoothly interpolate between marginal and conditional coverage guarantees, giving practitioners flexibility to up-weight or down-weight different types of test cases, as needed.

The ground-truth matrix $\mat{M} \in \mathbb{R}^{400 \times 400}$ is generated according to the random factorization model described in Section~\ref{sec:exp-synthetic-uniform}, with rank $l = 4$.
We observe $n_{\mathrm{obs}}=48,000$ entries of this matrix, sampled based on the model in~\eqref{model:sampling_wo_replacement} with uniform weights~$\mat{w}$; these are indexed by $\cD_{\mathrm{obs}}$, whose complement is $\cD_{\mathrm{miss}} = [n_r] \times [n_c] \setminus \cD_{\mathrm{obs}}$.
Algorithm~\ref{alg:simultaneous-confidence-region} is then applied as in the previous experiments, using $n = \min \cbrac{2000, \lfloor \xi_{\mathrm{obs}}/2 \rfloor}$ calibration groups and allocating the remaining $n_{\mathrm{train}}=n_{\mathrm{obs}}-Kn$ observations for training.
For the latter purpose, we rely on the usual alternating least square approach, with hypothesized rank $4$, and thus obtain a point estimate $\hat{\mat{M}}$ and its corresponding factor matrices $\hat{\mat{U}} \in \mathbb{R}^{n_r \times l}$ and $\hat{\mat{V}}\in \mathbb{R}^{n_c \times l}$, such that $\hat{\mat{M}} = \hat{\mat{U}}(\hat{V})^{\top}$.

The weights $\mat{w}^*$ for $\mat{X}^*$ in~\eqref{model:weighted_test_generation} are based on an {\em oracle} procedure that leverages perfect knowledge of $\mat{M}$ and $\hat{\mat{M}}$ to construct a sampling process that over-represents portions of the matrix for which the point estimate is less accurate.
This process is controlled by a parameter $\delta \in (0,1]$, which determines the heterogeneity of $\mat{w}^*$.
In the special case of $\delta=1$, the test weights become $w_{r,c}^* = 1$ for all matrix entries, recovering the experimental setup considered earlier in Section~\ref{sec:exp-synthetic-uniform}. By contrast, smaller values of $\delta$ tend to increasingly over-sample portions of the matrix for which the point estimate $\hat{\mat{M}}$ is less accurate.
We refer to Appendix~\ref{app:worst-slab-experiments-details} for details about this construction of the test sampling weights, which gives rise to an interesting and particularly challenging experimental setting in which attaining high coverage is intrinsically difficult.

To highlight the importance of correctly accounting for the heterogeneous nature of the test sampling weights $\mat{w}^*$, in these experiments we compare the performance of joint confidence regions obtained with two alternative approaches.
The first approach consists of applying Algorithm~\ref{alg:simultaneous-confidence-region} based on the correct values of the data-generating weights $\mat{w}$ and $\mat{w}^*$.
The second approach consists of applying Algorithm~\ref{alg:simultaneous-confidence-region} based on the correct values of the data-generating weights $\mat{w}$ but incorrectly specified weights $w^*_{r,c} = 1$ for all $r \in [n_r], c \in [n_c]$.
In both cases, the nominal significance level is $\alpha = 10\%$, and the methods are evaluated based on 100 random test groups sampled from $\cD_{\mathrm{miss}} \setminus \cD_{\mathrm{wse}}$, according to the model in~\eqref{model:weighted_test_generation} with the weights $\mat{w}^*$ defined in Equation~\eqref{eq:wse-test-weights} within Appendix~\ref{app:worst-slab-experiments-details}.
All results are averaged over 300 independent experiments.

Figure~\ref{fig:exp_conditional_vary_k} compares the performances of the two aforementioned implementations of our method as a function of $K$, for different values of the parameter $\delta$.
The results show that our method applied with the correct weights $\mat{w}^*$ always achieves the desired 90\% simultaneous coverage, as predicted by the theory.
By contrast, using mis-specified uniform test sampling weights $\mat{w}^*$ leads to lower coverage than expected, especially for lower values of the parameter $\delta$.
Figure~\ref{fig:exp_conditional_vary_delta} provides an alternative but qualitatively consistent view of these findings, varying the parameter $\delta$ separately for different values of the group size $K$.

\begin{figure}[!htb]
    \centering
    \includegraphics[width=0.8\linewidth]{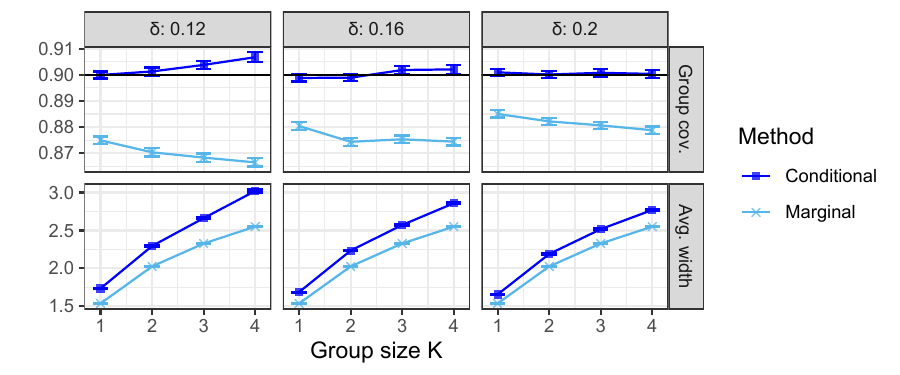}
    \caption{Performance of SCMC on simulated data with heterogeneous test weights.
Two alternative implementations of SCMC are compared: one based on the correct heterogeneous test sampling weights, and one based on mis-specified homogeneous weights.
      The results are shown as a function of the group size $K$, for different values of the parameter $\delta$ which controls the heterogeneity of the sampling weights.
      The performance is measured in terms of the empirical simultaneous coverage (top) and the average width of the confidence regions (bottom). The nominal coverage level is 90\%.
    }
    \label{fig:exp_conditional_vary_k}
\end{figure}

\begin{figure}[!htb]
    \centering
    \includegraphics[width=0.8\linewidth]{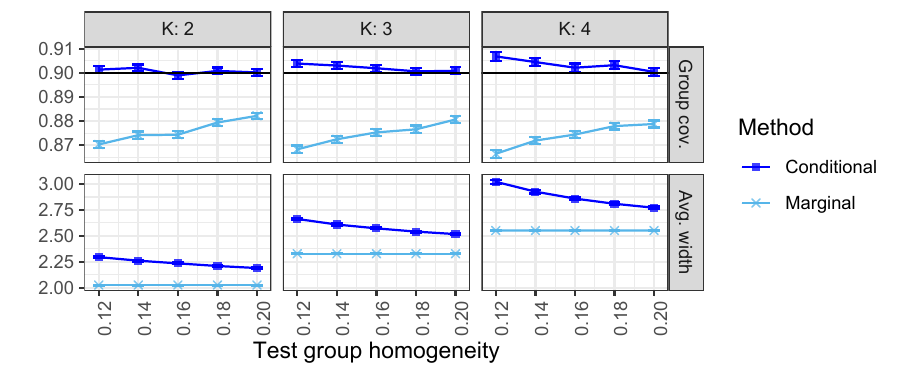}
    \caption{Performance of SCMC on simulated data with heterogeneous test weights.
      The results are shown as a function of the heterogeneity parameter $\delta$, for different group sizes $K$.
      Other details are as in Figure~\ref{fig:exp_conditional_vary_k}.
    }
    \label{fig:exp_conditional_vary_delta}
\end{figure}

It is interesting to note from Figures~\ref{fig:exp_conditional_vary_k} and~\ref{fig:exp_conditional_vary_delta} that our method is sometimes slightly over-conservative when applied with highly heterogeneous test sampling weights $\mat{w}^*$ (corresponding to small values of the parameter $\delta$).
This phenomenon is due to the unavoidable challenge of constructing valid confidence regions in the presence of strong distribution shifts, and it can be understood more precisely as follows.
Smaller values of $\delta$ result in a stronger distribution shift between the observed data in $\mathcal{D}_{\mathrm{obs}}$ and $\mat{X}^*$, increasing the likelihood that the weighted empirical quantile $\hat{\tau}_{\alpha, K}$ defined in~\eqref{eq:weighted-quantile} might become infinite, leading to trivially wide confidence regions.
In those (relatively rare) cases in which $\hat{\tau}_{\alpha, K}$ diverges, to avoid numerical issues we simply set $\hat{\tau}_{\alpha,K}$ equal to $S_{(n)}$, the highest calibration conformity score.
Fortunately, as shown explicitly in Figure~\ref{fig:exp_conditional_vary_k_full}, this issue is not very common (it is observed in fewer than 2.5\% of the cases), which explains why our method appears to be only slightly over-conservative in Figures~\ref{fig:exp_conditional_vary_k} and~\ref{fig:exp_conditional_vary_delta}.

\begin{figure}[!htb]
    \centering
    \includegraphics[width=0.8\linewidth]{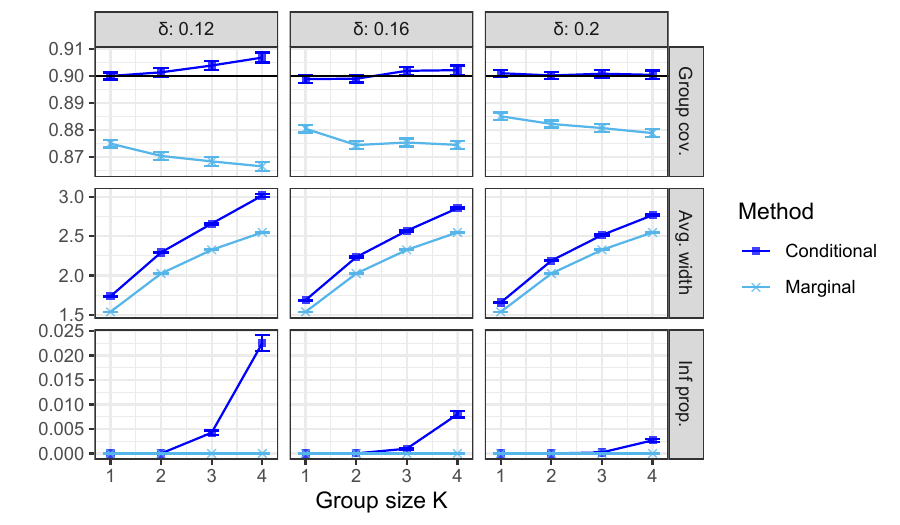}
    \caption{Performance of SCMC on simulated data with heterogeneous sampling weights for the test group.
      The bottom panel reports the empirical probability that our method produces trivially wide confidence regions for a random test group.
      Other details are as in Figure~\ref{fig:exp_conditional_vary_k}.
      }
    \label{fig:exp_conditional_vary_k_full}
\end{figure}

% \begin{figure}[!htb]
%     \centering
%     \begin{subfigure}[b]{0.8\linewidth}
%         \includegraphics[width=\linewidth]{figures/exp_conditional_vary_delta_full.pdf}
%         \caption{Performance of simultaneous conformal inference (conditional) and benchmark (marginal) on simulated data, as a function of the proportion $\delta$ in~\eqref{eq:worst-slice-error}, for different group sizes. The performance is measured in terms of the empirical coverage of the groups (Group Cov.), the average interval size (Size) and the occurrence proportion of the trivial confidence regions (Inf Prop.). The horizontal line corresponds to the nominal $90\%$ coverage level.}
%         \label{fig:exp_conditional_vary_delta_full}
%     \end{subfigure}
%     \hfill
%     \begin{subfigure}[b]{0.8\linewidth}
%         \includegraphics[width=\linewidth]{figures/exp_conditional_vary_k_full.pdf}
%         \caption{Performance of simultaneous conformal inference (conditional) and benchmark (marginal) on simulated data, as a function of the group size $K$, for different proportion $\delta$ used for estimating the worst-estimation slab in Equation~\eqref{eq:worst-slice-error}. All other details are the same as in Figure~\ref{fig:exp_conditional_vary_delta_full}.}
%         \label{fig:exp_conditional_vary_k_full}
%     \end{subfigure}
% \end{figure}

\FloatBarrier

\subsubsection{Additional Details for Section~\ref{sec:worst-slab-experiments}} \label{app:worst-slab-experiments-details}

The sampling weights $\mat{w}^*$ for $\mat{X}^*$ utilized in the experiments of Section~\ref{sec:worst-slab-experiments} are defined based on the following {\em oracle} procedure, which leverages perfect knowledge of $\mat{M}$ and $\hat{\mat{M}}$ to construct a sampling process that over-represents portions of the matrix for which the point estimate is less accurate.
This gives rise a particularly challenging experimental setting.
For each entry $(r,c) \in [n_r]\times[n_c]$, define the {\em latent feature} vector $\mat{y}_{r,c} \coloneqq (\hat{\mat{U}}_{r\circ}, \hat{\mat{V}}_{c\circ}) \in \mathbb{R}^{2l}$, where $\hat{\mat{U}}_{r\circ}$ and $\hat{\mat{V}}_{c\circ}$ are the $r$-th row of $\hat{\mat{U}}$ and the $c$-th row of $\hat{\mat{V}}_{c\circ}$, respectively.
Let also $\cD_{\mathrm{wse}} \subset [n_r]\times[n_c]$ denote a subset containing $25\%$ of the matrix indices in $\cD_{\mathrm{miss}}$, chosen uniformly at random.

The values of $\mat{M}$ and $\hat{\mat{M}}$ indexed by $\cD_{\mathrm{wse}}$ are utilized by the oracle to construct $\mat{w}^*$ with an approach inspired by  \citet{cauchois2020knowing} and \citet{romano2020classification}.
For any fixed $\delta \in (0,1]$, define the {\em worst-slab estimation error},
\begin{align}\label{eq:worst-slice-error}
    \mathrm{WSE}(\hat{\mat{M}}; \delta, \cD_{\mathrm{wse}}) = \sup\limits_{\mat{v}\in \mathbb{R}^{2l}, a < b \in \mathbb{R}} \cbrac*{\frac{\sum_{(r,c) \in S_{\mat{v},a,b}}\abs{\hat{M}_{r,c}-M_{r,c}}}{\abs{S_{\mat{v},a,b}}} \mathrm{ s.t. } \frac{\abs{S_{\mat{v},a,b}}}{\abs{\cD_{\mathrm{wse}}}} \geq \delta},
\end{align}
where, for any $\mat{v} \in \mathbb{R}^{2l}$ and $a < b \in \mathbb{R}$, the subset $S_{\mat{v},a,b} \subset \cD_{\mathrm{wse}}$ is defined as
\begin{align}\label{eq:worst-slab}
    S_{\mat{v},a,b} = \cbrac{(r,c) \in \cD_{\mathrm{wse}}: a \leq \mat{v}^{\top}\mat{y}_{r,c} \leq b }.
\end{align}
Intuitively, $S_{\mat{v},a,b}$ is a subset (or {\em slab}) of the matrix entries in $\cD_{\mathrm{wse}}$ characterized by a direction $\mat{v}$ in the latent feature space and two scalar thresholds $a < b$.
Accordingly, $\mathrm{WSE}(\hat{\mat{M}}; \delta, \cD_{\mathrm{wse}})$ is the average absolute residual between $\mat{M}$ and $\hat{\mat{M}}$ evaluated for the entries within $S_{\mat{v},a,b}$, after selecting the worst-case subset $S_{\mat{v},a,b}$ containing at least a fraction $\delta$ of the observations within $\cD_{\mathrm{wse}}$.

In practice, the optimal (worst-case) choice of $\mat{v}$ in~\eqref{eq:worst-slice-error} is approximated by fitting an ordinary least square regression model to predict the absolute residuals $\cbrac{\abs{\hat{M}_{r,c}-M_{r,c}}}_{(r,c) \in \cD_{\mathrm{wse}}}$ as a linear function of the latent features $\cbrac{\mat{y}_{r,c}}_{(r,c) \in \cD_{\mathrm{wse}}}$. Then, the corresponding optimal values of $a^*, b^*$ in \eqref{eq:worst-slice-error} are approximated through a grid search, for a fixed value of the parameter $\delta$.

Finally, the test sampling weights $\mat{w}^* = \cbrac*{w^*_{r,c}}_{(r,c)\in[n_r]\times[n_c]}$ are given by
\begin{align}\label{eq:wse-test-weights}
    w^*_{r,c}=\begin{cases}
\cfrac{\mathrm{normpdf}(\mat{v}^{*\top}\mat{y}_{r,c},a^*,\sigma^2)}{\mathrm{normpdf}(a^*, a^*,\sigma^2)}, & \mat{v}^{*\top}\mat{y}_{r,c} < a^*\\
  1,  & (r,c) \in  S_{\mat{v}^*,a^*,b^*},\\
\cfrac{\mathrm{normpdf}(\mat{v}^{*\top}\mat{y}_{r,c},b^*,\sigma^2)}{\mathrm{normpdf}(b^*, b^*,\sigma^2)}, & \mat{v}^{*\top}\mat{y}_{r,c} > b^*,
\end{cases}
\end{align}
where $\mathrm{normpdf}(\cdot, a, \sigma^2)$ denotes the density function of the Gaussian distribution with mean $a$ and variance $\sigma^2$.
This density function is introduced for smoothing purposes, setting $\sigma = (b^*-a^*)/5$.
These sampling weights allow us to select test groups from indices that predominantly fall within the worst-slab region for which $\hat{\mat{M}}$ estimates $\mat{M}$ least accurately.
Intuitively, attaining valid coverage for this portion of the matrix should be especially challenging.

\FloatBarrier

\subsection{Additional Experiments with MovieLens Data} \label{app:movielens-conditional}

Having demonstrated in Appendix~\ref{app:additional-experiments} that SCMC can achieve approximate conditional validity in synthetic settings, we now extend this analysis to the MovieLens dataset~\citep{movielens100k} by constructing conditional confidence regions for specific movie genres.
This is particularly relevant for streaming platforms, where recommendations are often categorized by genre.
For instance, to ensure better coverage for drama films, we can assign higher test weights ($\mat{w}^*$) to dramas while downweighting unrelated genres.

Appendix~\ref{app:additional-experiments} demonstrates how SCMC can achieve approximate conditional validity using synthetic data. In this section, we extend this concept to the MovieLens dataset~\citep{movielens100k}, illustrating the construction of conditional confidence regions for specific movie genres. This approach is particularly relevant in real-world applications, such as streaming platforms, where recommendations are often organized into categories based on broad genres. For instance, when recommending dramas, we might want to prioritize joint coverage for ratings specifically within the drama genre. To achieve this, we could set $\mat{w}^*$  to larger values for dramas and smaller values for movies with unrelated themes. 

The MovieLens dataset, beyond the partially observed rating matrix, also provides additional movie information, including genre labels. Among the 1000 randomly selected movies in Section~\ref{sec:experiments-data}, there are 19 distinct genres. For each movie $c \in [n_c]$, we are given a binary vector $\mat{v}_{c} = (v_{c,1} \dots, v_{c,19}) \in \{0,1\}^{19}$ where for $g \in [19]$, $v_{c,g}=1$ if movie $c$ is labeled as genre $g$.

In this experiment, we focus on four genres: $\mathcal{G}=\cbrac{\text{Children's, Crime, Drama, Romance}}$. For each $g \in \mathcal{G}$ and $(r,c) \in [n_r]\times[n_c]$, we set $w^*_{r,c} = v_{c,g}+\epsilon$, where $\epsilon>0$ is a constant noise. 
With this construction, $\mat{w}^*$ is a matrix with constant columns, valued 1 for movies of genre $g$ and 0 otherwise, plus the noise term. Fixing $\epsilon = 1e^{-2}$, we randomly sample $50$ test queries from $\mathcal{D}_{\text{miss}}$ according to the model in \eqref{model:weighted_test_generation} with weights $\mat{w}^*$. We use $n = \min \cbrac{2000, \lfloor \xi_{\text{obs}}/2 \rfloor}$ calibration groups and vary the group size $K$ as a control parameter. Figure~\ref{fig:movielens-conditional} compares the performance of SCMC implemented with the correct heterogeneous weights (Conditional) versus incorrect homogeneous weights (Marginal). As in Figure~\ref{fig:movielens-preview}, we compare only the average width of confidence regions, because ground-truth ratings for $\mathcal{D}_{\text{miss}}$ are unknown. Although empirical verification is not possible due to this limitation, Theorem~\ref{thm:coverage-lower} guarantees that SCMC implemented with correct weights $\mat{w}^*$ would achieve the valid nominal coverage. 

Figure~\ref{fig:movielens-conditional} reveals that conditional confidence regions are narrower for Children's and Crime movies but wider for Drama and Romance movies compared to marginal confidence regions. This disparity may be attributed to the inherent predictability of different genres. Drama and Romance ratings might be noisier and harder to predict, causing marginal SCMC to potentially fail in providing conditional coverage for these genres. By using test sampling weights $\mat{w}^*$, conditional SCMC adaptively enlarges confidence regions to ensure valid conditional coverage. Conversely, Children's and Crime movies may be easier to predict, possibly due to a more focused audience. In these cases, conditional SCMC can provide narrower, more informative confidence regions by prioritizing calibration queries from the same target genres.

While our experiment utilized genre labeling to focus on specific movie subsets, practitioners have the freedom to tailor these weights to their particular needs. For instance, $\mat{w}^*$ could be derived from a smooth function of additional user-item covariates, encompassing not only movie characteristics but also user demographic information. This approach allows for more nuanced conditioning to capture finer-grained patterns in user preferences or item attributes. Such versatility in weight construction enables SCMC to adapt to a wide range of recommendation scenarios, from broad genre-based suggestions to highly personalized content curation.

\begin{figure}[!htb]
\centering
\includegraphics[width=0.8\textwidth]{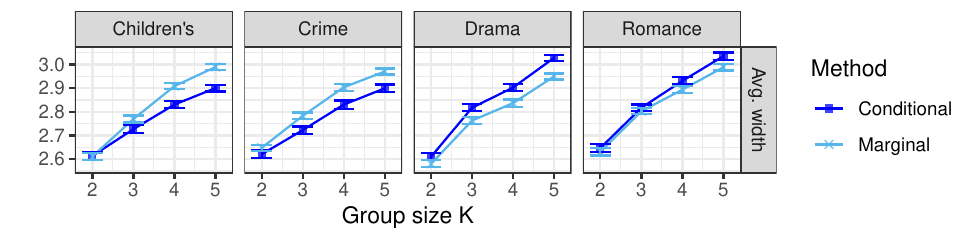}
\caption{Performance of SCMC on MovieLens data with heterogeneous test weights. Two
alternative implementations of SCMC are compared: one based on the correct heterogeneous test sampling weights, and one based on mis-specified homogeneous weights. The
results are shown as a function of the group size K, for four different movie genres. Results are averaged over 300 independent experiments.}
\label{fig:movielens-conditional}
\end{figure}

\subsection{Additional Result with Different Completion Algorithms}\label{app:exp-mc-algorithms}
The validity of SCMC is independent of the matrix completion algorithms used. However, the size of the confidence interval can be affected by the algorithm's effectiveness. Intuitively, if the matrix completion algorithm performs poorly, the corresponding conformal confidence interval would also be large, as it is built upon the estimation residuals. In this section, we conduct additional experiments to investigate the performance of SCMC and benchmarks with two popular alternative low-rank matrix algorithms: nuclear norm minimization (NNM)~\citep{candes_exact_2008} and singular value thresholding (SVT)~\citep{svt_cai_2010}.

NNM solves the convex optimization problem
$\hat{\mat{M}}(\cD_{\mathrm{train}}) = \underset{\mat{Z} \in \mathbb{R}^{n_r \times n_c}}{\arg\min} \norm{\mat{Z}}_*$ subject to $\mathcal{P}_{\cD_{\mathrm{train}}}(\mat{Z}) = \mathcal{P}_{\cD_{\mathrm{train}}}(\mat{M})$,
where $\norm{\cdot}_*$ denotes the nuclear norm and $\mathcal{P}_{\cD_{\mathrm{train}}}(\mat{M})$ is the orthogonal projection of $\mat{M}$ onto the subspace of matrices that vanish outside $\cD_{\mathrm{train}}$. On the other hand, SVT is an efficient iterative algorithm that approximates the optimal convex optimization solution by computing a sequence of estimation matrices and applying soft-thresholding to their singular values. 

Figure~\ref{fig:exp-solver-uniform} shows the results by applying SCMC and benchmarks with the two matrix completion algorithms. The experiment setup follows from that in Section~\ref{sec:exp-synthetic-uniform} where the observed indices are sampled with uniform weights while the ground truth $\mat{M}$ shows a column-wise dependency controlled by parameter $\mu$, here we fix $\mu=15$. Figure~\ref{fig:exp-solver-biased} illustrates the comparison under the experimental conditions detailed in Section~\ref{sec:exp-synthetic-heterogeneous}. In this setting, the ground truth matrix follows a simple low-rank structure, but the observation sampling weights demonstrate column-wise heterogeneity controlled by the parameter $s$, which we fix at $s=0.14$ in Figure~\ref{fig:exp-solver-biased}.

Both plots demonstrate that the behavior and relative performance of SCMC and benchmark methods remain consistent across different matrix completion algorithms. However, more robust matrix completion algorithms lead to more informative confidence intervals for all methods. Among the three completion algorithms, Alternating Least Squares (ALS) exhibits comparable performance to the others while often being the most time-efficient option~\citep{nguyen-2019-mcsurvey}. Given this balance of performance and efficiency, we choose to present the numerical results using ALS in the main text.

\begin{figure}[!htb]
\centering
\includegraphics[width=0.8\textwidth]{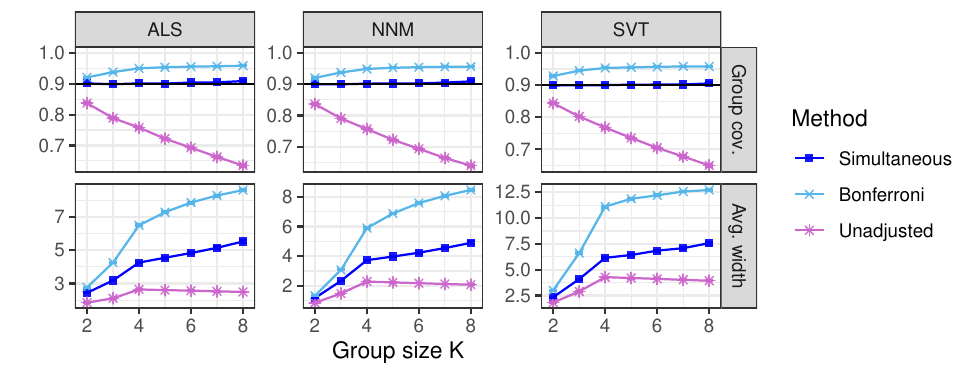}
\caption{Performance of SCMC and two baselines applied with three different matrix completion algorithms, ALS, NNM and SVT. All implementation details follow from Figure~\ref{sec:exp-synthetic-uniform} with parameter $\mu=15$. Results are averaged over 300 independent experiments.}
\label{fig:exp-solver-uniform}
\end{figure}

\begin{figure}[!htb]
\centering
\includegraphics[width=0.8\textwidth]{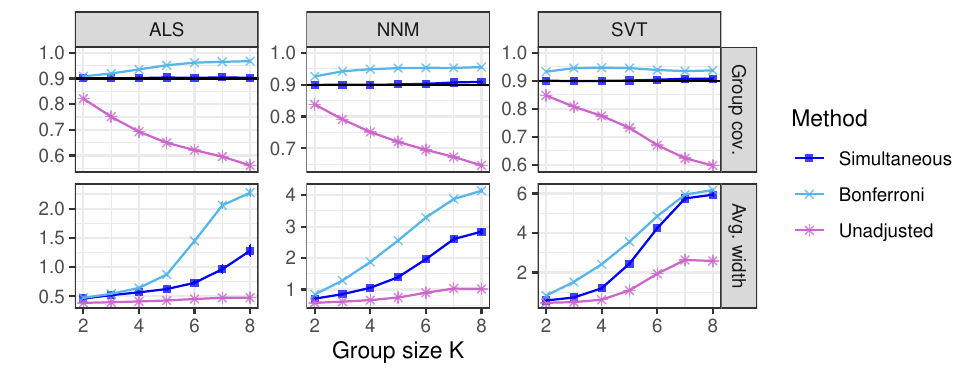}
\caption{Performance of SCMC and two baselines applied with three different matrix completion algorithms, ALS, NNM and SVT. All implementation details follow from Section~\ref{sec:exp-synthetic-heterogeneous} with parameter $s=0.14$. Results are averaged over 300 independent experiments.}
\label{fig:exp-solver-biased}
\end{figure}

\FloatBarrier

\subsection{Additional Results With Estimated Sampling Weights}\label{app:exp-estimated-missingness}
The main text presents synthetic experiments using oracle sampling weights $\mat{w}$ for both SCMC and benchmark methods, highlighting methodological differences without estimation error confounds. However, in practice, $\mat{w}$ is often unknown. We address this aspect by conducting additional experiments with estimated sampling weights, using the method detailed in Appendix~\ref{app:missingness-estimation}, to assess our method's robustness to weight estimation error.

The estimation method assumes $\mat{w}$ follows a model parametrized by a low-rank matrix, then the sampling weights are estimated by maximum likelihood estimation. We choose the guessed rank from $\cbrac{3,5,7}$. Figure~\ref{fig:exp-est-uniform} compares the performance of SCMC with oracle and estimated weights as a function of $K$ under uniform $\mat{w}$. In this simple case, estimation is relatively easy, yielding very close performance between oracle and estimated SCMC. 
Since the weight matrix $\mat{w}$ is rank-1 for uniform $\mat{w}$, 
higher guessed ranks may result in overfitting, which leads to a performance gap for larger group sizes when the error accumulates. 

Figure~\ref{fig:exp-est-biased} shows the comparison under heterogeneous sampling weights. As detailed in Appendix~\ref{app:exp-synthetic-heterogeneous}, heterogeneity is controlled by parameter $s$, with smaller $s$ indicating greater heterogeneity and estimation difficulty. 
The performance gap between oracle and estimated SCMC is largest at $s=0.1$, and the gap decreases as $s$ grows.
Similar to the uniform setting, higher guessed ranks tend to produce more conservative estimates. 

Both Figures~\ref{fig:exp-est-uniform} and~\ref{fig:exp-est-biased} show that the coverage level of SCMC with estimated missingness remains close to nominal levels even under challenging conditions. Therefore, SCMC demonstrates relative robustness to estimation error, particularly when the underlying sampling weights do not exhibit strong heterogeneity.

\begin{figure}[!htb]
\centering
\includegraphics[width=0.8\textwidth]{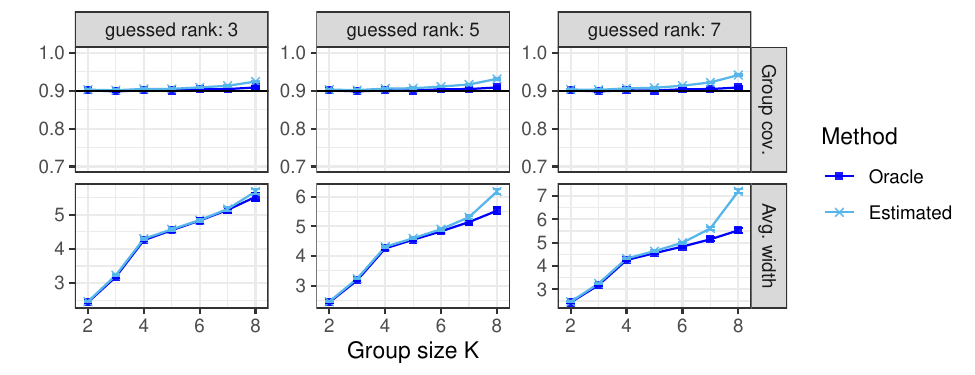}
\caption{SCMC performance comparison using oracle vs. estimated sampling weights with varying guessed ranks. Uniform sampling weights applied; implementation details as in Section~\ref{sec:exp-synthetic-uniform} ($\mu=15$). The result is averaged over 300 independent experiments.}
\label{fig:exp-est-uniform}
\end{figure}

\begin{figure}[!htb]
\centering
\includegraphics[width=0.8\textwidth]{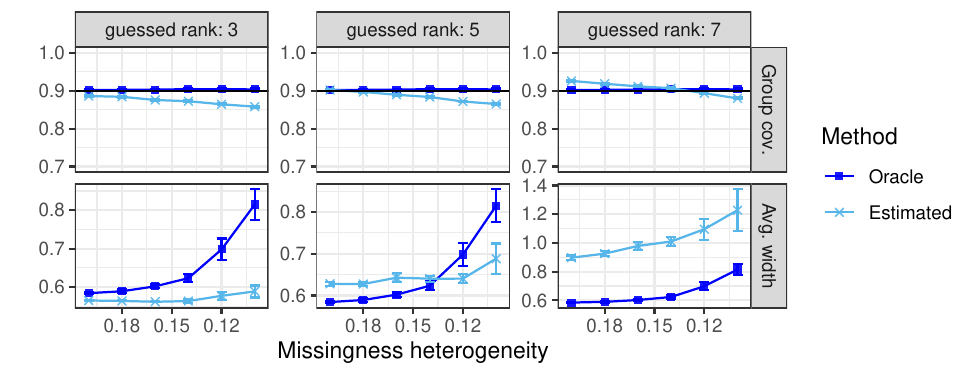}
\caption{SCMC performance with oracle vs. estimated sampling weights under heterogeneous conditions. Heterogeneity parameter $s$  ranges from $0.2$ to $0.1$. Smaller $s$ indicates stronger heterogeneity. Implementation details as in Appendix~\ref{app:exp-synthetic-heterogeneous} ($k=5$). The result is averaged over 300 independent experiments.}
\label{fig:exp-est-biased}
\end{figure}
\FloatBarrier

\subsection{Investigation of the Coverage Upper Bound} \label{sec:upper-bound-numerical}

In this section, we investigate numerically the upper coverage bound for our method established by Theorem~\ref{thm:upper_bound}, which is equal to $1-\alpha + \mathbb{E}[\max_{i \in [n+1]} p_i(\mathbf{X}^{\mathrm{cal}}_{1}, \dots, \mathbf{X}^{\mathrm{cal}}_{n}, \mathbf{X}^*)]$.
Ideally, a small expected value in this equation would guarantee that our conformal inferences are not too conservative.
However, given that it would be unfeasible to evaluate this expected value analytically, we rely on a Monte Carlo numerical study.

We begin by focusing on groups of size $K=2$ and consider for simplicity matrices with an equal number of rows and columns; i.e., $n_r = n_c = 200$.
We simulate the observation process by sampling $n_{\mathrm{obs}} = 0.2 \cdot n_r  n_c$ matrix entries without replacement according to the model defined in~\eqref{model:sampling_wo_replacement}, with  
\begin{align*}
    w_{r,c} = (n_r(c-1)+r)^{s}, \qquad \forall r \in [n_r],  \; c \in [n_c],
\end{align*}
with a scaling parameter $s=2$. 
% Note that this is the same choice of sampling weights utilized in the experiments of Section~\ref{sec:exp-synthetic-heterogeneous}.
For simplicity, the test group $X^*$ is sampled from the model defined in~\eqref{model:weighted_test_generation} using weights $\mat{w}^*=\mat{w}$.
Then, the conformalization weights $p_i$ for all $i \in [n+1]$ are computed by applying Algorithm~\ref{alg:simultaneous-confidence-region}, and varying the number $n$ of calibration groups as a control parameter.
Finally, we estimate $\mathbb{E}[\max_{i \in [n+1]} p_i(\mathbf{X}^{\mathrm{cal}}_{1}, \dots, \mathbf{X}^{\mathrm{cal}}_{n}, \mathbf{X}^*)]$ by taking the empirical average
% of $\max_{i \in [n+1]} p_i(\mathbf{X}^{\mathrm{cal}}_{1}, \dots, \mathbf{X}^{\mathrm{cal}}_{n}, \mathbf{X}^*)$ 
over 10 independent experiments.

\begin{figure}[!htb]
\centering
\includegraphics[width=0.4\textwidth]{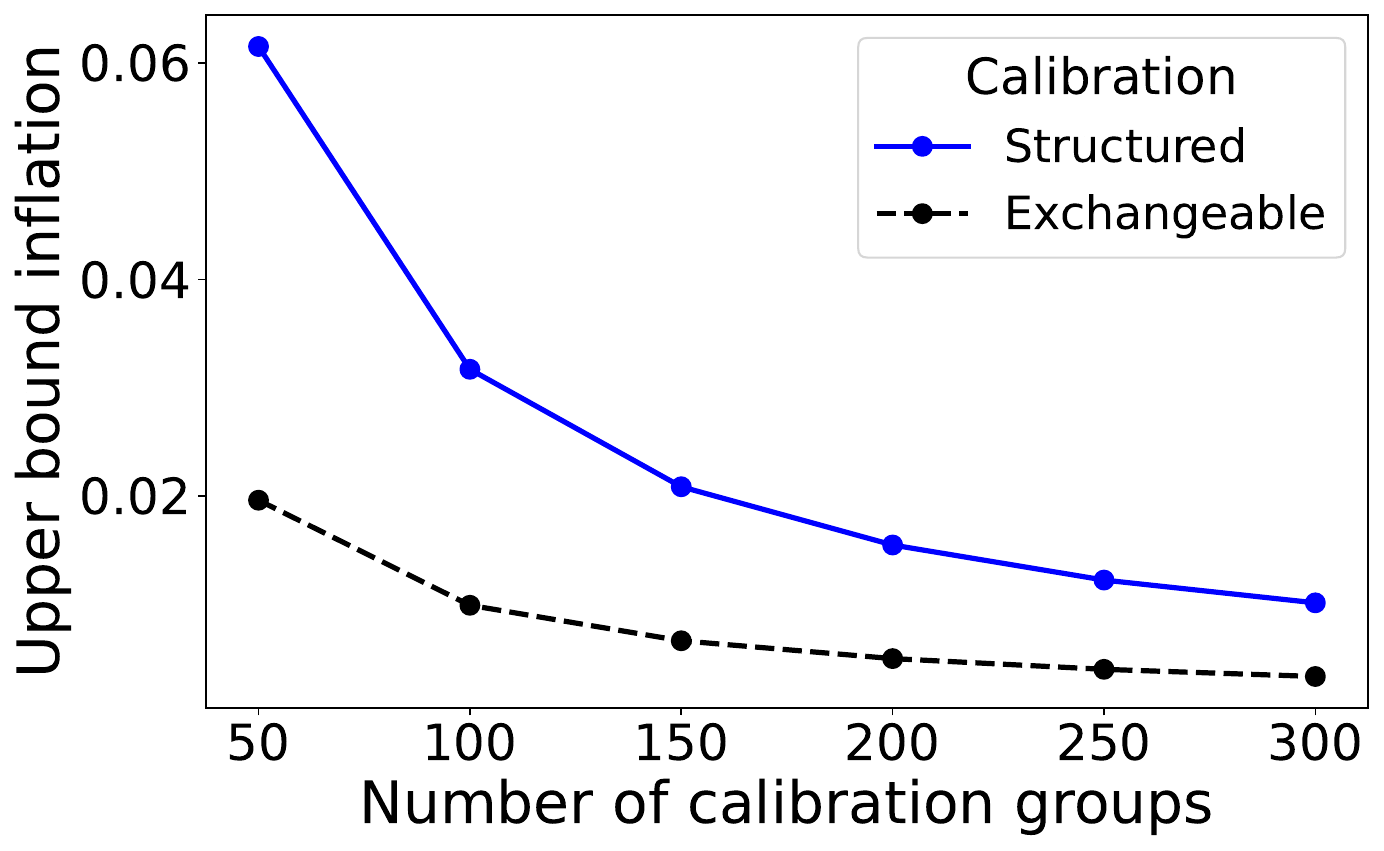} \hspace{1cm}
\includegraphics[width=0.4\textwidth]{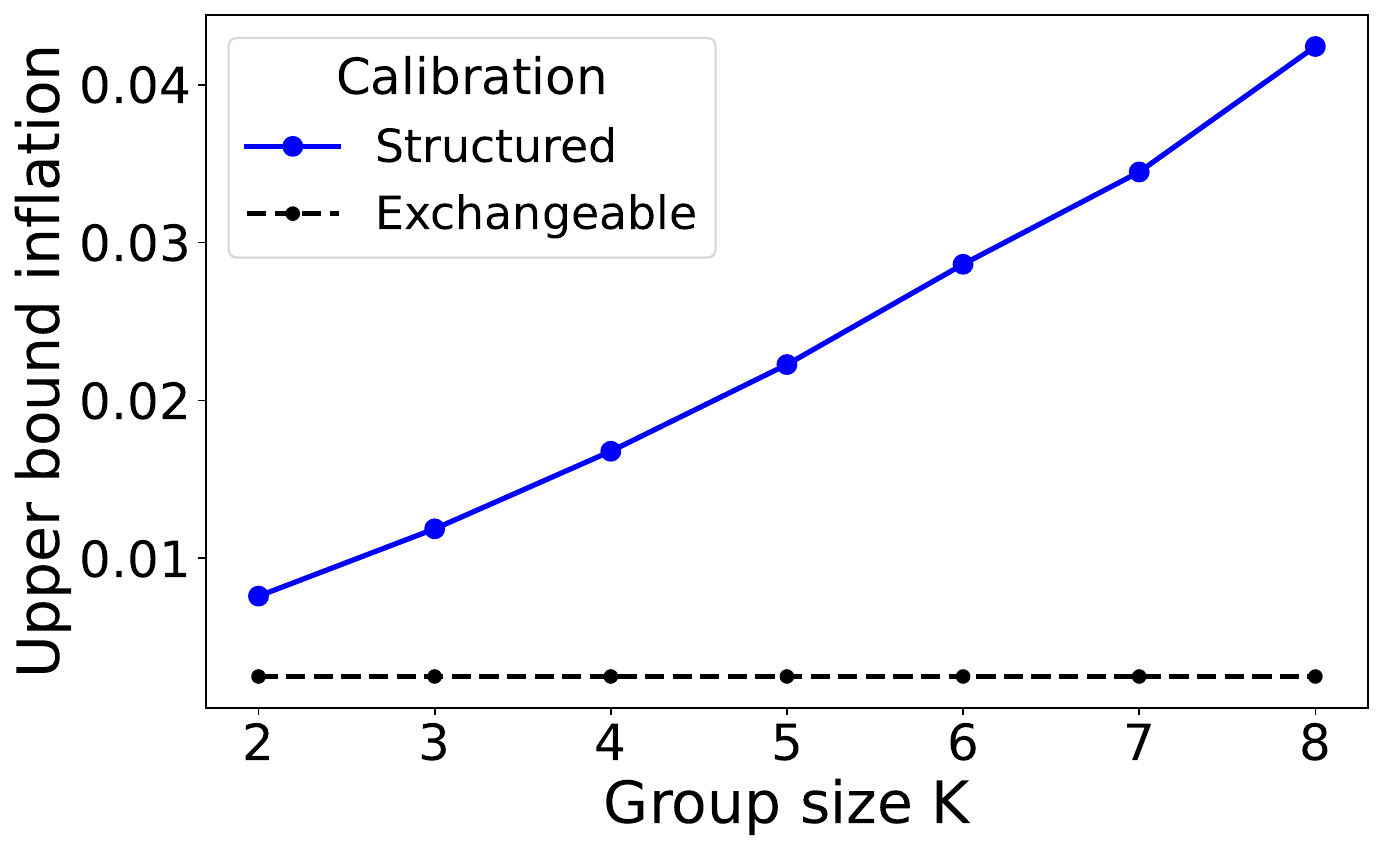}
\caption{Monte Carlo investigation of the theoretical coverage upper bound for the structured confidence regions computed by Algorithm~\ref{alg:simultaneous-confidence-region}.
We plot the expectation of the maximum of the conformalization weights, which is the difference between the upper bound established by Theorem~\ref{thm:upper_bound} and the desired coverage level $1-\alpha$.
The results are compared to the curve $1/n$, which corresponds to the standard upper bound inflation factor corresponding to conformal inferences under exchangeability.
Left: the results are shown as a function of the number of calibration groups $n$, for $K=2$.
Right: the results are shown as a function of the calibration group size $K$, for a structured calibration set of cardinality $n=400$.
}
\label{fig:weighted_upper_bound}
\end{figure}

Figure~\ref{fig:weighted_upper_bound} [left] reports on the results of these experiments as a function of $n$.
The results show that our coverage upper bound approaches $1-\alpha$ roughly at rate $1/n$, as one would generally expect in the case of standard conformal inferences based on exchangeable data \citep{vovk2005algorithmic}.
This is consistent with our empirical observations that Algorithm~\ref{alg:simultaneous-confidence-region} is typically not too conservative in practice.
Figure~\ref{fig:weighted_upper_bound} [right] reports on the results of additional experiments in which the group size $K$ is varied, while keeping the number of calibration groups fixed to $n=400$. The results shows that the coverage upper bound tends to become more conservative as the group size increases, reflecting the intrinsic higher difficulty of producing valid simultaneous conformal inferences for larger groups.

\FloatBarrier

\subsection{Additional Results for Section~\ref{sec:exp-synthetic-uniform}} \label{app:exp-synthetic-uniform}

In the experiments of Section~\ref{sec:exp-synthetic-uniform}, the ground truth matrix $\mat{M}$ is obtained as $\mat{M} = 0.5 \cdot \bar{\mat{M}} + 0.5 \cdot \mat{N}$, where $\bar{\mat{M}}  \in \mathbb{R}^{n_r \times n_c}$ is low-rank while $\mat{N}  \in \mathbb{R}^{n_r \times n_c}$ is a noise matrix exhibiting column-wise dependencies whose strength can be tuned as a control parameter.
Specifically, $\bar{\mat{M}}  \in \mathbb{R}^{n_r \times n_c}$ is given by a \textit{random factorization model} with rank $l=5$; i.e., $\bar{\mat{M}}=\bar{\mat{U}}(\bar{\mat{V}})^\top$, where $\bar{\mat{U}} = (U_{r,c})_{r \in [n_r], c \in [l]}$ and $\bar{\mat{V}} = (V_{r',c'})_{r' \in [n_c], c' \in [l]}$ 
      satisfy $U_{r,c}~\overset{\mathrm{i.i.d.}}{\sim}~\mathcal{N}(0,1)$ and $V_{r',c'}~\overset{\mathrm{i.i.d.}}{\sim}~\mathcal{N}(0,1)$.
Meanwhile, $\mat{N} = 0.1 \cdot \mat{\epsilon} + 0.9 \cdot \mat{1} \tilde{\mat{\epsilon}}^{\top}$,
    where $\mat{1} \in \mathbb{R}^{n_r \times 1}$ is a vector of ones, $\mat{\epsilon} \in \mathbb{R}^{n_r \times n_c}$ has i.i.d.~standard normal components, and $\tilde{\mat{\epsilon}} \in \mathbb{R}^{n_c \times 1}$ is such that, for all $c \in [n_c]$,
        $\tilde{\epsilon}_{c} ~\overset{\mathrm{i.i.d.}}{\sim}~\paren*{1-\gamma} \cdot \mathcal{N}(0,1) + \gamma \cdot \mathcal{N}(\mu,0.1^2)$,
    for suitable parameters $\gamma \in (0,1)$ and $\mu \in \mathbb{R}$.
    Thus, $\mat{1}\tilde{\mat{\epsilon}}^\top \in \mathbb{R}^{n_r \times n_c}$ has constant columns, and a larger $\mu \in \mathbb{R}$ results in stronger column-wise dependencies. Here, we vary $\mu$ as a control parameter, while fixing $\gamma = \alpha/2$.

\begin{figure}[!htb]
    \centering
    \includegraphics[width=0.8\linewidth]{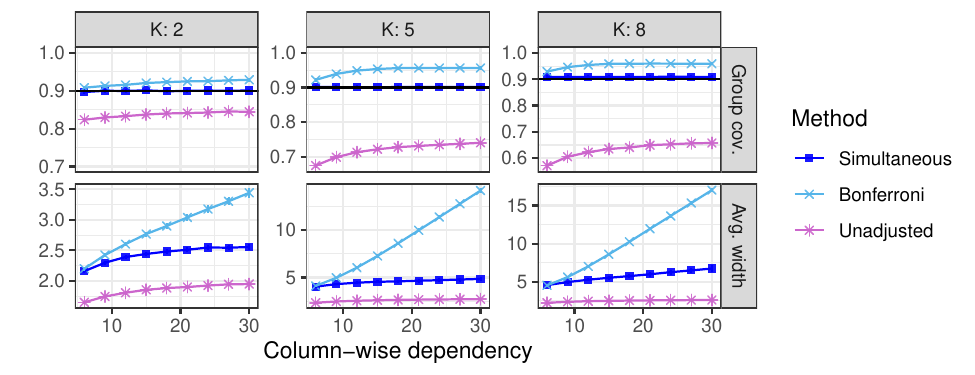}
    \caption{Performance of SCMC and two baselines on simulated data, as a function of the column-wise noise parameter $\mu$ and of the group size $K$. Other details are as in~Figure~\ref{fig:exp_uniform_vary_k}.}
    \label{fig:exp_uniform_vary_mu}
\end{figure}

\FloatBarrier

\subsection{Additional Results for Section~\ref{sec:exp-synthetic-heterogeneous}}
\label{app:exp-synthetic-heterogeneous}
In the experiments of Section~\ref{sec:exp-synthetic-heterogeneous}, we construct the ground truth matrix $\mat{M}$ with a ``signal plus noise'' structure $\mat{M} = \bar{\mat{M}} + \mat{N}$. Here $\bar{\mat{M}}  \in \mathbb{R}^{n_r \times n_c}$  is derived from the random factorization model described in Appendix~\ref{app:exp-synthetic-uniform} with rank $l=8$. The noise matrix $\mat{N}  \in \mathbb{R}^{n_r \times n_c}$ has elements $N_{r,c} ~\overset{\mathrm{i.i.d.}}{\sim}~\mathcal{N}(\mu,0.1^2) \text{ for all } (r,c) \in [N_r]\times[N_c]$. 

The sampling bias $\mat{w}^*$ is constructed such that some columns are less likely to be observed. Specifically, fixing a suitable parameter $\gamma \in [0,1]$, let $I_c~\overset{\mathrm{i.i.d.}}{\sim}~\text{Bernoulli}(\gamma), \text{ for all } c \in [n_c]$  indicate whether a column is sparsely observed. Let $\mathcal{I} = \cbrac{I_c \mid I_c = 1, c \in [n_c]}$ be the index set of sparsely observed columns. Then, for all $(r,c) \in [n_r]\times[n_c]$, define
\begin{align*}
    w_{r,c} = 
    \begin{cases} s, & \text{if $c \in \mathcal{I}$} \\ 1, & \text{otherwise},
    \end{cases}
\end{align*}
where $s\in (0,1)$ controls the heterogeneity of the sampling weights.  A smaller value of $s$ indicates greater heterogeneity, while as $s$ approaches 1, $\mat{w}$ becomes closer to uniform weights.
In this setup, $\mat{w}$  has constant columns, with columns from set $\mathcal{I}$ receiving small sampling weights. We vary $s$ as a control parameter while fixing $\gamma = \alpha/2$. 
\begin{figure}[!htb]
    \centering
    \includegraphics[width=0.8\linewidth]{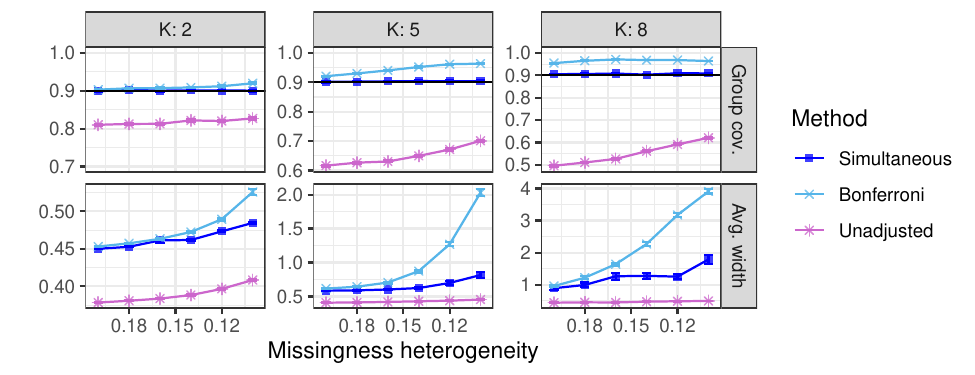}
    \caption{Performance of simultaneous conformal inference and benchmarks on simulated data, as a function of scaling parameter $s$, for different the group size $K$. Other details are as in Figure~\ref{fig:exp_biased_obs_vary_k}.}
    \label{fig:exp_biased_obs_vary_scale}
\end{figure}

\FloatBarrier

\subsection{Additional Results for Section~\ref{sec:experiments-data}} \label{app:additional-exp-movielens}
\begin{figure}[!htb]
    \centering
    \includegraphics[width=0.8\linewidth]{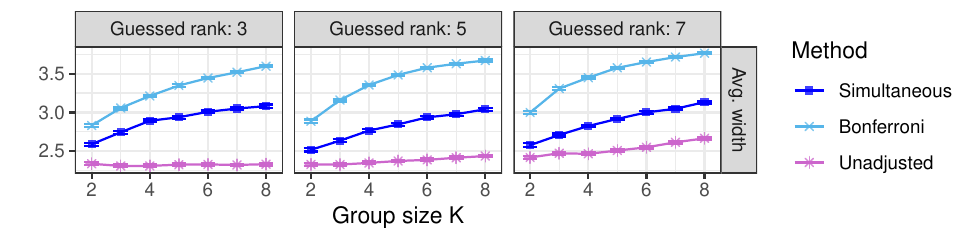}
    \caption{Performance of simultaneous conformal inference and benchmarks on the MovieLens 100K data.
The {\em unadjusted} baseline is designed to achieve valid coverage for a single user at a time, and is generally invalid for a group of $K \geq 2$ users. The output confidence regions are assessed in terms of their average width (lower is better), as a function of the number of users $K$ in the group. Other details are as in Figure~\ref{fig:movielens-preview}.}
    \label{fig:movielens-fullmiss}
\end{figure}

\FloatBarrier

\section{Mathematical Proofs} \label{app:proofs}

\subsection{A General Quantile Inflation Lemma} \label{app:general-quantile}

\begin{proof}[Proof of Lemma~\ref{lem:general_quant}]
The proof follows the same strategy as that of \citet{tibshirani-covariate-shift-2019}.
Let $E_z$ denote the event that $\{Z_1,\ldots,Z_{n+1}\}=\{z_1,\ldots,z_{n+1}\}$, for some possible realization $z = (z_1,\ldots,z_{n+1})$ of $Z_1,\ldots,Z_{n+1}$,
and let $v_i=\calS(z_i, z_{-i} )$ for all $i \in [n+1]$.
By the definition of conditional probability, for each $i \in [n+1]$,
$$
\mathbb{P}\{V_{n+1} = v_i \mid E_z\} = \mathbb{P}\{Z_{n+1} = z_i \mid E_z\}
= \frac{\sum_{\sigma : \sigma(n+1)=i} f(z_{\sigma(1)},\ldots, z_{\sigma(n+1)})}
{\sum_\sigma f(z_{\sigma(1)},\ldots, z_{\sigma(n+1)})}
= p^f_i(z_1,\ldots,z_{n+1}),
$$
where $\sigma$ is a permutation of $[n+1]$.
In other words,
$$
V_{n+1} \mid E_z \sim \sum_{i=1}^{n+1} p^f_i(z_1,\ldots,z_{n+1})\delta_{v_i},
$$
where $\delta_{v_i}$ denotes a point mass at $v_i$.
This  implies that
$$
\mathbb{P}\bigg\{ V_{n+1} \leq \mathrm{Quantile}\bigg(\beta; \, \sum_{i=1}^{n+1}
p^f_i(z_1,\ldots,z_{n+1})\delta_{v_i}\bigg) \mid E_z\bigg\}
\geq \beta,
$$
which is equivalent to
$$
\mathbb{P}\bigg\{ V_{n+1} \leq \mathrm{Quantile}\bigg(\beta; \, \sum_{i=1}^{n+1}
p^f_i(Z_1,\ldots,Z_{n+1}) \delta_{V_i}\bigg) \mid E_z\bigg\}
\geq \beta.
$$
Finally, marginalizing over $E_z$ leads to
$$
\mathbb{P}\bigg\{ V_{n+1} \leq \mathrm{Quantile}\bigg(\beta; \, \sum_{i=1}^{n+1}
p^f_i(Z_1,\ldots,Z_{n+1}) \delta_{V_i}\bigg) \bigg\} \geq \beta.
$$
This is equivalent to the desired result because, by Lemma~\ref{lemma:quantile-infty},
\begin{align*}
  & V_{n+1} \leq \mathrm{Quantile}\bigg(\beta; \, \sum_{i=1}^{n} p^f_i(Z_1,\ldots,Z_{n+1}) \delta_{V_i} + p^f_{n+1}(Z_1,\ldots,Z_{n+1}) \delta_{V_{n+1}}\bigg) \\
  & \qquad \Longleftrightarrow \\
  & V_{n+1} \leq \mathrm{Quantile}\bigg(\beta; \, \sum_{i=1}^{n} p^f_i(Z_1,\ldots,Z_{n+1}) \delta_{V_i} + p^f_{n+1}(Z_1,\ldots,Z_{n+1}) \delta_{\infty}\bigg).
\end{align*}

\end{proof}

\begin{lemma}[also appearing implicitly in \citet{tibshirani-covariate-shift-2019}] \label{lemma:quantile-infty}
Consider $n+1$ random variables $V_1,\ldots,V_{n+1}$ and some weights $p_1, \ldots, p_{n+1}$ such that $p_i >0$ and $\sum_{i=1}^{n+1} p_i=1$.
Then, for any $\beta \in (0,1)$,
\begin{align*}
V_{n+1} \leq Q \big(\beta; \sum_{i=1}^{n+1} p_i \delta_{V_i} \big) \iff V_{n+1} \leq Q \big(\beta; \sum_{i=1}^{n} p_i \delta_{V_i} + p_{n+1}\delta_{\infty}\big).
\end{align*}

\end{lemma}

\begin{proof}[Proof of Lemma~\ref{lemma:quantile-infty}]
This result was previously utilized by \citet{tibshirani-covariate-shift-2019} and a proof is included here for completeness.
It is straightforward to establish one direction of the result, namely
\begin{align*}
V_{n+1} \leq Q \big(\beta; \sum_{i=1}^{n+1} p_i \delta_{V_i} \big) \Longrightarrow V_{n+1} \leq Q \big(\beta; \sum_{i=1}^{n} p_i \delta_{V_i} + p_{n+1}\delta_{\infty}\big),
\end{align*}
because, almost surely, $V_{n+1} \leq \infty$, and hence
\begin{align*}
  Q \big(\beta; \sum_{i=1}^{n} p_i \delta_{V_i} + p_{n+1}\delta_{\infty}\big) \geq Q \big(\beta; \sum_{i=1}^{n+1} p_i \delta_{V_i} \big).
\end{align*}

To prove the other direction, suppose $V_{n+1} > Q \big(\beta; \sum_{i=1}^{n+1} p_i \delta_{V_i} \big)$.
By definition of the quantile function, we can write without loss of generality that $Q \big(\beta; \sum_{i=1}^{n+1} p_i \delta_{V_i} \big)= V_{(j)}$,
where $j \in [n+1]$ is defined such that
\begin{align*}
& p_{(1)}+\ldots+p_{(j)}\geq \beta,
& p_{(1)}+\ldots+p_{(j-1)} < \beta,
\end{align*}
where $p_{(1)} \leq \ldots p_{(n+1)}$ are the order statistics of $p_1,\ldots,p_{n+1}$.
Therefore, $V_{n+1} > V_{(j)}$, and re-assigning $V_{n+1} \to \infty$ does not change $V_{(j)}$. This means that $Q \big(\beta; \sum_{i=1}^{n+1} p_i \delta_{V_i} \big)= Q \big(\beta; \sum_{i=1}^{n} p_i \delta_{V_i} + p_{n+1}\delta_{\infty}\big)$, leading to $V_{n+1} > Q \big(\beta; \sum_{i=1}^{n} p_i \delta_{V_i} + p_{n+1} \delta_{\infty} \big)$.
Thus, we have shown that
\begin{align*}
V_{n+1} > Q \big(\beta; \sum_{i=1}^{n+1} p_i \delta_{V_i} \big) \Longrightarrow V_{n+1} > Q \big(\beta; \sum_{i=1}^{n} p_i \delta_{V_i} + p_{n+1}\delta_{\infty}\big).
\end{align*}

\end{proof}

\begin{proof}[Proof of Lemma~\ref{lem:partial_quant}]
Let $E_z$ denote the event that $\{Z_1,\ldots,Z_{n+1}\}=\{z_1,\ldots,z_{n+1}\}$, for some possible realization $z = (z_1,\ldots,z_{n+1})$ of $Z_1,\ldots,Z_{n+1}$,
and let $v_i=\calS(z_i, z_{-i} )$ for all $i \in [n+1]$.
As in the proof of Lemma~\ref{lem:general_quant}, for each $i \in [n+1]$,
$$
\mathbb{P}\{V_{n+1} = v_i \mid E_z\} = \mathbb{P}\{Z_{n+1} = z_i \mid E_z\}
= \frac{\sum_{\sigma : \sigma(n+1)=i} f(z_{\sigma(1)},\ldots, z_{\sigma(n+1)})}
{\sum_\sigma f(z_{\sigma(1)},\ldots, z_{\sigma(n+1)})}.
$$
Further, because $Z_1,\ldots,Z_{n+1}$ are also leave-one-out exchangeable,
\begin{align*}
\frac{\sum_{\sigma : \sigma(n+1)=i} f(z_{\sigma(1)},\ldots, z_{\sigma(n+1)})}
{\sum_\sigma f(z_{\sigma(1)},\ldots, z_{\sigma(n+1)})} &=
\frac{\sum_{\sigma : \sigma(n+1)=i}
  g(z_{\sigma(1)},\ldots,z_{\sigma(n+1)})
  \cdot  h (z_{\sigma(1)},\ldots,z_{\sigma(n+1)})
  }
{\sum_\sigma
  g(z_{\sigma(1)},\ldots,z_{\sigma(n+1)})
    \cdot  h (z_{\sigma(1)},\ldots,z_{\sigma(n+1)})
  }  \\
&= \frac{\sum_{\sigma : \sigma(n+1)=i}
  g(z_{1},\ldots,z_{n+1})
  \cdot \bar h(z_{-i},z_i)
  }
{\sum_\sigma
  g(z_{1},\ldots,z_{n+1})
    \cdot \bar h(z_{-\sigma(n+1)},z_{\sigma(n+1)})
  } \\
  &= \frac{\sum_{\sigma : \sigma(n+1)=i} \bar h(z_{-i},z_i)
  }
{\sum_{j=1}^{n+1} \sum_{\sigma : \sigma(n+1)=j} \bar h(z_{-j},z_{j})
  } \\
&= \frac{n! \bar h(z_{-i},z_i)
  }
{n! \sum_{j=1}^{n+1} \bar h(z_{-j},z_{j})
  }
         = p_i(z_1,\ldots,z_{n+1}),
\end{align*}
which implies
$
V_{n+1} \mid E_z \sim \sum_{i=1}^{n+1} p_i(z_1,\ldots,z_{n+1})\delta_{v_i}.
$
The rest of the proof then follows with the same approach as the proof of Lemma~\ref{lem:general_quant}.
\end{proof}

\subsection{Conformal Inference with Structured Calibration} \label{app:paired-conformalization}

\begin{proof}[Proof of Proposition~\ref{prop:paired-calibration-partial-exch}]

This result is a direct consequence of Proposition~\ref{prop:jointdist-K}, which characterizes the joint distribution of $(\mat{X}^{\mathrm{cal}}_{1}, \dots, \mat{X}^{\mathrm{cal}}_{n}, \mat{X}^*)$ conditional on $\cD_{\mathrm{prune}}$ and $\cD_{\mathrm{train}}$.
It is easy to see from~\eqref{eq:joint-dist-k} that this distribution is invariant to permutations of $\mat{X}^{\mathrm{cal}}_{1}, \dots, \mat{X}^{\mathrm{cal}}_{n}$.
\end{proof}

\begin{proposition} \label{prop:jointdist-K}
Let $\cD_{\mathrm{obs}}$ and $\cD_{\mathrm{miss}}$ be subsets of observed and missing entries, respectively, sampled according to~\eqref{model:sampling_wo_replacement}.
Consider $\mat{X}^*$ sampled according to~\eqref{model:corner-weighted_test_generation} conditional on $\cD_{\mathrm{obs}}$.
Suppose $\mat{X}^{\mathrm{cal}}_{1}, \dots, \mat{X}^{\mathrm{cal}}_{n}$ are the calibration groups output by Algorithm~\ref{alg:calibration-group}, while $\cD_{\mathrm{prune}}$ and $\cD_{\mathrm{train}}$ are the corresponding training and pruned subsets.
Let $D_{1}$ and $D_{0} \subseteq D_{1}$ denote arbitrary realizations of $\cD_{\mathrm{train}}$ and $\cD_{\mathrm{prune}}$, respectively.
Let $\mat{x}_{1}, \mat{x}_{2}, \ldots, \mat{x}_{n}, \mat{x}_{n+1}$ be any sequence of $n+1$ $K$-groups involving elements of $\cD_{\mathrm{obs}}$ such that
\begin{align*}
\P{\mat{X}^{\mathrm{cal}}_{1} = \mat{x}_{1}, \mat{X}^{\mathrm{cal}}_{2} = \mat{x}_{2}, \ldots, \mat{X}^{\mathrm{cal}}_{n} = \mat{x}_{n}, \mat{X}^* = \mat{x}_{n+1} \mid \cD_{\mathrm{prune}} = D_{0}, \cD_{\mathrm{train}} = D_{1} } > 0.
\end{align*}
Define also $D_2 = \cup_{i\in[n], k\in[K]} \cbrac{x^k_{i}}$, the unordered collection of matrix entries indexed by the groups $\mat{x}_{1}, \ldots, \mat{x}_{n}$, noting that $\cD_{\mathrm{obs}} = D_1 \cup D_2$. Recall from model~\eqref{model:corner-weighted_test_generation}, $\bar{\cD}_{\mathrm{miss}} = \cbrac{(r,c) \in \cD_{\mathrm{miss}} : n^c_{\mathrm{miss}} \geq K}$ is the pruned missing set after removing columns with fewer than $K$ entries,
where $n^c_{\mathrm{miss}} = \abs{\cbrac{(r',c') \in \cD_{\mathrm{miss}} :  c' = c}}$ is the number of missing entries in column $c$.
Further, define
    \begin{align} \label{eq:def-n-bar}
    \begin{split}
      \bar{\cD}_{\mathrm{miss}}^c & \coloneqq \cbrac{(r',c')\in \bar{\cD}_{\mathrm{miss}} \mid c'=c}, \qquad \forall c \in [n_c], \\
      \bar{n}_{\mathrm{obs}} & \coloneqq \abs*{\cD_{\mathrm{obs}} \setminus \cD_{\mathrm{prune}}}, \\
      \bar{n}^c_{\mathrm{obs}} & \coloneqq \abs{\cbrac{(r',c') \in \cD_{\mathrm{obs}} \setminus \cD_{\mathrm{prune}}: c' = c}}, \qquad \forall c \in [n_c], \\
      N_n^c & \coloneqq |\{i \in [n]: c = x_{i,1,2}=x_{i,2,2}=\ldots=x_{i,K,2}\}|, \qquad \forall c \in [n_c].
    \end{split}
    \end{align}
Intuitively, $\bar{\cD}_{\mathrm{miss}}^c$ represents the pruned missing indices in column $c$, $\bar{n}^c_{\mathrm{obs}}$ is the number of indices in $\cD_{\mathrm{obs}} \setminus \cD_{\mathrm{prune}}$ corresponding to entries in column $c$, while $N_n^c$ is the number of calibration groups in column $c$.
Note that all quantities in~\eqref{eq:def-n-bar} are uniquely determined by $D_0$ and $D_2$.
Then,
\begin{align}\label{eq:joint-dist-k}
\begin{split}
    & \mathbb{P}\paren*{\mat{X}^{\mathrm{cal}}_{1}=\mat{x}_{1}, \dots, \mat{X}^{\mathrm{cal}}_{n}=\mat{x}_{n}, \mat{X}^*=\mat{x}_{n+1} \mid  \cD_{\mathrm{prune}}=D_0, \cD_{\mathrm{train}}=D_1} \\
    & = \left[ \frac{w^*_{x_{n+1,1}}}{\sum_{(r,c) \in \bar{\cD}_{\mathrm{miss}}} w^*_{r,c}} \cdot
      \prod_{k=2}^{K} \frac{w^*_{x_{n+1,k}} }{\sum_{(r,c) \in \bar{\cD}_{\mathrm{miss}}} w^*_{r,c} \I{c = x_{n+1,1,2}} - \sum_{k'=1}^{k-1} w^*_{x_{n+1,k'}}} \right] \\
  & \quad \cdot \mathbb{P}\paren*{\cD_{\mathrm{obs}}= D_1 \cup D_2} \cdot \left[ \frac{1 }{ \mathbb{P}\paren*{ \cD_{\mathrm{prune}}=D_0, \cD_{\mathrm{train}}=D_1}} \cdot \prod_{c \in [n_c]} \frac{1}{ \mbinom{n^c_{\mathrm{obs}}}{ \bar{n}^c_{\mathrm{obs}}}} \right] \\
  & \quad \cdot \left( \prod_{i=1}^{n} \frac{1}{\bar{n}_{\mathrm{obs}}-K(i-1)} \right) \cdot \prod_{c=1}^{n_c} \prod_{j=1}^{N_n^c} \prod_{k=1}^{K-1} \frac{1}{\bar{n}^c_{\mathrm{obs}}-K(j-1)-k}.
\end{split}
\end{align}

\end{proposition}

\begin{proof}[Proof of Proposition~\ref{prop:jointdist-K}]
First, note that
\begin{align}
\begin{split} \label{eq:paired-partial-exch-1}
    & \mathbb{P}\paren*{\mat{X}^{\mathrm{cal}}_{1}=\mat{x}_{1}, \dots, \mat{X}^{\mathrm{cal}}_{n}=\mat{x}_{n}, \mat{X}^*=\mat{x}_{n+1} \mid  \cD_{\mathrm{prune}}=D_0, \cD_{\mathrm{train}}=D_1} \\
    &= \mathbb{P}\paren*{ \mat{X}^*=\mat{x}_{n+1} \mid \cD_{\mathrm{prune}}=D_0, \cD_{\mathrm{train}}=D_1, \mat{X}^{\mathrm{cal}}_{1}=\mat{x}_{1}, \dots, \mat{X}^{\mathrm{cal}}_{n}=\mat{x}_{n}} \\
    & \quad\cdot \mathbb{P}\paren*{\mat{X}^{\mathrm{cal}}_{1}=\mat{x}_{1}, \dots, \mat{X}^{\mathrm{cal}}_{n}=\mat{x}_{n} \mid  \cD_{\mathrm{prune}}=D_0, \cD_{\mathrm{train}}=D_1} \\
    &= \mathbb{P}\paren*{ \mat{X}^*=\mat{x}_{n+1} \mid \cD_{\mathrm{obs}} = D_1 \cup D_2} \\
    & \quad\cdot \mathbb{P}\paren*{\mat{X}^{\mathrm{cal}}_{1}=\mat{x}_{1}, \dots, \mat{X}^{\mathrm{cal}}_{n}=\mat{x}_{n} \mid  \cD_{\mathrm{prune}}=D_0, \cD_{\mathrm{train}}=D_1} \\
    & = \left[ \frac{w^*_{x_{n+1,1}}}{\sum_{(r,c) \in \bar{\cD}_{\mathrm{miss}}} w^*_{r,c}} \cdot
      \prod_{k=2}^{K} \frac{w^*_{x_{n+1,k}} }{\sum_{(r,c) \in \bar{\cD}_{\mathrm{miss}}} w^*_{r,c} \I{c = x_{n+1,1,2}} - \sum_{k'=1}^{k-1} w^*_{x_{n+1,k'}}} \right] \\
    & \quad \cdot  \mathbb{P}\paren*{\mat{X}^{\mathrm{cal}}_{1}=\mat{x}_{1}, \dots, \mat{X}^{\mathrm{cal}}_{n}=\mat{x}_{n} \mid  \cD_{\mathrm{prune}}=D_0, \cD_{\mathrm{train}}=D_1},
  \end{split}
\end{align}
where the first term on the right-hand-side above was written explicitly using the sequential sampling characterization of $\Psi^{\text{col}}$ in~\eqref{model:weighted_test_generation_sequential}.

Next, we focus on the second term on the right-hand-side of~\eqref{eq:paired-partial-exch-1}:
\begin{align} \label{eq:paired-partial-exch-2}
  \begin{split}
    &\mathbb{P}\paren*{\mat{X}^{\mathrm{cal}}_{1}=\mat{x}_{1}, \dots, \mat{X}^{\mathrm{cal}}_{n}=\mat{x}_{n} \mid  \cD_{\mathrm{prune}}=D_0, \cD_{\mathrm{train}}=D_1}\\
    & \quad = \frac{\mathbb{P}\paren*{\mat{X}^{\mathrm{cal}}_{1}=\mat{x}_{1}, \dots, \mat{X}^{\mathrm{cal}}_{n}=\mat{x}_{n}, \cD_{\mathrm{prune}}=D_0, \cD_{\mathrm{train}}=D_1}}{\mathbb{P}\paren*{ \cD_{\mathrm{prune}}=D_0, \cD_{\mathrm{train}}=D_1}}\\
        &= \frac{\mathbb{P}\paren*{\mat{X}^{\mathrm{cal}}_{1}=\mat{x}_{1}, \dots, \mat{X}^{\mathrm{cal}}_{n}=\mat{x}_{n}, \cD_{\mathrm{prune}}=D_0, \cD_{\mathrm{train}}=D_1, \cD_{\mathrm{obs}}= D_1 \cup D_2}}{ \mathbb{P}\paren*{ \cD_{\mathrm{prune}}=D_0, \cD_{\mathrm{train}}=D_1}}\\
         &= \mathbb{P}\paren*{\mat{X}^{\mathrm{cal}}_{1}=\mat{x}_{1}, \dots, \mat{X}^{\mathrm{cal}}_{n}=\mat{x}_{n}, \cD_{\mathrm{train}}=D_1 \mid \cD_{\mathrm{prune}} = D_0,  \cD_{\mathrm{obs}}= D_1 \cup D_2} \\
         & \textcolor{white}{=} \cdot \frac{\mathbb{P}\paren*{ \cD_{\mathrm{prune}} = D_0, \cD_{\mathrm{obs}}= D_1 \cup D_2 }}{ \mathbb{P}\paren*{ \cD_{\mathrm{prune}}=D_0, \cD_{\mathrm{train}}=D_1}} \\
         &= \mathbb{P}\paren*{\mat{X}^{\mathrm{cal}}_{1}=\mat{x}_{1}, \dots, \mat{X}^{\mathrm{cal}}_{n}=\mat{x}_{n} \mid \cD_{\mathrm{prune}} = D_0,  \cD_{\mathrm{obs}}= D_1 \cup D_2 } \\
         & \textcolor{white}{=} \cdot \mathbb{P}\paren*{\cD_{\mathrm{train}}=D_1 \mid \mat{X}^{\mathrm{cal}}_{1}=\mat{x}_{1}, \dots, \mat{X}^{\mathrm{cal}}_{n}=\mat{x}_{n}, \cD_{\mathrm{prune}} = D_0,  \cD_{\mathrm{obs}}= D_1 \cup D_2} \\
         & \textcolor{white}{=} \cdot \frac{\mathbb{P}\paren*{ \cD_{\mathrm{prune}} = D_0, \cD_{\mathrm{obs}}= D_1 \cup D_2 }}{ \mathbb{P}\paren*{ \cD_{\mathrm{prune}}=D_0, \cD_{\mathrm{train}}=D_1}} \\
         &= \mathbb{P}\paren*{\mat{X}^{\mathrm{cal}}_{1}=\mat{x}_{1}, \dots, \mat{X}^{\mathrm{cal}}_{n}=\mat{x}_{n} \mid \cD_{\mathrm{prune}} = D_0,  \cD_{\mathrm{obs}}= D_1 \cup D_2 } \\
         & \textcolor{white}{=} \cdot \frac{\mathbb{P}\paren*{ \cD_{\mathrm{prune}} = D_0, \cD_{\mathrm{obs}}= D_1 \cup D_2}}{ \mathbb{P}\paren*{ \cD_{\mathrm{prune}}=D_0, \cD_{\mathrm{train}}=D_1}},
  \end{split}
\end{align}
    where the last equality above follows from the fact that $\cD_{\mathrm{train}}$ is uniquely determined by $\mat{X}_{1},\ldots,\mat{X}_{n}$,  $\cD_{\mathrm{prune}}$ and $\cD_{\mathrm{obs}}$.

    The first term on the right-hand-side of~\eqref{eq:paired-partial-exch-2} is given by Lemma~\ref{lemma:cal-distribution}:
    \begin{align} \label{eq:paired-partial-exch-3}
    \begin{split}
    &\mathbb{P}\paren*{\mat{X}^{\mathrm{cal}}_{1}=\mat{x}_{1}, \dots, \mat{X}^{\mathrm{cal}}_{n}=\mat{x}_{n} \mid \cD_{\mathrm{prune}}, \cD_{\mathrm{obs}}}\\
    & \quad = \left( \prod_{i=1}^{n} \frac{1}{\bar{n}_{\mathrm{obs}}-K(i-1)} \right) \cdot \prod_{c=1}^{n_c} \prod_{j=1}^{N_n^c} \prod_{k=1}^{K-1} \frac{1}{\bar{n}^c_{\mathrm{obs}}-K(j-1)-k}.
       \end{split}
    \end{align}
    Note that~\eqref{eq:paired-partial-exch-3} implies that, conditional on $\cD_{\mathrm{obs}}$ and $\cD_{\mathrm{prune}}$, the distribution of $\mat{X}^{\mathrm{cal}}_{1}, \ldots, \mat{X}^{\mathrm{cal}}_{n}$ does not depend on the order of these calibration groups.

    Next, we focus on the second term on the right-hand-side of~\eqref{eq:paired-partial-exch-2}, namely
    \begin{align} \label{eq:paired-partial-exch-4}
      \frac{\mathbb{P}\paren*{ \cD_{\mathrm{prune}} = D_0, \cD_{\mathrm{obs}}= D_1 \cup D_2}}{ \mathbb{P}\paren*{ \cD_{\mathrm{prune}}=D_0, \cD_{\mathrm{train}}=D_1}}.
    \end{align}
    The numerator of~\eqref{eq:paired-partial-exch-4} is
    \begin{align}  \label{eq:paired-partial-exch-4-num}
    \begin{split}
        &\mathbb{P}\paren*{ \cD_{\mathrm{prune}} = D_0, \cD_{\mathrm{obs}}= D_1 \cup D_2} \\
        &= \mathbb{P}\paren*{\cD_{\mathrm{obs}}= D_1 \cup D_2} \cdot \mathbb{P}\paren*{ \cD_{\mathrm{prune}} = D_0 \mid \cD_{\mathrm{obs}}= D_1 \cup D_2} \\
        &= \mathbb{P}\paren*{\cD_{\mathrm{obs}}= D_1 \cup D_2} \cdot \prod_{c \in [n_c]} \frac{1}{ \mbinom{n^c_{\mathrm{obs}}}{ m^c}} \\
        &= \mathbb{P}\paren*{\cD_{\mathrm{obs}}= D_1 \cup D_2} \cdot \prod_{c \in [n_c]} \frac{1}{ \mbinom{n^c_{\mathrm{obs}}}{ \bar{n}^c_{\mathrm{obs}}}}.
    \end{split}
    \end{align}
    where $m^c := n^c_{\mathrm{obs}} \mod K$ denotes the remainder of the integer division $n^c_{\mathrm{obs}} /K$ , and $\bar{n}^c_{\mathrm{obs}} = \lfloor n^c_{\mathrm{obs}} / K\rfloor = n^c_{\mathrm{obs}} - m^c$.
    Above, the denominator does not need to be simplified because it only depends on $D_0$ and $D_1$.

    Finally, combining~\eqref{eq:paired-partial-exch-1}, \eqref{eq:paired-partial-exch-2}, \eqref{eq:paired-partial-exch-3}, \eqref{eq:paired-partial-exch-4}, and \eqref{eq:paired-partial-exch-4-num}, we arrive at:
\begin{align*}
    & \mathbb{P}\paren*{\mat{X}^{\mathrm{cal}}_{1}=\mat{x}_{1}, \dots, \mat{X}^{\mathrm{cal}}_{n}=\mat{x}_{n}, \mat{X}^*=\mat{x}_{n+1} \mid  \cD_{\mathrm{prune}}=D_0, \cD_{\mathrm{train}}=D_1} \\
    & = \left[ \frac{w^*_{x_{n+1,1}}}{\sum_{(r,c) \in \bar{\cD}_{\mathrm{miss}}} w^*_{r,c}} \cdot
      \prod_{k=2}^{K} \frac{w^*_{x_{n+1,k}} }{\sum_{(r,c) \in \bar{\cD}_{\mathrm{miss}}} w^*_{r,c} \I{c = x_{n+1,1,2}} - \sum_{k'=1}^{k-1} w^*_{x_{n+1,k'}}} \right] \\
  & \quad \cdot \mathbb{P}\paren*{\cD_{\mathrm{obs}}= D_1 \cup D_2} \cdot \left[ \frac{1 }{ \mathbb{P}\paren*{ \cD_{\mathrm{prune}}=D_0, \cD_{\mathrm{train}}=D_1}} \cdot \prod_{c \in [n_c]} \frac{1}{ \mbinom{n^c_{\mathrm{obs}}}{\bar{n}^c_{\mathrm{obs}}}} \right] \\
  & \quad \cdot \left( \prod_{i=1}^{n} \frac{1}{\bar{n}_{\mathrm{obs}}-K(i-1)} \right) \cdot \prod_{c=1}^{n_c} \prod_{j=1}^{N_n^c} \prod_{k=1}^{K-1} \frac{1}{\bar{n}^c_{\mathrm{obs}}-K(j-1)-k}.
\end{align*}

\end{proof}

\begin{lemma} \label{lemma:cal-distribution}
Under the same setup as in Proposition~\ref{prop:jointdist-K},
  \begin{align} \label{eq:cal_distribution}
\begin{split}
    &\mathbb{P}\paren*{\mat{X}^{\mathrm{cal}}_{1}=\mat{x}_{1}, \dots, \mat{X}^{\mathrm{cal}}_{n}=\mat{x}_{n} \mid \cD_{\mathrm{prune}}, \cD_{\mathrm{obs}}}\\
    & \quad = \left( \prod_{i=1}^{n} \frac{1}{\bar{n}_{\mathrm{obs}}-K(i-1)} \right) \cdot \prod_{c=1}^{n_c} \prod_{j=1}^{N_n^c} \prod_{k=1}^{K-1} \frac{1}{\bar{n}^c_{\mathrm{obs}}-K(j-1)-k}.
  \end{split}
  \end{align}
\end{lemma}

\begin{proof}[Proof of Lemma~\ref{lemma:cal-distribution}]
    We prove this result by induction on the number of calibration groups, $n$.
    For ease of notation, we will denote the column of the $i$-th calibration group as $c_i$, for any $i \in [n]$; that is, $c_{i} = x_{i,k,2}$ for all $k \in [K]$.
    In the base case where $n=2$,
    \begin{align*}
        &\mathbb{P}\paren*{\mat{X}^{\mathrm{cal}}_{1}=\mat{x}_{1}, \mat{X}^{\mathrm{cal}}_{2}=\mat{x}_{2} \mid \cD_{\mathrm{prune}} = D_0, \cD_{\mathrm{obs}}= D_1}\\
        &\quad= \frac{1}{\bar{n}_{\mathrm{obs}}}\cdot\frac{1}{\bar{n}^{c_1}_{\mathrm{obs}}-1} \cdot \ldots \cdot \frac{1}{\bar{n}^{c_1}_{\mathrm{obs}}-K+1}\cdot \frac{1}{\bar{n}_{\mathrm{obs}}-K} \\
      & \quad \quad \cdot \left[ \left( \frac{1}{\bar{n}^{c_2}_{\mathrm{obs}}-1} \cdot \ldots \cdot \frac{1}{\bar{n}^{c_2}_{\mathrm{obs}}-K+1} \right) \mathbbm{1}\cbrac*{c_1 \neq c_2} + \left( \frac{1}{\bar{n}^{c_2}_{\mathrm{obs}}-K-1} \cdot \ldots \cdot \frac{1}{\bar{n}^{c_2}_{\mathrm{obs}}-2K+1} \right) \mathbbm{1}\cbrac*{c_1 = c_2} \right]\\
        & \quad= \left[ \prod_{i=1}^{2} \frac{1}{\bar{n}_{\mathrm{obs}}-K(i-1)} \right] \cdot \prod_{c=1}^{n_c} \prod_{j=1}^{N_n^c} \prod_{k=1}^{K-1} \frac{1}{\bar{n}^c_{\mathrm{obs}}-K(j-1)-k}.
    \end{align*}
    Now, for the induction step, suppose Equation~\eqref{eq:cal_distribution} holds for $n-1$.
    Then,
    \begin{align*}
        &\mathbb{P}\paren*{\mat{X}^{\mathrm{cal}}_{1}=\mat{x}_{1}, \dots, \mat{X}^{\mathrm{cal}}_{n}=\mat{x}_{n} \mid \cD_{\mathrm{prune}} = D_0, \cD_{\mathrm{obs}}= D_1}\\
        & \quad = \mathbb{P}\paren*{\mat{X}^{\mathrm{cal}}_{1}=\mat{x}_{1}, \dots, \mat{X}^{\mathrm{cal}}_{n-1}=\mat{x}_{n-1} \mid \cD_{\mathrm{prune}} = D_0, \cD_{\mathrm{obs}}= D_1} \\
        & \quad \quad \cdot \frac{1}{\bar{n}_{\mathrm{obs}}-K(n-1)} \cdot \frac{1}{{\bar{n}^{c_n}_{\mathrm{obs}}-K (N_n^{c_{n}}-1)-1}} \cdot \ldots \cdot \frac{1}{{\bar{n}^{c_n}_{\mathrm{obs}}-K N^n_{c_{n}} + 1}} \\
      & \quad = \left[ \prod_{i=1}^{n-1} \frac{1}{\bar{n}_{\mathrm{obs}}-K(i-1)} \right] \cdot \prod_{c=1}^{n_c} \prod_{j=1}^{N_{n-1}^c} \prod_{k=1}^{K-1} \frac{1}{\bar{n}^c_{\mathrm{obs}}-K(j-1)-k} \\
        & \quad \quad \cdot \frac{1}{\bar{n}_{\mathrm{obs}}-K(n-1)} \cdot \prod_{k=1}^{K-1} \frac{1}{{\bar{n}^{c_n}_{\mathrm{obs}}-K (N_{n}^{c_{n}}-1)-k}} \\
      & \quad = \left[ \prod_{i=1}^{n} \frac{1}{\bar{n}_{\mathrm{obs}}-K(i-1)} \right] \cdot \prod_{c=1}^{n_c} \prod_{j=1}^{N_{n}^c} \prod_{k=1}^{K-1} \frac{1}{\bar{n}^c_{\mathrm{obs}}-K(j-1)-k}.
    \end{align*}
    where the last equality above follows because $N_{n-1}^c=N_n^c$ for all $c \neq c_{n}$, while $N_{n}^{c_n}=N_{n-1}^{c_n}+1$.
\end{proof}

\subsection{Characterization of the Conformalization Weights} \label{app:proofs-weights}

\begin{proof}[Proof of Lemma~\ref{lemma:conformalization-weights-explicit}]

Note that Lemma~\ref{lemma:conformalization-weights-explicit} assumes the test group $\mat{X}^*$ is sampled from model~\eqref{model:weighted_test_generation}, assuming all columns has much more missing entries than $K$, namely,
\begin{align}\label{eq:more-than-K-missing}
    \abs{\cbrac{(r',c') \in \cD_{\mathrm{miss}}\mid c'=c}} \gg K, \text{for all } c \in [n_c]
\end{align}
We first show the results for $\mat{X}^*$ sampled from a general model~\eqref{model:corner-weighted_test_generation}, without assuming~\eqref{eq:more-than-K-missing}. Then, we demonstrate that the conformalization weights reduce to the form given in Lemma~\ref{lemma:conformalization-weights-explicit} under~\eqref{eq:more-than-K-missing}.

First, recall from Proposition~\ref{prop:jointdist-K} that
\begin{align*}
    &\mathbb{P}\paren*{\mat{X}^{\mathrm{cal}}_{1}=\mat{x}_{1}, \dots, \mat{X}^{\mathrm{cal}}_{n}=\mat{x}_{n}, \mat{X}^{*}=\mat{x}_{n+1} \mid  \cD_{\mathrm{prune}}=D_0, \cD_{\mathrm{train}}=D_1}\\
    &\quad = g(\{\mat{x}_1,\ldots,\mat{x}_{n+1}\}) \cdot \bar{h}(\{\mat{x}_1,\ldots,\mat{x}_{n}\}, \mat{x}_{n+1}),
\end{align*}
for some permutation-invariant function $g$ and
\begin{align} \label{eq:gwpm-weights-k-def-h}
\begin{split}
  \bar{h}(\{\mat{x}_1,\ldots,\mat{x}_{n}\}, \mat{x}_{n+1})
     & = \mathbb{P}\paren*{\cD_{\mathrm{obs}}= D_1 \cup {D}_2} \cdot \left[ \tilde{w}^*_{x_{n+1,1}} \cdot
      \prod_{k=2}^{K} \tilde{w}^*_{x_{n+1,k}} \right]  \\
    &\quad \cdot \left[ \prod_{c \in [n_c]} \frac{1}{ \mbinom{n^c_{\mathrm{obs}}}{\bar{n}^c_{\mathrm{obs}}}} \right]
   \cdot \prod_{c=1}^{n_c} \prod_{j=1}^{N_n^c} \prod_{k=1}^{K-1} \frac{1}{\bar{n}^c_{\mathrm{obs}}-K(j-1)-k},
 \end{split}
\end{align}
with
\begin{align*}
    \tilde{w}^*_{x_{n+1,1}} = \cfrac{w^*_{x_{n+1,1}}} {\sum_{(r,c)\in \bar{D}_{\mathrm{miss}}}w^{*}_{r,c}},
\end{align*}
and, for all $k \in \{2, \dots, K\}.$,
\begin{align*}
\tilde{w}^*_{x_{n+1, k}} = \frac{w^*_{x_{n+1,k}}}{\sum_{(r,c) \in \bar{D}_{\mathrm{miss}}} w^*_{r,c} \I{c = x_{n+1,1,2}} - \sum_{k'=1}^{k-1} w^*_{x_{n+1,k'}}}.
\end{align*}
Therefore, Lemma~\ref{lem:partial_quant} can be applied, with weights proportional to
\begin{align} \label{eq:gwpm-weights-k-tmp-1}
  p_i(\mat{x}_1,\ldots,\mat{x}_{n+1})
  & \propto \bar h( \{ \mat{x}_1,\ldots,\mat{x}_{n+1}\} \setminus \{\mat{x}_{i}\}, \mat{x}_{i}).
\end{align}

In order to compute the right-hand-side of~\eqref{eq:gwpm-weights-k-tmp-1}, one must understand how~\eqref{eq:gwpm-weights-k-def-h} changes when $\mat{x}_{n+1}$ is swapped with $\mat{x}_i$, for any fixed $i \in [n]$.
This can be done easily, one piece at a time.

To begin, it is immediate to see that swapping $\mat{x}_{n+1}$ with $\mat{x}_i$ results in $\mathbb{P}\paren*{\cD_{\mathrm{obs}}= D_1 \cup {D}_2}$ being replaced by $\mathbb{P}_{\mat{w}} \paren*{\cD_{\mathrm{obs}}= D_{\mathrm{obs};i}}$.
Similarly, $\tilde{w}^*_{x_{n+1,1}}, \ldots, \tilde{w}^*_{x_{n+1,K}}$ are replaced by $\tilde{w}^*_{x_{i,1}}, \ldots, \tilde{w}^*_{x_{i,K}}$, defined as
\begin{align*}
    \tilde{w}^*_{x_{i,1}}
  = \frac{w^*_{x_{i,1}}} {\sum_{(r,c)\in \bar{D}_{\mathrm{miss};i}}w^{*}_{r,c}},
\end{align*}
and, for all $k \in \{2,\ldots, K\}$,
\begin{align*}
  \tilde{w}^*_{x_{i,k}} &= \frac{w^*_{x_{n+1,k}}}{\sum_{(r,c) \in \bar{D}_{\mathrm{miss};i}} w^*_{r,c} \I{c = x_{i,1,2}} - \sum_{k'=1}^{k-1} w^*_{x_{i,k'}}}\\
  &=\frac{w^*_{x_{i,k}}}{\sum_{(r,c)\in \bar{D}^{c_{i}}_{\mathrm{miss};i}}  w^*_{r,c}-\sum^{k-1}_{k'=1}w^*_{x_{i;k'}}},
\end{align*}
here, for any $i \in [n+1]$, $c_i$ denotes the column to which $\mat{x}_{i}$ belongs; i.e., $c_i \coloneqq x_{i,k,2}, \forall k \in [K]$,
where $x_{i,k,2}$ is the column of the $k$th entry in $\mat{x}_i$.

To understand the notation in the equations above, recall that  $\bar{D}_{\mathrm{miss}} = \cbrac{(r,c) \in D_{\mathrm{miss}} \mid n^c_{\mathrm{miss}} \geq K}$ is a realization of the pruned missing set $\bar{\cD}_{\mathrm{miss}}$, and $n^c_{\mathrm{miss}} = \abs{\cbrac{(r',c') \in \cD_{\mathrm{miss}} \mid  c' = c}}$ is the number of missing entries in column $c$. In the parallel universe where $\mat{x}_{n+1}$ is swapped with $\mat{x}_i$, the realization of the missing indices is denoted as $D_{\mathrm{miss};i}$, and the realization of the pruned missing set is $\bar{D}_{\mathrm{miss};i} \coloneqq \cbrac{(r,c) \in D_{\mathrm{miss};i} \mid n^c_{\mathrm{miss};i} \geq K}$, where $ n^c_{\mathrm{miss};i} \coloneqq \abs{\cbrac{(r',c') \in \cD_{\mathrm{miss};i}}: c'=c}$. Similarly, $\bar{D}^{c}_{\mathrm{miss};i} \coloneqq \cbrac{(r',c') \in \bar{D}_{\mathrm{miss};i}: c'=c}$ denotes entries belonging to column $c$ in the imaginary pruned missing set. Thus, $\tilde{w}^*_{x_{i,1}}$ and $\tilde{w}^*_{x_{i,k}}$ can be interpreted as normalized sampling weights for the imaginary test group $\mat{x}_i$.

Next, let $n_{\mathrm{obs}}^c$ and $n^c_{\mathrm{obs};i}$ denote the numbers of observations in column $c$ from the sets $D_{\mathrm{obs}}$ and $D_{\mathrm{obs};i}$, respectively.
Define also $\bar{n}^c_{\mathrm{obs}} = \lfloor n^c_{\mathrm{obs}} / K\rfloor$ and $\bar{n}^c_{\mathrm{obs};i} = \lfloor n^c_{\mathrm{obs};i} / K\rfloor$, the corresponding numbers of observations remaining in column $c$ after the random pruning step of Algorithm~\ref{alg:calibration-group}. Let $N^c_{n} \coloneqq \abs{\cbrac*{i \in [n]: c_i =c}}$ denote the number of calibration groups in column $c \in [n_c]$.
Similarly, let $N^c_{n;i}$ denote the corresponding imaginary quantity obtained by swapping the calibration group $\mat{X}_i^{\mathrm{cal}}$ with the test group $\mat{X}^*$; i.e.,
\begin{align*}
    N^c_{n;i} \coloneqq \abs{\cbrac*{j \in [n+1]\setminus\cbrac{i}: c_j =c}}
\end{align*}

Further, swapping $\mat{x}_{n+1}$ with $\mat{x}_i$ results in $n^c_{\mathrm{obs}}$, $\bar{n}^c_{\mathrm{obs}}$, and $N^c_{n}$ being replaced by $n^c_{\mathrm{obs}; i}$, $\bar{n}^c_{\mathrm{obs}; i}$, and $N^c_{n; i}$, respectively.
Therefore,
\begin{align}\label{eq:gwpm-weights-k-complicated}
\begin{split}
  & p_i(\mat{x}_{1}, \dots, \mat{x}_{n}, \mat{x}_{n+1})
    \propto \bar h( \{ \mat{x}_1,\ldots,\mat{x}_{n+1}\} \setminus \{\mat{x}_{i}\}, \mat{x}_{i}) \\
    & \quad \propto \paren*{ \tilde{w}^*_{x_{i,1}} \prod\limits_{k=2}^{K}\tilde{w}^*_{x_{i,k}} } \paren*{ \prod\limits_{c=1}^{n_c} \mbinom{n^c_{\mathrm{obs};i}}{\bar{n}^c_{\mathrm{obs};i}}^{-1} }  \paren*{ \prod\limits_{c=1}^{n_c}\prod\limits_{j=1}^{N^c_{n;i}} \prod\limits_{k=1}^{K-1} \frac{1}{\bar{n}^c_{\mathrm{obs};i}-K(j-1)-k} } \cdot \mathbb{P}_{\mat{w}} \paren*{\cD_{\mathrm{obs}}= D_{\mathrm{obs};i}}.
  \end{split}
\end{align}

Now, we will further simplify the expression in Equation~\eqref{eq:gwpm-weights-k-complicated} to facilitate the practical computation of these weights.

Consider the first term on the right-hand-side of Equation \eqref{eq:gwpm-weights-k-complicated}, namely
$$
\paren*{ \tilde{w}^*_{x_{i,1}} \prod\limits_{k=2}^{K}\tilde{w}^*_{x_{i,k}} }.
$$
This quantity depends on the pruned set of missing indices $\bar{D}_{\mathrm{miss};i}$ and, by definition,
\begin{align*}
    \tilde{w}^*_{x_{i,1}}
  &= \frac{w^*_{x_{i,1}}} {\sum_{(r,c)\in \bar{D}_{\mathrm{miss};i}}w^{*}_{r,c}}
  = \frac{w^*_{x_{i,1}}} {\sum_{(r,c)\in \bar{D}_{\mathrm{miss}}}w^{*}_{r,c}-\sum\limits_{k=1}^{K} \paren*{w^{*}_{x_{n+1,k}}-w^{*}_{x_{i,k}}} + u^*_{x_{i,1}}},
\end{align*}
where
\begin{align*}
    u^*_{x_{i,1}}  =  \I{c_{i} \neq c_{n+1}} \paren*{\I{n^{c_{i}}_{\mathrm{miss}} <K} \paren*{\sum\limits_{(r,c)\in D^{c_{i}}_{\mathrm{miss}}} w^*_{r,c}} -\I{n^{c_{n+1}}_{\mathrm{miss}} <2K} \paren*{\sum\limits_{(r,c)\in D^{c_{n+1}}_{\mathrm{miss}} \setminus \mat{x}_{n+1}} w^*_{r,c}}},
\end{align*}
while, for all $k \in \{2,\ldots, K\}$,
\begin{align*}
  \tilde{w}^*_{x_{i,k}}
  & = \frac{w^*_{x_{i,k}}}{\sum_{(r,c)\in \bar{D}^{c_{i}}_{\mathrm{miss};i}}  w^*_{r,c}-\sum^{k-1}_{k'=1}w^*_{x_{i;k'}}} \\
  & = \frac{w^*_{x_{i,k}}}{\sum\limits_{(r,c)\in \bar{D}^{c_{i}}_{\mathrm{miss}}}  w^*_{r,c} + \sum\limits_{k'=k}^K w^{*}_{x_{i,k'}} -\I{c_i=c_{n+1}}\paren*{\sum\limits_{k'=1}^K w^{*}_{x_{n+1,k'}}} + \I{n^{c_{i}}_{\mathrm{miss}} <K} \paren*{\sum\limits_{(r,c)\in D^{c_{i}}_{\mathrm{miss}}} w^*_{r,c}}} \\
  &= \frac{w^*_{x_{i,k}}}{\sum\limits_{(r,c)\in D^{c_{i}}_{\mathrm{miss}}}  w^*_{r,c} + \sum\limits_{k'=k}^K w^{*}_{x_{i,k'}} -\I{c_i=c_{n+1}}\paren*{\sum\limits_{k'=1}^K w^{*}_{x_{n+1,k'}}}}.
\end{align*}
In the equations above, with a slight abuse of notation, we denoted the set of missing indices in column $c_{n+1}$ excluding those in the group $\mat{x}_{n+1}$ as $D^{c_{n+1}}_{\mathrm{miss}}\setminus \mat{x}_{n+1} \coloneqq D^{c_{n+1}}_{\mathrm{miss}}\setminus \cbrac{x_{n+1,k}}_{k=1}^K$.

Note that under the mild technical assumption~\eqref{eq:more-than-K-missing}, $\bar{D}_{\mathrm{miss}} = D_{\mathrm{miss}}$, and the term $u^*_{x_{i,1}}$ vanishes, leading to the simplified form in Lemma~\ref{lemma:conformalization-weights-explicit} of the main paper.

Next, let us consider the second term on the right-hand-side of Equation \eqref{eq:gwpm-weights-k-complicated}, namely
$$\paren*{ \prod\limits_{c=1}^{n_c} \mbinom{n^c_{\mathrm{obs};i}}{\bar{n}^c_{\mathrm{obs};i}}^{-1} }.$$
This evaluates the probability of observing a particular realization of $\calD_{\mathrm{prune}}$.
Since the pruned indices are chosen uniformly at random, this quantity only depends on the number of observations within each column before and after pruning, namely, $n^c_{\mathrm{obs};i}$ and $\bar{n}^c_{\mathrm{obs};i}$. By definition, we have
\begin{align} \label{eq:relation-n-obs}
    n_{\mathrm{obs};i}^c = \begin{cases}
  n_{\mathrm{obs}}^c -K\I{c_i \neq c_{n+1}}, & c=c_i, \\
  n_{\mathrm{obs}}^c +K\I{c_i \neq c_{n+1}}, & c=c_{n+1}, \\
  n_{\mathrm{obs}}^c, & \mathrm{otherwise}. \\
\end{cases}
\end{align}
The above equivalence from the fact that swapping $\mat{x}_i$ with $\mat{x}_{n+1}$ only affects the observed indices in column $c_i$ and $c_{n+1}$, while all other indices remain the same.
In particular, upon swapping, column $c_i$ will contain $K$ fewer observations, because $\mat{x}_{i}$ is treated as the unobserved test group, and column $c_{n+1}$ will contain $K$ more observations, because $\mat{x}_{n+1}$ is treated as the calibration group.
Similarly,
\begin{align}\label{eq:relation-n-pruned}
    \bar{n}_{\mathrm{obs};i}^c = \begin{cases}
  \bar{n}_{\mathrm{obs}}^c -K\I{c_i \neq c_{n+1}}, & c=c_i, \\
  \bar{n}_{\mathrm{obs}}^c +K\I{c_i \neq c_{n+1}}, & c=c_{n+1}, \\
  \bar{n}_{\mathrm{obs}}^c, & \mathrm{otherwise}. \\
\end{cases}
\end{align}
Combining \eqref{eq:relation-n-obs} and \eqref{eq:relation-n-pruned}, we can rewrite the second term in \eqref{eq:gwpm-weights-k-complicated} as:
\begin{align} \label{eq:rewrite-prune}
\begin{split}
   \prod\limits_{c=1}^{n_c} \mbinom{n^c_{\mathrm{obs};i}}{\bar{n}^c_{\mathrm{obs};i}}^{-1} &=  \brac*{\prod\limits_{c=1}^{n_c} \mbinom{n^c_{\mathrm{obs}}}{\bar{n}^c_{\mathrm{obs}}}^{-1}}
   \cdot \cfrac{\mbinom{n^{c_i}_{\mathrm{obs};i}}{\bar{n}^{c_i}_{\mathrm{obs};i}}^{-1}}{\mbinom{n^{c_i}_{\mathrm{obs}} }{\bar{n}^{c_i}_{\mathrm{obs}}}^{-1}} \cdot  \cfrac{\mbinom{n^{c_{n+1}}_{\mathrm{obs};i}}{\bar{n}^{c_{n+1}}_{\mathrm{obs};i}}^{-1} }{\mbinom{n^{c_{n+1}}_{\mathrm{obs}} }{\bar{n}^{c_{n+1}}_{\mathrm{obs}}}^{-1} } \\
   &=\brac*{\prod\limits_{c=1}^{n_c} \mbinom{n^c_{\mathrm{obs}}}{\bar{n}^c_{\mathrm{obs}}}^{-1}} \cdot \cbrac*{\cfrac{\mbinom{n^{c_i}_{\mathrm{obs}} }{\bar{n}^{c_i}_{\mathrm{obs}}}}{\mbinom{n^{c_i}_{\mathrm{obs}} -K}{\bar{n}^{c_i}_{\mathrm{obs}} -K}} \cdot  \cfrac{\mbinom{n^{c_{n+1}}_{\mathrm{obs}} }{\bar{n}^{c_{n+1}}_{\mathrm{obs}}}}{\mbinom{n^{c_{n+1}}_{\mathrm{obs}} +K}{\bar{n}^{c_{n+1}}_{\mathrm{obs}}+K}}}^{\I{c_i \neq c_{n+1}}}.
\end{split}
\end{align}

Then, the third term on the right-hand-side of \eqref{eq:gwpm-weights-k-complicated} is
$$\paren*{ \prod\limits_{c=1}^{n_c}\prod\limits_{j=1}^{N^c_{n;i}} \prod\limits_{k=1}^{K-1} \frac{1}{\bar{n}^c_{\mathrm{obs};i}-K(j-1)-k} }.$$
This relates to the probability of choosing a specific realization of the calibration groups given the observed indices remaining after random pruning.
To aid the simplification of this term, we point out the following relation between $N^c_{n;i}$ and $N^c_{n}$, i.e., the number of calibration groups from each column in the imagined observed set and original observed set respectively:
\begin{align}\label{eq:relation-batch}
    N^c_{n;i} = \begin{cases}
  N^c_{n} -\I{c_i \neq c_{n+1}}. & c=c_i, \\
  N^c_{n} +\I{c_i \neq c_{n+1}}, & c=c_{n+1}, \\
 N^c_{n} & \mathrm{otherwise}. \\
\end{cases}
\end{align}
Then, using \eqref{eq:relation-n-pruned} and \eqref{eq:relation-batch}, we can write:
\begin{align} \label{eq:rewrite-calib}
\begin{split}
    &\textcolor{white}{=}\prod\limits_{c=1}^{n_c}\prod\limits_{j=1}^{N^c_{n;i}} \prod\limits_{k=1}^{K-1} \frac{1}{\bar{n}^c_{\mathrm{obs};i}-K(j-1)-k}\\
    & \qquad =\brac*{\prod\limits_{c=1}^{n_c}\prod\limits_{j=1}^{N^c_{n}} \prod\limits_{k=1}^{K-1} \frac{1}{\bar{n}^c_{\mathrm{obs}}-K(j-1)-k}} \\
    &\qquad \textcolor{white}{=} \cdot \paren*{
    \cfrac{\prod\limits_{j=1}^{N^{c_i}_{n;i}} \prod\limits_{k=1}^{K-1} \frac{1}{\bar{n}^{c_i}_{\mathrm{obs};i}-K(j-1)-k}}{\prod\limits_{j=1}^{N^{c_i}_{n}} \prod\limits_{k=1}^{K-1} \frac{1}{\bar{n}^{c_i}_{\mathrm{obs}}-K(j-1)-k}} \cdot
    \cfrac{\prod\limits_{j=1}^{N^{c_{n+1}}_{n;i}} \prod\limits_{k=1}^{K-1} \frac{1}{\bar{n}^{c_{n+1}}_{\mathrm{obs};i}-K(j-1)-k}}{\prod\limits_{j=1}^{N^{c_{n+1}}_{n}} \prod\limits_{k=1}^{K-1} \frac{1}{\bar{n}^{c_{n+1}}_{\mathrm{obs}}-K(j-1)-k}}
    }^{\I{c_i \neq c_{n+1}}}\\
    &\qquad =\brac*{\prod\limits_{c=1}^{n_c}\prod\limits_{j=1}^{N^c_{n}} \prod\limits_{k=1}^{K-1} \frac{1}{\bar{n}^c_{\mathrm{obs}}-K(j-1)-k}}
    \cdot \paren*{\prod\limits_{k=1}^{K-1} \frac{\bar{n}^{c_i}_{\mathrm{obs}}-k}{\bar{n}^{c_{n+1}}_{\mathrm{obs}}+K-k}}^{\I{c_i \neq c_{n+1}}}.
\end{split}
\end{align}
Above, the last equality follows from the following simplification based on \eqref{eq:relation-batch} and \eqref{eq:relation-n-pruned}, assuming that $c_i \neq c_{n+1}$:
\begin{align*}
\begin{split}
    &\textcolor{white}{=}
    \cfrac{\prod\limits_{j=1}^{N^{c_i}_{n;i}} \prod\limits_{k=1}^{K-1} \frac{1}{\bar{n}^{c_i}_{\mathrm{obs};i}-K(j-1)-k}}{\prod\limits_{j=1}^{N^{c_i}_{n}} \prod\limits_{k=1}^{K-1} \frac{1}{\bar{n}^{c_i}_{\mathrm{obs}}-K(j-1)-k}} \cdot
    \cfrac{\prod\limits_{j=1}^{N^{c_{n+1}}_{n;i}} \prod\limits_{k=1}^{K-1} \frac{1}{\bar{n}^{c_{n+1}}_{\mathrm{obs};i}-K(j-1)-k}}{\prod\limits_{j=1}^{N^{c_{n+1}}_{n}} \prod\limits_{k=1}^{K-1} \frac{1}{\bar{n}^{c_{n+1}}_{\mathrm{obs}}-K(j-1)-k}}\\
    &=\frac{ \prod\limits_{j=1}^{N^{c_i}_{n}} \prod\limits_{k=1}^{K-1} \bar{n}^{c_i}_{\mathrm{obs}}-K(j-1)-k} { \prod\limits_{j=1}^{N^{c_i}_{n}-1} \prod\limits_{k=1}^{K-1} \paren*{\bar{n}^{c_i}_{\mathrm{obs}}-K}-K(j-1)-k} \cdot
    \frac{\prod\limits_{j=1}^{N^{c_{n+1}}_{n}} \prod\limits_{k=1}^{K-1}\bar{n}^{c_{n+1}}_{\mathrm{obs}}-K(j-1)-k}{\prod\limits_{j=1}^{N^{c_{n+1}}_{n}+1} \prod\limits_{k=1}^{K-1} \paren*{\bar{n}^{c_{n+1}}_{\mathrm{obs}}+K}-K(j-1)-k}\\
    &=\frac{ \prod\limits_{j=1}^{N^{c_i}_{n}} \prod\limits_{k=1}^{K-1} \bar{n}^{c_i}_{\mathrm{obs}}-K(j-1)-k} { \prod\limits_{j=1}^{N^{c_i}_{n}-1} \prod\limits_{k=1}^{K-1} \bar{n}^{c_i}_{\mathrm{obs}}-Kj-k} \cdot
    \frac{\prod\limits_{j=1}^{N^{c_{n+1}}_{n}} \prod\limits_{k=1}^{K-1}\bar{n}^{c_{n+1}}_{\mathrm{obs}}-K(j-1)-k}{\prod\limits_{j=1}^{N^{c_{n+1}}_{n}+1} \prod\limits_{k=1}^{K-1} \bar{n}^{c_{n+1}}_{\mathrm{obs}}-K(j-2)-k}\\
%    &=\frac{ \prod\limits_{j=1}^{N^{c_i}_{n}} \prod\limits_{k=1}^{K-1} \bar{n}^{c_i}_{\mathrm{obs}}-K(j-1)-k} { \prod\limits_{j=2}^{N^{c_i}_{n}} \prod\limits_{k=1}^{K-1} \bar{n}^{c_i}_{\mathrm{obs}}-K(j-1)-k} \cdot
%    \frac{\prod\limits_{j=1}^{N^{c_{n+1}}_{n}} \prod\limits_{k=1}^{K-1}\bar{n}^{c_{n+1}}_{\mathrm{obs}}-K(j-1)-k}{\prod\limits_{j=0}^{N^{c_{n+1}}_{n}} \prod\limits_{k=1}^{K-1} \bar{n}^{c_{n+1}}_{\mathrm{obs}}-K(j-1)-k}\\
    &=\prod\limits_{k=1}^{K-1} \frac{\bar{n}^{c_i}_{\mathrm{obs}}-k}{\bar{n}^{c_{n+1}}_{\mathrm{obs}}+K-k}.
\end{split}
\end{align*}

Finally, combining \eqref{eq:gwpm-weights-k-complicated} with \eqref{eq:rewrite-prune} and \eqref{eq:rewrite-calib}, we arrive at
\begin{align*}
\begin{split}
  & p_i(\mat{x}_{1}, \dots, \mat{x}_{n}, \mat{x}_{n+1})  \\
  & \quad \propto \mathbb{P}_{\mat{w}} \paren*{\cD_{\mathrm{obs}}= D_{\mathrm{obs};i}} \cdot \paren*{ \tilde{w}^*_{x_{i,1}} \prod\limits_{k=2}^{K}\tilde{w}^*_{x_{i,k}} } \\
  & \quad \quad \cdot \brac*{\prod\limits_{c=1}^{n_c} \mbinom{n^c_{\mathrm{obs}}}{\bar{n}^c_{\mathrm{obs}}}^{-1}} \cdot
  \left[ \cfrac{\mbinom{n^{c_i}_{\mathrm{obs}} }{\bar{n}^{c_i}_{\mathrm{obs}}}}{\mbinom{n^{c_i}_{\mathrm{obs}} -K}{\bar{n}^{c_i}_{\mathrm{obs}} -K}} \cdot  \cfrac{\mbinom{n^{c_{n+1}}_{\mathrm{obs}} }{\bar{n}^{c_{n+1}}_{\mathrm{obs}}}}{\mbinom{n^{c_{n+1}}_{\mathrm{obs}} +K}{\bar{n}^{c_{n+1}}_{\mathrm{obs}}+K}} \right]^{\I{c_i \neq c_{n+1}}} \\
  & \quad \quad \cdot \brac*{\prod\limits_{c=1}^{n_c}\prod\limits_{j=1}^{N^c_{n}} \prod\limits_{k=1}^{K-1} \frac{1}{\bar{n}^c_{\mathrm{obs}}-K(j-1)-k}}
    \cdot \paren*{\prod\limits_{k=1}^{K-1} \frac{\bar{n}^{c_i}_{\mathrm{obs}}-k}{\bar{n}^{c_{n+1}}_{\mathrm{obs}}+K-k}}^{\I{c_i \neq c_{n+1}}} \\
  & \quad \propto \mathbb{P}_{\mat{w}} \paren*{\cD_{\mathrm{obs}}= D_{\mathrm{obs};i}} \cdot \paren*{ \tilde{w}^*_{x_{i,1}} \prod\limits_{k=2}^{K}\tilde{w}^*_{x_{i,k}} } \\
  & \quad \quad \cdot \left[ \cfrac{\mbinom{n^{c_i}_{\mathrm{obs}} }{\bar{n}^{c_i}_{\mathrm{obs}}}}{\mbinom{n^{c_i}_{\mathrm{obs}} -K}{\bar{n}^{c_i}_{\mathrm{obs}} -K}} \cdot  \cfrac{\mbinom{n^{c_{n+1}}_{\mathrm{obs}} }{\bar{n}^{c_{n+1}}_{\mathrm{obs}}}}{\mbinom{n^{c_{n+1}}_{\mathrm{obs}} +K}{\bar{n}^{c_{n+1}}_{\mathrm{obs}}+K}} \cdot \prod\limits_{k=1}^{K-1} \frac{\bar{n}^{c_i}_{\mathrm{obs}}-k}{\bar{n}^{c_{n+1}}_{\mathrm{obs}}+K-k} \right]^{\I{c_i \neq c_{n+1}}}.
  \end{split}
\end{align*}

\end{proof}

\subsection{Finite-Sample Coverage Bounds} \label{app:proofs-bounds}

\begin{proof}[Proof of Theorem~\ref{thm:coverage-lower}]
Recall that, by construction,
    \begin{align*}
      \mat{M}_{\mat{X}^*} \in \mathcal{C}(\mat{X}^*, \hat{\tau}_{\alpha,K}, \hat{\mat{M}})
      \Longleftrightarrow
      S^* \leq \hat{\tau}_{\alpha,K} = Q\paren[\Big]{1-\alpha; \sum_{i=1}^n p_i\delta_{S_i} + p_{n+1}\delta_{\infty}}.
    \end{align*}
Therefore, Theorem~\ref{thm:coverage-lower} follows directly by combining Proposition~\ref{prop:paired-calibration-partial-exch}, Lemma~\ref{lem:partial_quant}, and the characterization of the conformalization weights given by Equation~\eqref{eq:gwpm-weights-k}.

\end{proof}

\begin{proof}[Proof of Theorem~\ref{thm:upper_bound}]
Recall that, by construction,
    \begin{align*}
      \mat{M}_{\mat{X}^*} \in \mathcal{C}(\mat{X}^*, \hat{\tau}_{\alpha,K}, \hat{\mat{M}})
      \Longleftrightarrow
      S^* \leq \hat{\tau}_{\alpha,K} = Q\paren[\Big]{1-\alpha; \sum_{i=1}^n p_i\delta_{S_i} + p_{n+1}\delta_{\infty}}.
    \end{align*}
Therefore, applying Lemma~\ref{lemma:quantile-infty}, we see that it suffices to prove
\begin{align}
  \P{ S^* \leq Q\paren[\Big]{1-\alpha; \sum_{i=1}^n p_i\delta_{S_i} + p_{n+1}\delta_{S^*}} } \leq 1-\alpha+\EE{\max_{i \in [n+1]} p_i(\mat{X}^{\mathrm{cal}}_{1}, \dots, \mat{X}^{\mathrm{cal}}_{n}, \mat{X}^*)}.
\end{align}

Let $E_x$ denote the event that $\cbrac{\mat{X}^{\mathrm{cal}}_{1}, \dots, \mat{X}^{\mathrm{cal}}_{n}, \mat{X}^{*}} = \cbrac{\mat{x}_{1}, \dots, \mat{x}_{n+1}}$, $\cD_{\mathrm{drop}}=D_0$, and $\cD_{\mathrm{train}}=D_1$, for some possible realizations $\mat{x} = (\mat{x}_1,\ldots,\mat{x}_{n+1})$ of $\mat{X}^{\mathrm{cal}}_{1}, \dots, \mat{X}^{\mathrm{cal}}_{n}, \mat{X}^{*}$, $D_0$ of $\cD_{\mathrm{drop}}$, and $D_1$ of $\cD_{\mathrm{train}}$.
Let also $\{v_1,\ldots,v_{n+1}\}$ indicate the realization of $\{S_1,\ldots,S_n,S^*\}$ corresponding to the event $E_x$, for all $i \in [n+1]$.
Applying the definition of conditional probability, as in the proof of Lemma~\ref{lem:partial_quant}, we can see that
$ S^* \mid E_x \sim \sum_{i=1}^{n+1} p_i(\mat{x}_{1}, \dots, \mat{x}_{n+1})\delta_{v_i}$,
with the weights $p_i$ given by~\eqref{eq:gwpm-weights-k}.
This implies that
$$
\P{ S^* \leq Q\paren[\Big]{1-\alpha; \sum_{i=1}^{n+1} p_i\delta_{v_i} } \mid E_x }
\leq 1-\alpha+\max_{i \in [n+1]} p_i(\mat{x}_{1}, \dots, \mat{x}_{n+1}),
$$
and further, by taking an expectation with respect to the randomness in $E_x$,
$$
\P{ S^* \leq Q\paren[\Big]{1-\alpha; \sum_{i=1}^{n} p_i\delta_{S_i} + p_{n+1}\delta_{S^*} } }
\leq 1-\alpha + \E{\max_{i \in [n+1]} p_i(\mat{X}^{\mathrm{cal}}_{1}, \dots, \mat{X}^{\mathrm{cal}}_{n}, \mat{X}^*)}.
$$
\end{proof}

\begin{proof}[Proof of~Theorem~\ref{thm:coverage-lower-empirical}]
    % First, we prove that \(\hat{\tau}_{\alpha,K} \geq {\tau}_{\alpha+\Delta, K}\). 
    By definition, we know that
\[
\hat{\tau}_{\alpha,K}^{\mathrm{est}} = Q\left(1-\alpha; \sum_{i=1}^n \widehat{p}_i\delta_{S_i} + \widehat{p}_{n+1}\delta_{\infty}\right)
\]

The total variation distance between the two distributions satisfies:
\[
\mathrm{d}_{\mathrm{TV}}\left(\sum_{i=1}^{n+1} \widehat{p}_i \delta_{S_i}, \sum_{i=1}^{n+1} p_i \delta_{S_i}\right) \leq \frac{1}{2} \sum_{i=1}^{n+1}\left|\widehat{p}_i - p_i\right| = \Delta
\]

Therefore, we have
\begin{align*}
    Q\left(1-\alpha; \sum_{i=1}^n \widehat{p}_i\delta_{S_i} + \widehat{p}_{n+1}\delta_{\infty}\right) &\geq Q\left(1-\alpha; \sum_{i=1}^{n+1} \widehat{p}_i \delta_{S_i}\right) \\
    & \geq Q\left(1-\alpha-\Delta; \sum_{i=1}^{n+1} {p}_i \delta_{S_i}\right)
\end{align*}

Combining the lemmas, we have the following:
\[
\P{ \mat{M}_{\mat{X}^{*}} \in \hat{\mat{C}}^{\mathrm{est}}(\mat{X}^*; \mat{M}_{\mat{X}_{\mathrm{obs}}}, \alpha) \mid \cD_{\mathrm{train}}, \cD_{\mathrm{prune}}} \geq 1-\alpha-\mathbb{E}[\Delta].
\]

This inequality completes the proof of Theorem \ref{thm:coverage-lower-empirical}. We have shown we can obtain the desired coverage level adjusted by the estimation gap.

\end{proof}

\subsection{Efficient Evaluation of the Conformalization Weights} \label{app:proofs-weights-fast}

\begin{proof}[Proof of~Proposition~\ref{prop:integral}]

We begin by focusing on the special case of $i =n+1$.
In that case, Equation~\eqref{eq:prob-obs-approx-i} becomes a special case of the results derived for the multivariate Wallenius’ noncentral hypergeometric distribution \citep{wallenius1963biased, chesson1976non, fog2007wnchypg}. While the original problem addresses biased sampling without replacement from an urn containing colored balls, our model in \eqref{model:sampling_wo_replacement} can be equivalently interpreted as drawing samples without replacement from an urn comprising $n_{r}n_{c}$ balls. Each ball is uniquely labeled with a color represented by $(r,c) \in [n_r] \times [n_c]$, and it is drawn with a probability proportional to $w_{r,c}$; e.g., see Section~\ref{sec:preliminaries}.
Therefore, from Equation (19) in \citet{fog2007wnchypg}:
\begin{align}\label{eq:prob-obs-approx}
\begin{split}
   \mathbb{P}_{\mat{w}}(\calD_{\mathrm{obs}}=D_{\mathrm{obs}}) &= \int_{0}^{1} \Phi(\tau; h) \,d\tau.
\end{split}
\end{align}

Next, we turn to proving Equation~\eqref{eq:prob-obs-approx-i} for a general $i \in [n+1]$.

For any $i \in [n+1]$, imagine an alternative world in which $\mat{x}_{n+1}$ is swapped with $\mat{x}_i$.
Let $ \delta_i:=\sum_{(r,c)\in D_{\mathrm{miss};i}}w_{r,c}$ indicate the cumulative weight of all missing entries analogous to $\delta$ in the aforementioned imaginary world.
It is easy to see that $\delta_i = \delta + d_i$, where $d_i := \sum_{k=1}^K(w_{x_{i,k}}-w_{x_{n+1,k}})$. Therefore, we can express the probability using Equation~\eqref{eq:prob-obs-approx} in the imaginary world:
\begin{align*}
\begin{split}
    \mathbb{P}_{\mat{w}}(\calD_{\mathrm{obs}}=D_{\mathrm{obs};i}) &=  \int_{0}^{1} h\delta_{i}\tau^{h\delta_{i}-1}\prod_{(r,c)\in D_{\mathrm{obs};i}} \paren*{1-\tau^{hw_{r,c}}} \,d\tau \\
    &= \int_{0}^{1} h\paren*{\delta + d_i}\tau^{h\paren*{\delta + d_i}-1}\prod_{(r,c) \in D_{\mathrm{obs}}} \paren*{1-\tau^{hw_{r,c}}} \cdot
    \prod\limits_{k=1}^K \paren*{\frac{1-\tau^{hw_{n+1,k}}}{1-\tau^{hw_{i,k}}}}\,d\tau\\
    &= \int_{0}^{1} h\delta\tau^{h\delta-1} \paren*{\prod_{(r,c) \in D_{\mathrm{obs}}} 1-\tau^{hw_{r,c}}} \cdot \frac{h\paren*{\delta + d_i}\tau^{h\paren*{\delta + d_i}-1}}{h\delta\tau^{h\delta-1}} \cdot
    \paren*{\prod\limits_{k=1}^K \frac{1-\tau^{hw_{n+1,k}}}{1-\tau^{hw_{i,k}}}}\,d\tau\\
%    &= \int_{0}^{1} \Phi(\tau; h) \cdot  \frac{\tau^{hd_i}(\delta + d_i)}{\delta} \cdot
%    \paren*{\prod\limits_{k=1}^K \frac{1-\tau^{hw_{n+1,k}}}{1-\tau^{hw_{i,k}}}}\,d\tau \\
    &= \int_{0}^{1} \Phi(\tau; h) \cdot  \eta_i(\tau; h)\,d\tau.
\end{split}
\end{align*}
% Therefore, it follows from Proposition~\ref{prop:integral} that
% \begin{align}
% % \label{eq:prob-obs-approx-i}
%   & \mathbb{P}_{\mat{w}}(\calD_{\mathrm{obs}}=D_{\mathrm{obs};i})
%     = \int_{0}^{1} \Phi(\tau; h) \cdot \eta_i(\tau) \,d\tau,
% \end{align}
where, for any $\tau \in (0,1)$,
\begin{align}
% \label{eq:def-integral-tau}
  & \eta_i(\tau; h) := \frac{\tau^{hd_i}(\delta + d_i)}{\delta} \cdot \paren*{\prod\limits_{k=1}^K \frac{1-\tau^{hw_{n+1,k}}}{1-\tau^{hw_{i,k}}}}.
\end{align}
\end{proof}

\begin{proof}[Proof of Lemma~\ref{lemma:phi-max}]
Recall that the logarithm of $\Phi(\tau;h)$ takes the form
\begin{align} \label{eq:def-phi}
    \phi(\tau; h) := \log \Phi(\tau; h) = \log(h\delta) + (h\delta-1)\log(\tau) + \sum_{(r,c) \in D_{\mathrm{obs}}} \log(1-\tau^{h w_{r,c}}),
\end{align}
while its first derivative with respect to $\tau$ is:
\begin{align*}
    \phi'(\tau; h) & = \frac{h\delta-1}{\tau}-\sum_{(r,c) \in D_{\mathrm{obs}}}\frac{h w_{r,c}\tau^{h w_{r,c}-1}}{1-\tau^{h w_{r,c}}}.
\end{align*}

Consider the function $ \tau \phi'(\tau; h)$, 
\begin{align*}
    \tau \phi'(\tau; h) & = h\delta-1 -\sum_{(r,c) \in D_{\mathrm{obs}}}\frac{h w_{r,c}\tau^{h w_{r,c}}}{1-\tau^{h w_{r,c}}},
\end{align*}
which is strictly decreasing in $\tau$ for all $\tau \in (0,1)$, because $w_{r,c}>0$ for all $(r,c)$.
Further, if $h > 1/\delta$,
\begin{align*}
& \lim_{\tau \to 0^+} \tau \phi'(\tau; h) = h\delta - 1 > 0,
& \lim_{\tau \to 1^-} \tau \phi'(\tau; h) = -\infty.
\end{align*}
Then, by the intermediate value theorem, $\tau \phi'(\tau; h)$ must have exactly one zero for $\tau \in (0,1)$, as long as $h > 1/\delta$.
In turn, this implies that $\phi'(\tau; h)$ has exactly one zero for $\tau \in (0,1)$, as long as $h > 1/\delta$.
Further, the unique zero of $\phi'(\tau; h)$ on $\tau \in (0,1)$ must be the unique maximum of $\phi(\tau; h)$, because, under $h > 1/\delta$, 
\begin{align*}
& \lim_{\tau \to 0^+} \phi(\tau; h) = -\infty,
& \lim_{\tau \to 1^-} \phi(\tau; h) = -\infty.
\end{align*}
\end{proof}

%\subsection{Proof of Equation~\eqref{eq:relation-obs}}
\subsection{Consistency of the Generalized Laplace Approximation}\label{app:laplace-theory}

\begin{proof}[Proof of Theorem~\ref{thm:laplace-independent}]

The preliminary part of this result is proved in Appendix~\ref{app:scaling-parameter}, where we show that selecting the scaling parameter $h_n$ as the unique root of the function in~\eqref{eq:laplace-hn-root} leads to $\tau_n^* \coloneqq \argmax_{\tau \in [0,1]} \Phi_n(\tau) = \frac{1}{2}$.

Our main objective is to approximate the integral
\begin{align*}
    \int_0^1 f_n(\tau) \Phi_n(\tau) \, d\tau = \int_0^1 f_n(\tau) e^{\phi_n(\tau)} \, d\tau,
\end{align*}
leveraging a suitable extension of the classical Laplace method reviewed in Appendix~\ref{app:laplace-review}.
To this end, we begin by applying a Taylor series expansion around $\tau_n^*$, including Lagrange remainder terms; this leads to:
\begin{align*}
    f_n(\tau) &= f_n\left(\frac{1}{2}\right) + f_n'(\xi_1)\left(\tau - \frac{1}{2}\right), \\
    \phi_n(\tau) &= \phi_n\left(\frac{1}{2}\right) + \phi_n'\left(\frac{1}{2}\right)\left(\tau - \frac{1}{2}\right) + \frac{\phi_n''\left(\frac{1}{2}\right)}{2}\left(\tau - \frac{1}{2}\right)^2 + \frac{\phi_n'''(\xi_2)}{6}\left(\tau - \frac{1}{2}\right)^3,
\end{align*}
for some real numbers $\xi_1, \xi_2 \in [1/2,\tau]$, and
   \begin{align*}
    \phi_n(\tau) &= \log(h_n) + \log(\delta_n) + (h_n\delta_n-1)\log(\tau) + \sum_{i=1}^n x_i \log\left(1-\tau^{h_n w_i}\right) \\
    \phi_n'(\tau) &= \frac{h_n\delta_n-1}{\tau} - \sum_{i=1}^n x_i \frac{h_n w_i \tau^{h_n w_i-1}}{\left(1-\tau^{h_n w_i}\right)} \\
    \phi_n''(\tau) &= -\frac{h_n\delta_n-1}{\tau^2} - \sum_{i=1}^n x_i \frac{h_n w_i \tau^{h_n w_i-2}(\tau^{h_n w_i}+h_n w_i -1)}{\left(1-\tau^{h_n w_i}\right)^2} \\
    \phi_n'''(\tau) &= 2\frac{h_n\delta_n-1}{\tau^3} \\
    & + \sum_{i=1}^n x_i \frac{(h_n w_i) \tau^{( h_n w_i -3)} \left(3 (h_n w_i) ( \tau^{h_n w_i} -1 ) + 2 (\tau^{h_n w_i} -1)^2 + (h_n w_i)^2 ( \tau^{h_n w_i} +1)\right)}{(\tau^{h_n w_i} -1 )^3} .
    \end{align*}

By definition of $h_n$, we know that $\phi_n'\left(\frac{1}{2}\right) = 0$.
Next, we need to establish a suitable bound for $\phi_n''$. This task is complicated by the fact that we do not have an explicit expression for $h_n$.
Fortunately, however, it is possible to obtain sufficiently tight lower and upper bounds for $h_n$.

\begin{lemma} \label{lemma:laplace-independent-bound-hn}
In the setting of Theorem~\ref{thm:laplace-independent}, for any $n>1$,
\begin{align} \label{eq:hn-bound-fs}
  \left(\frac{\frac{2n}{\delta_n}}{2^{\frac{2n}{\delta_n}} - 1} \frac{s_n}{n} + \frac{1}{n} \right) \frac{1}{\delta_n}
  \leq h_n
  \leq \left(1 + \frac{n}{\log{2}}\right)\frac{1}{\delta_n}.
\end{align}
Further, in the limit of $n \to \infty$, it holds almost-surely that
\begin{align} \label{eq:hn-bound}
%    (nJ+1)\frac{1}{\delta_n} \leq h_n \leq \left(1 + \frac{n}{\log{2}}\right)\frac{1}{\delta_n},
\frac{J}{L_2} \leq h_n \leq \frac{1}{\log(2)L_2},
\end{align}
where
\begin{align*}
  & J := L_1 \cdot \frac{2}{L_2 }\cdot \frac{1}{2^{\frac{2}{L_2}}-1},
  & L_1 :=  \mathbb{E}[x],
  && L_2 := \mathbb{E}[(1-x)w].
\end{align*}

\end{lemma}

The bounds on $h_n$ provided by Lemma~\ref{lemma:laplace-independent-bound-hn} in turns allow us to bound $\phi_n''$ away from 0 almost surely for large $n$.
This gives us the necessary ingredients to tackle the approximation of the integral.
To this end, note that
\begin{align*}
    \int_0^1 f_n(\tau) e^{\phi_n(\tau)} \, d\tau &= \int_0^1  \left\{f_n\left(\frac{1}{2}\right) + f_n'(\xi_1(\tau))\left(\tau - \frac{1}{2}\right)\right\} e^{\phi_n\left(\frac{1}{2}\right) + \frac{\phi_n''\left(\frac{1}{2}\right)}{2}\left(\tau - \frac{1}{2}\right)^2 + \frac{\phi_n'''(\xi_2(\tau))}{6}\left(\tau - \frac{1}{2}\right)^3} \, d\tau,
\end{align*}
since $\phi_n'(\frac{1}{2})=0$.
Applying the change of variables
$u = \sqrt{-\phi_n''\left(\frac{1}{2}\right)}(\tau-\frac{1}{2})$
 and defining $K_n=\sqrt{-\phi_n''(1/2)}$, the integral becomes
\begin{align} \label{eq:integral-taylor}
    & \int_0^1 f_n(\tau) e^{\phi_n(\tau)} \, d\tau \\
    & \quad = \frac{e^{\phi_n\left(\frac{1}{2}\right)}}{K_n} \int_{-K_n/2}^{K_n/2} \left(f_n\left(\frac{1}{2}\right) +\frac{u}{K_n} f_n'(\xi_1(u))\right) e^{-\frac{u^2}{2} + \frac{\phi_n'''(\xi_2(u)) u^3}{6 K_n^3}}  \, du \nonumber \\
    & \quad = \frac{e^{\phi_n\left(\frac{1}{2}\right)}}{K_n} \left\{
    f_n\left(\frac{1}{2}\right) \int_{-K_n/2}^{K_n/2} e^{-\frac{u^2}{2}}\, du
    + f_n\left(\frac{1}{2}\right) \int_{-K_n/2}^{K_n/2} e^{-\frac{u^2}{2}}\left(e^{\frac{\phi_n'''(\xi_2(u)) u^3}{6 K_n^3}}-1  \right)\, du \right. \nonumber \\
    &\quad  \qquad \left. + \int_{-K_n/2}^{K_n/2} \frac{u}{K_n} f_n'(\xi_1(u))e^{-\frac{u^2}{2}} \, du
    + \int_{-K_n/2}^{K_n/2} \frac{u}{K_n} f_n'(\xi_1(u)) e^{-\frac{u^2}{2}}\left(e^{\frac{\phi_n'''(\xi_2(u)) u^3}{6 K_n^3}}-1  \right) \, du \right\}. \label{eq:taylor-error}
\end{align}
Note that now $\xi_1$ and $\xi_2$ depend implicitly on $u$ due to the change of variables (they previously depended on $\tau$).
We will now separately analyze each term in~\eqref{eq:integral-taylor}.

The limit of the first integral on the right-hand-side of~\eqref{eq:integral-taylor} can be found by leveraging the following lower bound for $K_n$:
\begin{align*}
    K_n = \sqrt{-\phi_n''\left(\frac{1}{2}\right)} = \sqrt{n}\sqrt{\frac{-\phi_n''\left(\frac{1}{2}\right)}{n}} \geq \sqrt{n}\sqrt{4\frac{h_n\delta_n - 1}{n}} \geq 2\sqrt{n}\sqrt{J},
\end{align*}
which leads to
\begin{align} \label{eq:integral-taylor-1}
\lim_{n \to \infty} \int_{-K_n/2}^{K_n/2} e^{-\frac{u^2}{2}} \, du = \sqrt{2\pi}.
\end{align}

By contrast, the third integral on the right-hand-side of~\eqref{eq:integral-taylor} is asymptotically negligible.
Recall that, by assumption, $f_n$ satisfies $|f_n'(x)| \leq M$ for all $x \in (0,1)$ and $n$; therefore,
\begin{align} \label{eq:integral-taylor-3}
    \left|\int_{-K_n/2}^{K_n/2} \frac{u}{K_n} f_n'(\xi_1(u))e^{-\frac{u^2}{2}} \, du \right| \leq  \frac{M}{K_n}\int_{-K_n/2}^{K_n/2} |u| e^{-\frac{u^2}{2}}\, du \to 0 \quad \text{ as } n \to \infty.
\end{align}

We continue with the analysis of the second and fourth terms on the right-hand-side of~\eqref{eq:taylor-error}, which are more involved.
The remainder of the second-order Taylor expansion, previously denoted as $\phi_n'''(\xi_2(u)) u^3 / (6 K_n^3)$, can be expanded back into an infinite power series, given the smoothness of $\phi_n$ within the interval $(0,1)$. This expansion is expressed as:
\begin{align}
    y_n(u) \coloneqq \frac{\phi_n'''(\xi_2(u)) u^3}{6 K_n^3} = \sum_{j=3}^{\infty} \frac{\phi_n^{(j)}\left(\frac{1}{2}\right)}{j!K_n^j} u^j,
\end{align}
where $\phi_n^{(j)}\left(\frac{1}{2}\right)$ represents the $j$-th derivative of $\phi_n$ evaluated at $\frac{1}{2}$, and this series converge for all $n$ and all $u \in (-K_n/2,K_n/2)$.

Note that the $j$-th derivative of $\phi_n$ at $1/2$ can be written as:
\begin{align} \label{eq:phi-n-0.5}
    \frac{\phi_n^{(j)}\left(\frac{1}{2}\right)}{n} &= \frac{h_n\delta_n - 1}{n} \frac{d^j}{d\tau^j}(\log \tau)\Big|_{\tau=\frac{1}{2}} + \frac{1}{n} \sum_{i=1}^{n} x_i \frac{d^j}{d\tau^j} \log\left(1 - \tau^{h_nw_i}\right)\Big|_{\tau=\frac{1}{2}}.
\end{align}
To control $y_n(u)$, we will show that~\eqref{eq:phi-n-0.5} is bounded for large $n$.
Lemma~\ref{lemma:laplace-independent-bound-hn} tells us that the first term in~\eqref{eq:phi-n-0.5} is bounded by constants almost surely.
For the second term, define:
\begin{align}
    g(h,x,w) \coloneqq x \frac{d^j}{d\tau^j} \log\left(1 - \tau^{hw}\right)\Big|_{\tau=\frac{1}{2}},
\end{align}
which is continuous for $h > 0$ and $w > 0$. Given that the interval $C \coloneqq [J, 1/\log(2)]$ is compact, we can use the maximum value in $h$ to bound the function $g$.
Thus, we obtain:
\begin{align}
    \frac{1}{n} \sum_{i=1}^{n} g(h_n,x_i,w_i) \leq \frac{1}{n} \sum_{i=1}^{n} \max_{h\in C} g(h,x_i,w_i) \quad \text{almost surely as } n \to \infty.
\end{align}
Then, by the strong law of large numbers,
\begin{align}
    \frac{1}{n} \sum_{i=1}^{n} \max_{h\in C} g(h,x_i,w_i) \xrightarrow{\mathrm{a.s.}} \mathbb{E}\left[\max_{h\in C} g(h,x,w)\right],
\end{align}
leading to:
\begin{align}
    \frac{1}{n} \sum_{i=1}^{n} g(h_n,x_i,w_i) \leq  \mathbb{E}\left[\max_{h\in C} g(h,x,w)\right] \quad \text{almost surely.}
\end{align}
Similarly, we can also show:
\begin{align}
    \frac{1}{n} \sum_{i=1}^{n} g(h_n,x_i,w_i) \geq  \mathbb{E}\left[\min_{h\in C} g(h,x,w)\right] \quad \text{almost surely.}
\end{align}
Therefore, the second term in \eqref{eq:phi-n-0.5} is also bounded by constants almost surely.
As a result, the whole expression in~\eqref{eq:phi-n-0.5} is bounded by constants almost surely for large $n$.

Combining the above results, we conclude that, for all $j \geq 3$,
\begin{align} \label{eq:phi-bound-big-O}
    \left| \frac{\phi_n^{(j)}\left(\frac{1}{2}\right)}{j!K_n^j}  \right|
    \leq \left| \frac{n \frac{\phi_n^{(j)}\left(\frac{1}{2}\right)}{n}}{j! K_n^j}  \right|
    \leq \left| \frac{n \frac{\phi_n^{(j)}\left(\frac{1}{2}\right)}{n}}{j! (2\sqrt{J})^j n^{\frac{j}{2}}} \right|
    = \mathcal{O}\left(n^{-\left(\frac{j}{2}-1\right)}\right) \text{ almost surely.}
\end{align}

We can now proceed to analyze the integral involving the exponential of the remainder term $y_n(u)$, which appears in the second term on the right-hand-side of~\eqref{eq:integral-taylor}. Specifically,
\begin{align} \label{eq:integral-power-normal}
  \int_{-\infty}^{\infty} e^{-\frac{u^2}{2}}\left(e^{y_n(u)}-1\right) \, du 
  &= \sqrt{2\pi}\mathbb{E}\left[e^{y_n(Z)}-1\right]
    = \sqrt{2\pi} \sum_{k=1}^{\infty} \frac{ \mathbb{E}\left[ (y_n(Z))^k \right]}{k!} ,
\end{align}
where $Z \sim \mathcal{N}(0,1)$. Below, we show that all terms on the right-hand-side of~\eqref{eq:integral-power-normal} are finite and converge to 0 as $n$ increases.
To this end, let us start from $k=1$, noting that
\begin{align*}
    \mathbb{E}[y_n(Z)] 
  = \mathbb{E}\left[\sum_{j=3}^{\infty} \frac{\phi_n^{(j)}\left(\frac{1}{2}\right)}{j!K_n^j} Z^j\right]
    = \sum_{\substack{j\geq 4, \\ j \text{ is even}}} \frac{\phi_n^{(j)}\left(\frac{1}{2}\right)}{j!K_n^j} \frac{j!}{2^{\frac{j}{2}}(\frac{j}{2})!}
    = \sum_{\substack{j\geq 4, \\ j \text{ is even}}} \frac{\phi_n^{(j)}\left(\frac{1}{2}\right)}{K_n^j 2^{\frac{j}{2}}(\frac{j}{2})!},
\end{align*}
where the simplification arises because all odd moments of a standard normal are zero, and the even moments follow from the moment generating function. Given~\eqref{eq:phi-bound-big-O}, it follows that $\mathbb{E}[y_n(Z)] \to 0$ as $n \to \infty$.
Similarly, it also can be shown that all higher moments of $y_n(Z)$ are finite and converge to 0 as $n$ increases.
Consequently, we conclude that
\begin{align*} 
    \mathbb{E}\left[e^{y_n(Z)}-1\right] \to 0 \quad \text{as } n \to \infty.
\end{align*}

Therefore, the limit of the second term in the Taylor error expansion~\eqref{eq:taylor-error} is
\begin{align} \label{eq:integral-taylor-2}
    \lim_{n \to \infty} \int_{-\frac{K_n}{2}}^{\frac{K_n}{2}} e^{-\frac{u^2}{2}}\left(e^{\frac{\phi_n'''(\xi_2(u)) u^3}{6 K_n^3}}-1\right) \, du &= 0 \text{ almost surely.}
\end{align}
The fourth term in the Taylor error expansion~\eqref{eq:taylor-error} vanishes similarly because $f_n'$ is bounded; i.e.,
\begin{align} \label{eq:integral-taylor-4}
    \lim_{n \to \infty} \int_{-\frac{K_n}{2}}^{\frac{K_n}{2}} \frac{u}{K_n} f_n'(\xi_1(u)) e^{-\frac{u^2}{2}}\left(e^{\frac{\phi_n'''(\xi_2(u)) u^3}{6 K_n^3}}-1\right) \, du &= 0 \text{ almost surely.}
\end{align}

Combining \eqref{eq:integral-taylor} with~\eqref{eq:integral-taylor-1}, \eqref{eq:integral-taylor-3}, \eqref{eq:integral-taylor-2}, and \eqref{eq:integral-taylor-4}, we arrive at the desired result.
% \begin{align}
%     \lim_{n \to \infty} \frac{\int_0^1 f_n(\tau) \Phi_n(\tau) \, d\tau}{f_n\left(\frac{1}{2}\right) \cdot \Phi_n\left(\frac{1}{2}\right) \sqrt{\frac{-2\pi}{\phi_n''\left(\frac{1}{2}\right)}}} = 1 \text{ almost surely.}
% \end{align}

\end{proof}

\begin{proof}[Proof of Lemma~\ref{lemma:laplace-independent-bound-hn}]
It is immediate from~\eqref{eq:laplace-hn-root} that $h_n > 1/\delta_n$.
To obtain an upper bound, recall that $z(h_n) = 0$, for the function $z$ in~\eqref{eq:laplace-hn-root}.
This implies that
\begin{align*}
    \delta_n &= \frac{1}{h_n} + \sum_{i=1}^n \frac{x_i w_{i}}{2^{h_n w_{i}} - 1} \\
    &\leq \frac{1}{h_n} + \sum_{i=1}^n x_i \frac{w_i}{\log(2) h_n w_i} \quad \text{(since } 2^x - 1 > \log(2)x \quad \text{ for all } x > 0) \\
    &= \left(1 + \frac{s_n}{\log(2)}\right)\frac{1}{h_n}
      \leq \left(1 + \frac{n}{\log(2)}\right)\frac{1}{h_n},
\end{align*}
where $s_n = \sum_{i=1}^n x_i$ denotes the number of successful Bernoulli trials.
Therefore,
\begin{align} \label{eq:hn-upper-bound}
  h_n \leq \left(1 + \frac{n}{\log(2)}\right)\frac{1}{\delta_n}.
\end{align}
The upper bound in~\eqref{eq:hn-upper-bound} also allows us to find a tighter lower bound.
Note that $\frac{x}{2^x-1}$ is a decreasing function of $x > 0$, and $w_i < 1$ for all $i$. Therefore, $z(h_n) = 0$ implies
\begin{align*}
  h_n\delta_n - 1
  = \sum_{i=1}^n \frac{x_i h_n w_{i}}{2^{h_n w_{i}} - 1}
  \geq \sum_{i=1}^n \frac{x_i h_n}{2^{h_n} - 1}
  = \frac{h_n }{2^{h_n} - 1} s_n.
\end{align*}
Combining this result with~\eqref{eq:hn-upper-bound}, we obtain the following lower bound:
\begin{align} \label{eq:hn-lower-bound}
    \frac{h_n\delta_n - 1}{n} \geq \frac{h_n }{2^{h_n} - 1} \frac{s_n}{n} \geq \frac{\left(1 + \frac{n}{\log(2)}\right)\frac{1}{\delta_n}}{2^{\left(1 + \frac{n}{\log(2)}\right)\frac{1}{\delta_n}} - 1} \frac{s_n}{n} \geq \frac{\frac{2n}{\delta_n}}{2^{\frac{2n}{\delta_n}} - 1} \frac{s_n}{n}.
\end{align}
This completes the proof of~\eqref{eq:hn-bound-fs}.

To prove the second part, we apply the strong law of large numbers, by which
$ \frac{s_n}{n} = \frac{\sum_{i=1}^n x_i}{n} \xrightarrow{\mathrm{a.s.}} \mathbb{E}[x]$ as $n \to \infty$,
and $\frac{\delta_n}{n} =  \frac{\sum_{i=1}^n (1-x_i)w_i}{n} \xrightarrow{\mathrm{a.s.}} \mathbb{E}[(1-x)w]$ as $n \to \infty$.
Therefore, by the continuous mapping theorem,
\begin{align*}
    \frac{h_n\delta_n - 1}{n} \geq \frac{h_n }{2^{h_n} - 1} \frac{s_n}{n} \geq L_1 \cdot \frac{2}{L_2 }\cdot \frac{1}{2^{\frac{2}{L_2}}-1} := J \quad \text{ almost surely as } n \to \infty ,
\end{align*}
where the expected values $L_1 =  \mathbb{E}[x]$ and $L_2 = \mathbb{E}[(1-x)w]$ do not depend on $n$.
Therefore, in the limit of $n \to \infty$, it holds almost-surely that
\begin{align*}
    (nJ+1)\frac{1}{\delta_n} \leq h_n \leq \left(1 + \frac{n}{\log{2}}\right)\frac{1}{\delta_n},
\end{align*}
Finally, recall that $\delta_n/n \xrightarrow{\mathrm{a.s.}} L_2$ as $n \to \infty$. Consequently, we have, almost surely, that:
    $J/L_2 \leq h_n \leq 1/[\log(2)L_2]$ as $n \to \infty$.
\end{proof}

%%% Local Variables:
%%% mode: latex
%%% TeX-master: "supplement_jasa"
%%% End:

% \newpage
\end{document}